\documentclass[12pt]{amsart}

\pdfoutput=1

\usepackage[text={420pt,660pt},centering]{geometry}

\usepackage{color}
\usepackage{stmaryrd}
\usepackage[mathscr]{euscript}
\usepackage{esint,amssymb}
\usepackage{graphicx}
\usepackage{MnSymbol}
\usepackage{mathtools}
\usepackage[colorlinks=true, pdfstartview=FitV, linkcolor=blue, citecolor=blue, urlcolor=blue,pagebackref=false]{hyperref}
\usepackage{microtype}

\usepackage{bm}
\usepackage{scalerel} 
\usepackage{dsfont}
\usepackage[font={footnotesize}]{caption}

\definecolor{darkgreen}{rgb}{0,0.5,0}
\definecolor{darkblue}{rgb}{0,0,0.7}
\definecolor{darkred}{rgb}{0.9,0.1,0.1}

\newtheorem{proposition}{Proposition}
\newtheorem{theorem}[proposition]{Theorem}
\newtheorem{lemma}[proposition]{Lemma}
\newtheorem{corollary}[proposition]{Corollary}

\theoremstyle{definition}

\newcommand{\cref}[1]{Corollary~\ref{c.#1}}

\numberwithin{equation}{section}
\numberwithin{proposition}{section}

\newcommand{\Z}{\mathbb{Z}}
\newcommand{\N}{\mathbb{N}}
\newcommand{\R}{\mathbb{R}}

\newcommand{\E}{\mathbb{E}}
\renewcommand{\P}{\mathbb{P}}

\newcommand{\Zd}{\mathbb{Z}^d}
\newcommand{\Rd}{{\mathbb{R}^d}}

\renewcommand{\a}{\mathbf{a}}
\newcommand{\h}{\mathbf{h}}
\newcommand{\g}{\mathbf{g}}
\newcommand{\f}{\mathbf{f}}

\renewcommand{\subset}{\subseteq}

\renewcommand{\a}{\mathbf{a}}

\newcommand{\ahom}{{\overline{\mathbf{a}}}}
\newcommand{\fhom}{{\overline{\mathbf{f}}}}
\newcommand{\chom}{{\overline{c}}}
\newcommand{\Id}{\mathrm{Id}}

\renewcommand{\subset}{\subseteq}

\newcommand{\cu}{{\scaleobj{1.2	}{\square}}}

\renewcommand{\L}{\mathcal{L}}
\renewcommand{\fint}{\strokedint}

\DeclareMathOperator{\dist}{dist}

\DeclareMathOperator{\var}{var}
\DeclareMathOperator{\cov}{cov}
\DeclareMathOperator{\diam}{diam}

\DeclareMathOperator{\size}{size}

\renewcommand{\bar}{\overline}

\renewcommand{\tilde}{\widetilde}

\newcommand{\indc}{\mathds{1}}

\DeclareMathOperator{\data}{data}

\renewcommand{\hat}{\widehat}

\newcommand{\B}{\mathcal{B}}

\newcommand{\per}{\mathrm{per}}

\begin{document}

\title[$C^2$ regularity of surface tension for the $\nabla \phi$ interface model]{$C^2$ regularity of the surface tension for \\ the $\nabla \phi$ interface model}

\begin{abstract}
We consider the $\nabla \phi$ interface model with a uniformly convex interaction potential possessing H\"older continuous second derivatives. Combining ideas of Naddaf and Spencer with methods from quantitative homogenization, we show that the surface tension (or free energy) associated to the model is at least $C^{2,\beta}$ for some~$\beta>0$. We also prove a fluctuation-dissipation relation by identifying its Hessian with the covariance matrix characterizing the scaling limit of the model. Finally, we obtain a quantitative rate of convergence for the Hessian of the finite-volume surface tension to that of its infinite-volume limit.

\end{abstract}

\author[S. Armstrong]{Scott Armstrong}
\address[S. Armstrong]{Courant Institute of Mathematical Sciences, New York University, 251 Mercer St., New York, NY 10012}
\email{scotta@cims.nyu.edu}

\author[W. Wu]{Wei Wu}
\address[W. Wu]{Statistics department, University of Warwick, Coventry CV4 7AL, UK}
\email{w.wu.9@warwick.ac.uk}

\date{\today}

\maketitle
\setcounter{tocdepth}{1}
\tableofcontents

\section{Introduction}
\label{s.intro}

In this paper, we study the large-scale behavior of a certain class of gradient lattice models with uniformly convex interactions, sometimes called~$\nabla \phi$--interface models. The interface is modeled by a real-valued random field~$\left\{ \phi(x) \,:\, x\in \Zd\right\}$; we think of the graph of~$\phi$ as modeling a random surface which may represent a surface of separation between two distinct pure phases (for instance, a simplified toy model for the interface of the ferromagnetic Ising model at equilibrium) or the deformation of a crystal.

\smallskip

The law of~$\left\{ \phi(x) \,:\, x\in \Zd\right\}$ is governed by nearest-neighbor interactions which depend only on the differences~$\phi(x) - \phi(y)$ for $|x-y|=1$. To each $\{ \phi(x) \}$ we associate an interaction energy given by the (formal) Hamiltonian
\begin{equation}
H(\phi):= 
\sum_{x,y\in\Zd, \, |x-y|=1}
\mathsf{V}(\phi(x) - \phi(y)),
\end{equation}
where the \emph{interaction potential} $\mathsf{V}:\R \to \R$ is an even function, belongs to $C^2(\R)$ and is uniformly convex and has bounded second derivative. The law~$\mu$ of the random field is then given (formally) by the Gibbs state
\begin{equation}
d\mu(\phi):= \frac{1}{Z} \exp\left( -H(\phi) \right) \,d\phi,
\end{equation}
where $d\phi$ denotes Lebesgue measure on $\R^{\Zd}$ and $Z$ is a normalizing constant which makes $\mu$ a probability measure on the space of configurations.

\smallskip

The above definition of the probability measure~$\mu$ does not make sense because~$\Zd$ is not a finite set and $H$ will be infinite for typical configurations. To define the infinite-volume measure rigorously, we take~$\mu$ to be the weak limit as $L\to \infty$ of finite-volume Gibbs states defined on the space of configurations restricted to the cube $Q_L:=[-L,L]^d\cap \Zd$ with zero boundary data (or alternatively, periodic boundary conditions). To define these, we let $\Omega_0(Q_L)$ be the set of functions $\phi:Q_L \to \R$ such that $\phi=0$ on $\partial Q_L$. Note that we can identify $\Omega_0(Q_L)$ with the Euclidean space $\R^{Q_L^\circ}$ by first identifying an element of $\R^{Q_L^\circ}$ with a function~$Q_L^\circ\to \R$ and then extending the function to~$Q_L$ by defining it to be zero on $\partial Q_L$. We let $d\phi$ denote Lebesgue measure on $\Omega_0(Q_L)$, with this identification in mind, and we define, for each \emph{tilt}~$\xi\in\Rd$, the measure $\mu_{L,\xi}$ by 
\begin{equation}
\label{e.defmuL}
d\mu_{L,\xi} (\phi):= \frac{1}{Z_{L,\xi}} \exp\left( - \sum_{x\in Q_L^\circ}\sum_{y \sim x} \mathsf{V}( \phi(y)- \phi(x)-\xi \cdot (y-x))  \right) \, d\phi,
\end{equation}
where the normalizing constant~$Z_{L,\xi}$, called the \emph{partition function}, is defined by
\begin{equation} 
\label{e.defZL}
Z_{L,\xi}:= 
\int_{\Omega_0(Q_L)} 
\exp\left( - \sum_{x\in Q_L^\circ}\sum_{y \sim x} \mathsf{V}( \phi(y)- \phi(x)-\xi \cdot (y-x))  \right) \, d\phi.
\end{equation}
We denote by~$\left\langle \cdot \right\rangle_{\mu_{L,\xi}}$ the expectation with respect to~$\mu_{L,\xi}$. It is well-known (see for instance~\cite{FS} or Section~\ref{ss.infinite} below) that, for each $\xi\in\Rd$, the measures $\mu_{L,\xi}$ converge weakly as $L\to \infty$ to a unique measure $\mu_{\infty,\xi}$ on the space of \emph{gradient fields} (or configurations on $\Zd$ modulo constants). This is our infinite-volume Gibbs state. In the case that $\mathsf{V}$ is quadratic, the law of $\{ \phi(x)\}$ under~$\mu_{\infty,\xi}$ is that of a discrete Gaussian free field (or massless free field). This model is thus sometimes called a ``gradient perturbation of a massless free field.''

\smallskip

There are natural \emph{Langevin dynamics} which are reversible with respect to~$\mu_{\infty,\xi}$. These are defined by
\begin{equation}
\label{e.langevin}
d\phi_t(x) = \sum_{y\in\Zd, \, |y-x|=1}
 \mathsf{V}'( \phi_t(y)-\phi_t(x) - \xi \cdot (y-x)) \,dt  + \sqrt{2} \,dB_t(x), 
\end{equation}
where $\{ B_t(x)\,:\, x\in \Zd\}$ is a family of independent Brownian motions. We may think of~\eqref{e.langevin} as a Markovian diffusion process on the space of configurations, with each site $\varphi_t(x)$ performing an independent Brownian motion and also interacting with its nearest neighbors through an elastic force given by~$\mathsf{V}'$. The measure~$\mu_{\infty,\xi}$ is the invariant measure of this process, and it therefore describes the law of the typical configuration, evolving according to~\eqref{e.langevin}, after a long time. 

\smallskip

Of primary interest is the large-scale (macroscopic) statistical behavior of the field $\nabla \phi$ under the equilibrium measures~$\mu_{L,\xi}$ or~$\mu_{\infty,\xi}$, as well as that of the Langevin dynamics. Since the model was introduced in the 1970s by Brascamp, Lieb and Lebowitz~\cite{BLL}, notable progress was made by Naddaf and Spencer~\cite{NS}, who proved a central limit for (rescaled) linear functions of $\nabla \phi$. In other words, they characterized the scaling limit for $\nabla \phi$ (with zero tilt) as a Gaussian free field with covariance matrix~$\ahom$, giving a satisfactory description of the fluctuations of the equilibrium profile. Their argument was based on a beautiful observation that the scaling limit can be derived from an elliptic homogenization problem via the Helffer-Sj\"ostrand representation \cite{HS}, with~$\ahom$ appearing as the homogenized matrix. Later, Miller~\cite{Mi} generalized the approach of~\cite{NS} to finite-volume measures and to general tilts~$\xi$ with a corresponding homogenized matrix~$\ahom(\xi)$ depending on~$\xi$. We also mention the earlier work of Brydges and Yau~\cite{BY}, who proved a similar result to~\cite{NS} in a perturbative setting using a renormalization group approach,  as well as the subsequent work of~\cite{BS} who extended the results of~\cite{NS} to a special class of nonconvex interaction potentials. 

\smallskip

After the breakthrough work of Naddaf and Spencer, macroscopic deviations from the average equilibrium profile were characterized by Funaki and Spohn~\cite{FS} in terms of the nonlinear PDE
\begin{equation}
\label{e.hydro}
\partial_t h - \nabla \cdot \left( D \sigma \!\left( \nabla h \right) \right) = 0 \quad \mbox{in} \ (0,\infty) \times \Rd.
\end{equation}
It was shown in~\cite{FS} that the (deterministic) solution~$h$ of~\eqref{e.hydro} describes macroscopic behavior of the field $\varphi_t$ evolving by the Langevin dynamics, starting from a smooth, macroscopic initial datum. In other words,~\eqref{e.hydro} is the hydrodynamic limit of~\eqref{e.langevin}. The nonlinear function $\sigma:\Rd\to \R$ is the \emph{surface tension} or \emph{free energy} for the model. The \emph{finite-volume surface tension} is defined for $L\in\N$  by 
\begin{equation}
\label{e.fst}
\sigma_{L}\left( \xi \right) 
:=
- \, \frac{1}{\left\vert Q_{L}\right\vert }
\log \frac{Z_{L,\xi}}{Z_{L,0}}, \quad \xi \in\Rd. 
\end{equation}
This is the energy per unit volume charged by the Hamiltonian for tilting a macroscopically flat interface to one with slope~$\xi$. 
It was proved in \cite{MMR,FS,Funaki,Sh,Dar} by a subadditive argument that the following limit exists, which defines the infinite-volume surface tension $\sigma$ for the model:
\begin{equation*}
\sigma \left( \xi\right) 
:= 
\lim_{L\rightarrow \infty }\sigma_{L}\left( \xi\right).
\end{equation*}
It is relatively easy to prove that~$\sigma$ is uniform convex and $C^{1,1}$, which imply that the equation~\eqref{e.hydro} is uniformly parabolic and therefore possesses a satisfactory well-posedness theory for weak solutions. However, in order to obtain the existence of a classical solution~$h$ of~\eqref{e.hydro} by the Schauder theory, we need that the map~$\xi \mapsto D\sigma(\xi)$ is $C^{1,\alpha}$ for some $\alpha>0$. That is, $\sigma \in C^{2,\alpha}$.

\smallskip

Later Giacomin, Olla and Spohn~\cite{GOS} went to the next order in this description by showing that the scaling limit of the fluctuations around the macroscopic profile~$h$ solving~\eqref{e.hydro} are given by an the SPDE of the form
\begin{equation}
\label{e.hydro.fluct}
\partial \zeta - \nabla \cdot \left( \ahom(\nabla h) \nabla \zeta \right)  =  \sqrt{2} \dot{W}
\quad
\mbox{in} 
\ (0,\infty) \times \Rd,
\end{equation}
where~$\dot{W}$ is a space-time Gaussian white noise. It is conjectured in~\cite{GOS}, 
that the Hessian $D^2\sigma(\xi)$ of the surface tension should coincide with the covariance matrix~$\bar{\a}(\xi)$ of the limiting GFF---which can be viewed as a ``fluctuation-dissipation relation.'' If confirmed, it would imply that the equation~\eqref{e.hydro.fluct} is the linearization of~\eqref{e.hydro} with an additional white noise forcing term. This conjecture is still open until now, with the main obstacle being the question of~$C^2$~regularity of the surface tension.

\smallskip

While it is simple to obtain that~$\sigma\in C^{1,1}$, it is still not known to be twice differentiable at any particular point, much less $C^2$. The question of the~$C^2$ regularity of~$\sigma$ has been open for many years: see the discussions in Funaki and Spohn~\cite{FS} and Caputo and Ioffe~\cite{CI}, for instance. In the lecture notes of Funaki~\cite[Problem 5.1]{Funaki} it was also called ``one of the important open problems'' for the~$\nabla\phi$ model. As far as we know, the only $C^2$ regularity result for the surface tension outside of the quadratic case was obtained in a paper of Adams, Koteck\'y and M\"uller~\cite{AdKM}. Using an elaborate renormalization group argument, they proved that $\sigma \in C^3$ in a neighborhood of the origin for certain small perturbations of a quadratic potential. These perturbations are required to be very small in a large ball centered at the origin---with smallness measured in a very strong norm (at least $C^{14}$), but they may be larger far from the origin and even permit~$\mathsf{V}''(t)$ to be negative for~$t$ very large (thus allowing certain  nonconvex interaction potentials). We also mention the earlier, related work of Cotar, Deuschel and M\"uller~\cite{CDM} who proved the strict convexity of surface tension in a perturbative setting (that also allows for nonconvex interaction potentials).

\smallskip

The main result of the paper, stated below in Theorem~\ref{t.surfacetension}, resolves the regularity question by showing that $\sigma \in C^{2,\beta}$ for some $\beta>0$ under an additional (mild) regularity assumption on~$\mathsf{V}$, namely that $\mathsf{V}''\in C^{0,\gamma}$ for some $\gamma>0$. We also characterize the Hessian of surface tension by showing that
\begin{equation}
D^2\sigma(\xi) = \ahom(\xi),
\end{equation}
thus positively resolving the fluctuation-dissipation relation conjecture of~\cite{GOS}.

\smallskip

Before presenting the theorem, we state our assumptions. Throughout the paper,~$d\ge2$ denotes the ambient dimension and we fix an exponent $\gamma\in \left(0,1\right]$ and parameters $0 <\lambda \leq \Lambda <\infty$ and $\mathsf{M}\in (0,\infty)$. For short, we write
\begin{equation*}
\data := \left(d,\gamma,\lambda,\Lambda\right).
\end{equation*}
The assumptions on the interaction potential $\mathsf{V}:\R\to\R$ are as follows:
\begin{enumerate}

\item[(i)] \emph{Regularity}: $\mathsf{V}\in C^{2,\gamma}(\R)$ for some $\gamma\in \left(0,\tfrac14\right]$ and 
\begin{equation}
\label{e.V.C2gamma}
\sup_{t,s\in\R, \, t\neq s} \frac{\left| \mathsf{V}''(t) - \mathsf{V}''(s) \right| }{|t-s|^\gamma} \leq \mathsf{M}.
\end{equation}

\smallskip

\item[(ii)] \emph{Uniform convexity}: for every $t\in\R$, we have $\lambda \leq \mathsf{V}''(t) \leq \Lambda$.

\smallskip

\item[(iii)] \emph{Symmetry}: for every $t\in\R$, we have $\mathsf{V}(t) = \mathsf{V}(-t)$.
\end{enumerate}

\smallskip

Under these assumptions on the potential~$\mathsf{V}$, we prove the following theorem, which is the main result of the paper.   

\begin{theorem}
\label{t.surfacetension}
There exists an exponent~$\beta(\data)\in \left(0,\tfrac12\right)$ such that $\sigma\in C^{2,\beta}_{\mathrm{loc}}(\Rd)$ and the Hessian of~$\sigma$ is given by
\begin{equation}
\label{e.ST.identification}
D^2\sigma(\xi) = \ahom(\xi). 
\end{equation}
Moreover, for every $R\in[1,\infty)$, there exists a constant $C(R,\mathsf{M},\data)<\infty$  such that, for every~$\xi,\xi' \in B_R$ and $L\in\N$, 
\begin{equation}
\label{e.sigmaC2beta}
\left|D^2 \sigma (\xi)- D^2 \sigma (\xi')\right|
\leq C\left|\xi - \xi'\right|^\beta
\end{equation}
and
\begin{equation}
\label{e.sigmaL.rate}
\sup_{\xi \in\Rd} \left| D^2\sigma_L(\xi) - D^2\sigma(\xi) \right| \leq C L^{-\beta}. 
\end{equation}
\end{theorem}

In addition to the $C^2$ regularity of~$\sigma$ and the identification of its Hessian, the theorem above specifies a quantitative, algebraic rate of convergence of the finite-volume surface tension~$\sigma_L$ to~$\sigma$ in the $C^2$ norm. This estimate is perhaps the strongest assertion in the theorem since, as we will see from the proof, it implies the other two statements. In particular, as it is relatively easy to prove that $D^2\sigma_L\in C^{0,\alpha}$ for some $\alpha$ provided that $\left\| D^2 \sigma_L \right\|_{C^{0,\alpha}}$ is allowed to depend on $L$, the estimate~\eqref{e.sigmaL.rate} implies the $C^2$ regularity of $\sigma$ (a uniform limit of continuous functions is continuous). 

\smallskip

Our proof of Theorem~\ref{t.surfacetension} starts from the insight of Naddaf and Spencer that the fluctuations of the~$\nabla\phi$ field are strongly related to an elliptic homogenization problem for the Helffer-Sj\"ostrand equation (see~\eqref{e.HS.eqn.QL} below) and combines it with some recent ideas developed in the theory of quantitative stochastic homogenization for elliptic equations in divergence-form (see~\cite{AS,GNO1,AKMbook,GO6} and the references therein). In particular, the recent variational approach to quantitative homogenization based on a multiscale analysis of certain subadditive energy quantities, developed in~\cite{AS,AKMbook}, is very natural in this context due to the fact that the analogue, for the Helffer-Sj\"ostrand equation, of one of the subadditive quantities used there  turns out to coincide precisely with the Hessian $D^2\sigma_L$ of the finite-volume surface tension. Our strategy is therefore to adapt the arguments of~\cite{AS,AKMbook} to obtain an algebraic rate for the convergence of these subadditive quantities to their limit, which amounts to proving the estimate~\eqref{e.sigmaL.rate}.

\smallskip

This adaptation of the methods of~\cite{AS,AKMbook} is not straightforward since they were developed for random coefficient fields with a finite range of dependence, whereas the Helffer-Sj\"ostrand equation is a deterministic equation in essentially infinite dimensions. However, as we will show, they turn out to be quite flexible; the finite range of dependence can be replaced by a combination of the Brascamp-Lieb inequality~\cite{BL} and some new coupling arguments based on the probabilistic interpretation of the equation. This provides us with sufficient decorrelation of the gradient field to implement the multiscale homogenization arguments of~\cite{AS,AKMbook}. 

\smallskip

The proof of Theorem~\ref{t.surfacetension} applies to potentials with a weaker regularity assumption, but we do need better than simply~$\mathsf{V}\in C^2$. What is required to obtain~$\sigma\in C^2$ is that, for a large exponent~$q(\data)<\infty$ (related to the smallness of the H\"older exponent in the parabolic Nash estimate), the interaction potential~$\mathsf{V}$ satisfies
\begin{equation*}
\left| \mathsf{V}''(s) - \mathsf{V}''(t) \right| 
\leq 
\omega\left( |s-t| \right)
\end{equation*}
where the modulus $\omega:[0,\infty) \to [0,\Lambda]$ is an increasing, continuous function such that  
\begin{equation*}
\limsup_{t\to 0} \, \left| \log t\right|^q  \omega ( t ) = 0.
\end{equation*}
In particular, a logarithmic-type modulus suffices but we do need a quantitative assumption for the arguments here to be applicable. 

\smallskip

The result of Theorem~\ref{t.surfacetension} that~$\sigma\in C^{2,\beta}$ is not close to giving the optimal regularity of~$\sigma$, as it is conjectured that~$\sigma\in C^\infty$, at least under suitable regularity assumptions on the potential~$\mathsf{V}$ (although Sheffield~\cite[Section 10.1.1]{Sh} has conjectured that~$\sigma$ is still smooth \emph{without any} regularity assumptions on $\mathsf{V}$). However, we do expect that by combining the ideas in the present paper with some recent methods developed in~\cite{AFK,AFK2} for obtaining higher regularity of homogenized coefficients in the context of stochastic homogenization for nonlinear equations, we will be able to show that~$\sigma \in C^\infty$ for sufficiently smooth interaction potentials. We will return to this problem in a forthcoming work. 

\smallskip

The analysis developed in this paper is of interest apart from the proof of Theorem~\ref{t.surfacetension}. In Proposition~\ref{s.convergence}, we obtain an estimate on the rate of convergence of the subadditive energy quantities, which by analogy to~\cite{AKMbook} represents the first step in a quantitative homogenization program. Since Naddaf and Spencer showed that qualitative homogenization implies the scaling limit of the~$\nabla \phi$ model, we can expect that quantitative homogenization will yield quantitative information regarding the fluctuations. In a forthcoming paper~\cite{AW}, we show that this is indeed the case and, by extending the results in this paper, prove a quantitative scaling limit with enough control to obtain information regarding the pointwise statistics of the gradient field. For instance, in dimension $d=2$ we are able to prove, for each $e\in \Rd$, that the random variable
\begin{equation*}
\frac{\phi([ te ] ) - \phi(0)}{\log^{\frac12} t}
\end{equation*}
converges in law, as $t\to \infty$, to a normal random variable. (Here $[x]$ denotes the nearest lattice point to $x\in\Rd$.)

\smallskip

This paper is organized as follows. In the next section we introduce some notation. In Section~\ref{s.HSe} we derive the Helffer-Sj\"ostrand equation and present some preliminary estimates. The couplings are constructed in Section~\ref{s.couple}, where we also compare the Helffer-Sj\"ostrand solutions with respect to different finite-volume measures. In Sections~\ref{s.subadditive} and~\ref{s.convergence} we introduce the subadditive energy quantities and show by a multiscale iterative argument that they converge at an algebraic rate. We finally prove Theorem~\ref{t.surfacetension} in Section~\ref{s.tension}. Some auxiliary estimates are stated in Appendix~\ref{s.aux}.

\section{Preliminaries and notation}

We work in the Euclidean lattice $\Z^d$, where $d\ge 2$. If~$x,y\in \Zd$, we write~$x\sim y$ if~$|x-y|=1$. We denote the lexicographical order on $\Zd$ by $\ll$, that is, we write $x\ll y$ if $x=(x_1,\ldots,x_d)$ and $y=(y_1,\ldots,y_d)$ and $x_i \leq y_i$ for every $i\in\{1,\ldots,d\}$. Notice that if $x\sim y$, then either $x\ll y$ or $y \ll x$. We let~$\mathcal{E}(\Zd)$ denote the set of directed edges $(x,y)$ on~$\Zd$ such that $x\sim y$ and $x \ll y$. The \emph{interior} $U^\circ$ and \emph{boundary} $\partial U$ of a subset $U\subseteq \Zd$ are defined by 
\begin{equation*}
U^\circ :=
\left\{
x\in U \,:\,
x \sim y \implies y\in U
\right\}
\quad\mbox{and} \quad 
\partial U:= U \setminus U^\circ.
\end{equation*} 
We define the set of interior edges in $U$ by
\begin{equation*}
\mathcal{E}(U)
:=
\left\{ 
(x,y) \in U\times U \,:\, 
(x,y) \not\in \partial U \times \partial U, \
x\sim y, \
x \ll y
\right\}.
\end{equation*}
Given a subset~$U \subseteq \Zd$, we denote by $\R^{U}$ the set of real-valued functions $\phi: U \to \R$.  Define $\Omega_0(U)$ to be the set of functions $\phi:U \to \R$ such that $\phi=0$ on $\partial U$. When $U =\Z^d$ we simply denote it by $\Omega$. Given $e=(x,y)\in \mathcal{E}(U)$ and $\phi\in \R^U$, we define $\nabla \phi(e):= \phi(y) - \phi(x)$. The formal adjoint~$\nabla^*$ of~$\nabla$, which is the discrete version of the negative of the divergence operator, is defined for functions~$\g:\mathcal{E}(U)\to \R$ by 
\begin{equation}
\left( \nabla^*\g\right)(x)
:=
\sum_{y\sim x, \, y\ll x} \g(y,x) 
- \sum_{y\sim x, \, x\ll y} \g(x,y), 
\quad x\in U^\circ.
\end{equation}

%

We interpret the canonical element $\Omega_0(Q_L)$ sampled by the measures~$\mu_{L,\xi}$ (or any probability measure $\mu$ on $\Omega_0(Q_L)$) as a random scalar field with~$\phi(x)$ representing the height of a discrete random surface (in $d+1$ dimensions) at the point~$x$. We denote expectations with respect to these measures by 
\begin{equation*}
\left\langle X \right\rangle_{\mu_{L,\xi}} 
:= \int_{\R^{Q_L}} X(\phi) \,d\mu_{L,\xi}(\phi)
\quad \mbox{and} \quad 
\left\langle F \right\rangle_{\mu} 
:= \int_{\R^{\Zd}} F(\phi) \,d\mu(\phi), 
\end{equation*}
and so forth. We also denote variances by
\begin{equation*}
\var_{\mu_{L,\xi}}\left[ X \right] := \left\langle \left| X- \left\langle X \right\rangle_{\mu_{L,\xi}} \right|^2 \right\rangle_{\mu_{L,\xi}}
\quad \mbox{and} \quad 
\var_\mu\left[ X \right] := \left\langle \left| X- \left\langle X \right\rangle_\mu \right|^2 \right\rangle_\mu.
\end{equation*}
We define, for each $x\in U$, the basis element $\omega_x\in \Omega_0(U)$ by 
\begin{equation*}
\omega_x(y):= \left\{ 
\begin{aligned}
& 1 & \mbox{if} \ x=y,\\
& 0 & \mbox{if} \ x\neq y,
\end{aligned}
\right. 
\end{equation*}
and the differential operator~$\partial_x$ by 
\begin{equation*}
\partial_x u (\phi):= \lim_{h\to 0} \frac1h \left( u(\phi+h\omega_x) - u(\phi) \right).
\end{equation*} 
We let $C^\infty(\Omega_0(U))$ denote the set of smooth functions on $\Omega_0(U)$, that is, the functions for which the mixed derivatives of all orders exist. 

For $p\in [1,\infty)$, and $X$ a Banach space, we define $L^p(U;X)$ to be the set of measurable functions $u:U \to \R$ with respect to the norm 
\begin{equation*}
\left\| u \right\|_{L^p(U;X)}
:=
\left(
\sum_{x\in U} \left\| u(x) \right\|_X^p \right)^{\frac1p}.
\end{equation*}
 Also define $L^p(\mu)$ to be the set of measurable functions $u:\Omega \to \R$ such that 
\begin{equation*}
\left\| u \right\|_{L^p(\mu)} 
: = \left( \int_\Omega \left| u(\phi)\right|^p \,d\mu(\phi) \right)^{\frac1p}
< +\infty.
\end{equation*}
We define $H^1(\mu)$ to be closure of the set of smooth functions $u\in C^\infty(\Omega)$ with respect to the norm
\begin{equation*}
\left\| u \right\|_{H^1(\mu)}
:=
\left(
\left\| u \right\|_{L^2(\mu)}^2 
+
\sum_{x\in\Zd} \left\| \partial_x u \right\|_{L^2(\mu)}^2
\right)^{\frac12}.
\end{equation*}
We let $H^{-1}(\mu)$ denote the dual space of~$H^1(\mu)$, that is, the closure of~$C^\infty(\Omega)$ functions under the norm
\begin{equation*}
\left\| 
w
\right\|_{H^{-1}(\mu)}:= 
\sup
\left\{ 
\int_\Omega u(\phi) w(\phi) \,d\mu(\phi)
\,:\,
u\in H^1(\mu), \ \left\| u \right\|_{H^1(\mu)} \leq 1
\right\}. 
\end{equation*}
We define the space $L^2(U,\mu)= L^2(U;L^2(\mu))$ to be the set of measurable functions $u:U\times \Omega_0(U) \to \R$ with respect to the norm 
\begin{equation*}
\left\| u \right\|_{L^2(U,\mu)}
:=
\left(
\sum_{x\in U} \left\| u(x,\cdot) \right\|_{L^2(\mu)}^2 \right)^{\frac12}.
\end{equation*}
We also define $H^1(U,\mu)$ by the norm
\begin{equation*}
\left\| u \right\|_{H^1(U,\mu)}
:=
\left( 
\sum_{x\in U} 
\left\| u(x,\cdot) \right\|_{H^1(\mu)}^2
+
\sum_{e\in \mathcal{E}(U)} 
\left\| \nabla u(e,\cdot) \right\|_{L^2(\mu)}^2 
\right)^{\frac12} 
\end{equation*}
The subset $H^1_0(U,\mu) \subseteq H^1(U,\mu)$ consists of those functions $u\in H^1(U,\mu)$ which satisfy
$u(x,\phi) = 0$ for every $\partial U\times \Omega_0(U)$. 
We also define the seminorm 
\begin{equation*}
\left\llbracket 
u 
\right\rrbracket_{H^1(U,\mu)} 
:=
\left( 
\sum_{x\in U} \sum_{y\in U^\circ}
\left\|  \partial_y u(x,\cdot) \right\|_{L^2(\mu)}^2
+
\sum_{e\in\mathcal{E}(U)} 
\left\| \nabla u(e,\cdot) \right\|_{L^2(\mu)}^2 
\right)^{\frac12}.
\end{equation*}
We define $H^{-1}(U,\mu)$ to be the dual space of $H^1_0(U,\mu)$. 
That is, $H^{-1}(U,\mu)$ is the closure of smooth functions
with respect to the norm
\begin{equation*}
\left\| 
w
\right\|_{H^{-1}(U,\mu)}:= 
\sup
\left\{ 
\sum_{x\in U} \int_{\Omega_0(U)} u(x,\phi) w(x,\phi) \,d\mu(\phi)
\,:\,
u\in H^1_0(U,\mu), \ \left\| u \right\|_{H^1(U,\mu)} \leq 1
\right\}. 
\end{equation*}
It is sometimes convenient to work with the volume-normalized versions of the $L^2$ and Sobolev norms, defined by 
\begin{equation*}
\left\| u \right\|_{\underline{L}^2(U,\mu)}
:=
\left(
\frac1{|U|}
\sum_{x\in U} \left\| u(x,\cdot) \right\|_{L^2(\mu)}^2 \right)^{\frac12},
\end{equation*}
\begin{equation*}
\left\| u \right\|_{\underline{H}^1(U,\mu)}
:=
\left( 
\frac1{|U|}
\sum_{x\in U} 
\left\| u(x,\cdot) \right\|_{H^1(\mu)}^2
+
\frac1{|U|}\sum_{e\in\mathcal{E}(U)} 
\left\| \nabla u(e,\cdot) \right\|_{L^2(\mu)}^2 
\right)^{\frac12},
\end{equation*}
\begin{multline*}
\left\| 
w
\right\|_{\underline{H}^{-1}(U,\mu)}
\\
:= 
\sup
\left\{ 
\frac1{|U|}\sum_{x\in U} \int_{\Omega_0(U)} u(x,\phi) w(x,\phi) \,d\mu(\phi)
\,:\,
u\in H^1_0(U,\mu), \ \left\| u \right\|_{\underline{H}^1(U,\mu)} \leq 1
\right\}. 
\end{multline*}
Finally we notice that the formal adjoint of~$\partial_x$ with respect to $\mu_{L,\xi}$, which we denote as $\partial_x^*$, is given by 
\begin{equation*}
\partial_x^* w := -\partial_x w 
+
\sum_{y\sim x} 
\mathsf{V}'(\phi(y)-\phi(x) - \xi \cdot (y-x)) w(\phi).
\end{equation*}
This can be easily checked by the identity for all $u,v \in H^1(\mu_{L,\xi})$ that
\begin{equation*}
\left\langle (\partial_x u) v \right\rangle_{\mu_{L,\xi}}
= \left\langle u (\partial^*_x v)  \right\rangle_{\mu_{L,\xi}}
\end{equation*}
We also have the commutator identity
\begin{equation}
\label{e.commutator}
\left[ \partial_x, \partial_y^* \right]  
=
- \indc_{\{x \sim y\}}\mathsf{V}''\left( \phi(y) - \phi(x) - \xi \cdot (y-x) \right)
+ \indc_{\{x=y\}} \sum_{e \ni x} \mathsf{V}''\left( \nabla \phi(e) - \nabla \ell_\xi(e)\right)
\end{equation}

\section{The Helffer-Sj\"ostrand equation}
\label{s.HSe}


In this section, we study the Langevin dynamics which are reversible with respect to the Gibbs measures~$\mu$ and~$\mu_L$ and their infinitesimal generators. Following~\cite{NS} (which was in turn inspired by the works~\cite{HS,Sj}), we introduce the Helffer-Sj\"ostrand operator, and show that it arises naturally when one considers the variance of certain observables with respect to the Gibbs measures. We then show that this operator is itself the generator of a Markov processes in which we augment the Langevin dynamics with a random walk. In the following section we will use this dynamical interpretation of the Helffer-Sj\"ostrand operator to construct couplings of the finite-volume and infinite-volume Gibbs measures which are well-behaved with respect to the Helffer-Sj\"ostrand operator. 

\smallskip


\subsection{Finite-volume Gibbs measures}

For reasons which will become apparent in the next section, in addition to the finite-volume measures~$\mu_{L,\xi}$ which are defined in~\eqref{e.defmuL} with Dirichlet boundary conditions, we also consider Gibbs measures defined the same cube~$Q_L$ but with periodic boundary conditions.  
These periodized Gibbs measures are denoted by $\mu_{L,\xi,\per}$ and we define them as follows. 
Given $L\in\N$ with $L\geq 2$, we fix a special point $x_0 \in \partial Q_L$ and define $\Omega_{\per}(Q_L)$ to be the set of $2L$--periodic functions on $\Zd$ which vanish at $x_0$, that is, $\phi\in \Omega_{\per}(Q_L)$ if $\phi:\Zd \to \R$, $\phi(x_0)=0$ and $\phi(x) = \phi(y)$ for every $x,y\in \Zd$ satisfying $x-y \in 2L\Zd$. Observe that $\Omega_{\per}(Q_L)$ can be identified with the Euclidean space $\R^{[-L,L)^d\cap \Zd \setminus \{x_0\}}$ which has dimension $(2L)^d-1$. 

\smallskip

We define the measure $\mu_{L,\xi,\per}$ on $\Omega_\per(Q_L)$ by 
\begin{equation}
\label{e.muper.def}
d\mu_{L,\xi,\per} (\phi):= \frac{1}{Z_{L,\xi,\per}} \exp\left( - \sum_{x\in Q_{L,\per}}\sum_{y \sim x} \mathsf{V}( \phi(y)- \phi(x)-\xi \cdot (y-x))  \right) \, d\phi,
\end{equation}
where $d\phi$ is Lebesgue measure with respect to the identification of~$\Omega_{\per}(Q_L)$ with Euclidean space mentioned above and $Z_{L,\xi,\per}$ is the normalizing constant which makes $\mu_{L,\xi,\per}$ a probability measure.

\smallskip

The reason that we choose to define $\Omega_\per(Q_L)$ in the way we did, by requiring $\phi(x_0)=0$, is because the right side of~\eqref{e.muper.def} does not change when we add constants to $\phi$, and therefore it must be considered as a measure on periodic functions modulo constants. Therefore we need to quotient by constant functions in some way in our definition of~$\Omega_\per(Q_L)$. It may seem natural to consider mean-zero periodic functions, however as we will discover below, in our context it is actually easier to work with functions vanishing at a fixed boundary point~$x_0$. For convenience we also require that $x_0 \in [-L,L)^d\cap \Zd$ and we denote $Q_{L,\per} := [-L,L)^d\cap \Zd \setminus \{x_0\}$.

\smallskip

It is immediate from the definition~\eqref{e.muper.def} that, with respect to~$\mu_{L,\xi,\per}$, the distribution of~$\nabla \phi(e)$ does not depend on~$e$ and thus
\begin{equation}
\label{e.perexpvanish}
\left\langle \nabla \phi(e) \right\rangle_{\mu_{L,\xi,\per}} = 0,\quad \forall e\in \mathcal{E}(Q_L). 
\end{equation}
Consequently, as $\phi(x_0) = 0$, we may use~\eqref{e.perexpvanish} to sum over a path from $x_0$ to any point $x$ to find  
\begin{equation}
\label{e.periodic.meanzero}
\left\langle  \phi(x) \right\rangle_{\mu_{L,\xi,\per}} = 0\quad \forall x\in Q_L. 
\end{equation}
More generally, it is easy to see that the law of~$\nabla \phi$ is invariant under the action of the translation group on~$\Z^d$. It is this stationarity property that makes~$\mu_{L,\xi,\per}$ convenient to work with in certain situations.

\subsection{The Helffer-Sj\"ostrand equation}
As mentioned in the introduction, the finite-volume Gibbs measures $\mu_{L,\xi}$ and~$\mu_{L,\xi,\per}$ can be realized as the invariant measures of a certain Markov process. Consider the diffusion process $\{ \phi_t\}$ on $\Omega_0(Q_L)$ evolving according to the Langevin dynamics 
\begin{equation}
\label{e.dynamics.phi.QL}
\left\{
\begin{aligned}
& d\phi_t(x) 
= \sum_{y\sim x} \mathsf{V}'( \phi_t(y)-\phi_t(x) - \xi \cdot (y-x)) \,dt  + \sqrt{2} \,dB_t(x), && x\in Q_L^\circ, 
\\ & 
\phi_t(x) = 0, && x \in \partial Q_L,
\end{aligned} 
\right.  
\end{equation}
where $\{ B_t(x) \,:\, x\in Q_L^\circ \}$ is a family of independent Brownian motions. 
The infinitesimal generator of this process 
is the operator~$\L_{\mu_{L,\xi}}$ defined by
\begin{align*}
\L_{\mu_{L,\xi}} F (\phi) 
: = 
\sum_{x\in Q_L^\circ} \partial_x^2 F(\phi) 
- 
\sum_{x\in Q_L^\circ} 
 \sum_{y\sim x}
\mathsf{V}'(\phi(y)-\phi(x) - \xi \cdot (y-x))\partial_xF(\phi).
\end{align*}
The domain of~$\L_{\mu_{L,\xi}}$ includes $C^2_c(\Omega_0(Q_L))$. Notice that we can write $\L_{\mu_{L,\xi}}$ as
\begin{equation*}
\L_{\mu_{L,\xi}} F = - \sum_{x\in Q_L^\circ} \partial_x^* \partial_x F, 
\end{equation*}
where~$\partial_x^*$ denotes the formal adjoint of~$\partial_x$ with respect to $\mu_{L,\xi}$, given by 
\begin{equation*}
\partial_x^* w := -\partial_x w 
+
\sum_{y\sim x} 
\mathsf{V}'(\phi(y)-\phi(x) - \xi \cdot (y-x)) w(\phi).
\end{equation*}
The operator $\L_{\mu_{L,\xi}}$ is thus self-adjoint with respect to the measure $\mu_{L,\xi}$, that is, 
\begin{equation*}
\left\langle F \L_{\mu_{L,\xi}}  G \right \rangle_{\mu_{L,\xi}} 
=
\left\langle G \L_{\mu_{L,\xi}}  F \right \rangle_{\mu_{L,\xi}} 
=
-\sum_{x\in\Zd} \langle \partial_x F, \partial_x G \rangle_{\mu_{L,\xi}},
\quad \forall F,G \in C^2_c(\Omega_0(Q_L)).
\end{equation*}
In particular, 
\begin{equation}
\label{e.GLmuLF.bnd}
\left| \left\langle G \L_{\mu_{L,\xi}}  F \right \rangle_{\mu_{L,\xi}} \right| 
\leq 
\left\| F \right\|_{H^1(\mu_{L,\xi})}\left\| G \right\|_{H^1(\mu_{L,\xi})}, 
\quad \forall F,G \in C^2_c(\Omega_0(Q_L)),
\end{equation}
where we define the norm~$\| \cdot\|_{H^1(\mu_{L,\xi})}$ by 
\begin{equation*}
\| F \|_{H^1(\mu_{L,\xi})} :=
\left\langle F^2 \right\rangle_{\mu_{L,\xi}}^{\frac12}
+
\left( \sum_{x\in Q_L^\circ} \left\langle (\partial_x F)^2 \right\rangle_{\mu_{L,\xi}} \right)^{\frac12}.
\end{equation*}
Let $H^1(\mu_{L,\xi})$ be the completion of $C^2_c(\Omega_0(Q_L))$ with respect to the norm~$\| \cdot\|_{H^1(\mu_{L,\xi})}$. It follows from~\eqref{e.GLmuLF.bnd} and a density argument that the domain of the operator $\L_{\mu_{L,\xi}}$ includes the space $H^1(\mu_{L,\xi})$, and we have 
\begin{equation}
\label{e.weakform.L.QL}
\left\langle G \L_{\mu_L} F \right \rangle_{\mu_{L,\xi}} 
=
-\sum_{x\in\Zd} \langle \partial_x F, \partial_x G \rangle_{\mu_{L,\xi}},
\quad \forall F,G \in H^1({\mu_{L,\xi}}).
\end{equation}
The dynamics~\eqref{e.dynamics.phi.QL} are therefore \emph{reversible} with respect to~$\mu_{L,\xi}$, as claimed. Since
\begin{equation}
\label{e.invariant.QL}
\left\langle  \L_{\mu_{L,\xi}} F \right \rangle_{\mu_{L,\xi}} 
=
0,
\quad \forall F \in H^1({\mu_{L,\xi}}),
\end{equation}
the measure~$\mu_{L,\xi}$ is \emph{invariant} under the dynamics.

\smallskip

We may also write down Langevin dynamics which are reversible with respect to~$\mu_{L,\xi,\per}$. 
We let~$\{ \phi_{\per,t}\}$ be  the diffusion process on $\Omega_{\per}(Q_L)$ governed by
\begin{equation}
\label{e.dynamics.phi.QLper}
\left\{
\begin{aligned}
& d\phi_{\per,t}(x) 
= \sum_{y\sim x} \mathsf{V}'( -\xi\cdot (y-x)+ \phi_{\per,t}(y)-\phi_{\per,t}(x) ) \,dt 
+ \sqrt{2} \,dB_t(x), \quad x\in Q_{L,\per}, 
\\ &
\phi_{\per,t}(x_0) = 0,
\\  &
\phi_{\per,t}(x) 
= \phi_{\per,t}(y), \quad x,y \in \Zd,\, x-y\in 2L\Zd.
\end{aligned} 
\right.  
\end{equation}
The corresponding infinitesimal generator for these dynamics is given by
\begin{equation}
\mathcal{L}_{\mu_{L,\xi,\per}} F (\phi):=
\sum_{x\in Q_{L,\per}} \partial_x^2 F(\phi) 
- 
\sum_{x\in Q_{L,\per}} 
 \sum_{y\sim x}
\mathsf{V}'(\phi(y)-\phi(x) - \xi \cdot (y-x))\partial_xF(\phi),
\end{equation}
Similar to the discussion above, the domain of $\mathcal{L}_{\mu_{L,\xi,\per}}$ includes $ H^1(\mu_{L,\xi,\per})$ which is defined to be the completion of $C^2_c(\Omega_\per(Q_L))$ with respect to the norm 
\begin{equation*}
\| F \|_{H^1(\mu_{L,\xi,\per})} :=
\left\langle F^2 \right\rangle_{\mu_{L,\xi,\per}}^{\frac12}
+
\left( \sum_{x\in Q_{L,\per}} \left\langle (\partial_x F)^2 \right\rangle_{\mu_{L,\xi,\per}} \right)^{\frac12}.
\end{equation*}
Following the same argument as for $\mu_{L,\xi}$, we find analogues of~\eqref{e.weakform.L.QL} and~\eqref{e.invariant.QL} with $\mu_{L,\xi,\per}$ in place of~$\mu_{L,\xi}$. In particular, the dynamics are reversible with respect to~$\mu_{L,\xi,\per}$ as claimed. 

\smallskip

We can infer information concerning the measures~$\mu_{L,\xi}$ and $\mu_{L,\xi,\per}$ by studying the behavior of the Markov process defined in~\eqref{e.dynamics.phi.QL}. The study of the 
latter we approach through their infinitesimal generators, the operators~$\L_{\mu_{L,\xi}}$ and~$\L_{\mu_{L,\xi,\per}}$. Throughout the rest of this section, we require that, for some we fix $L\in\N$ with $L \geq 2$ and $\xi\in\Rd$,
\begin{align*}
& \left( \mu,\Omega,\mathcal{L}_\mu,Q \right) 
\quad \mbox{denotes either} \quad  
\left( \mu_{L,\xi},\Omega_0(Q_L),\mathcal{L}_{\mu_{L,\xi}},Q_L^\circ \right) 
\\ 
& \qquad\qquad \mbox{or} \quad 
\left( \mu_{L,\xi,\per},\Omega_\per(Q_L),\mathcal{L}_{\mu_{L,\xi,\per}},Q_{L,\per} \right).
\end{align*}
We also denote $\partial Q:= Q_L\setminus Q_L^\circ$ in the case $\mu=\mu_{L,\xi}$ and $\partial Q:= \{x_0\}$ in the case $\mu=\mu_{L,\xi,\per}$. 
We note that~$\mu$ depends on both~$L$ and~$\xi$, but will leave this dependence implicit in the notation. 

\smallskip

We are motivated to study solutions of the equation 
\begin{equation}
\label{e.LmuL}
-\L_{\mu} F = G.
\end{equation}
We begin with the well-posedness of~\eqref{e.LmuL}, which is based on the following Poincar\'e-type inequality for the measure~$\mu$. 

\begin{lemma}[{Poincar\'e inequality for $H^1(\mu)$}]
\label{l.spectralgap.muL}
There exists $C(d,\lambda)<\infty$ such that, for every $F \in H^1(\mu)$, 
\begin{equation*}
\left\langle \left( F - \left\langle F \right\rangle_{\mu} \right)^2 \right\rangle_{\mu} 
\leq 
CL^2 \sum_{x\in Q} \left\langle (\partial_x F)^2 \right\rangle_{\mu}.
\end{equation*}
\end{lemma}
\begin{proof}
The lemma is a consequence of a more general Poincar\'e-type inequality for log-concave measures attributed to Bakry and Emery (a nice proof of which can be found for instance in~\cite{BBCG}). This result states that
there exists a universal constant~$C<\infty$, such that, for every $N\in\N$, $\theta>0$ and $W\in C^2(\R^N)$ satisfying 
\begin{equation}
\label{e.WhessLB} 
y\cdot D^2W(x) y \geq \theta |y|^2 \quad\forall x,y\in\R^N, 
\end{equation}
if we let $\nu$ denote the probability measure 
\begin{equation*}
d\nu(x) := \left( \int_{\R^N} \exp\left( -W(x') \right)\,dx' \right)^{-1} \exp\left( -W(x) \right)\,dx,
\end{equation*}
then we have
\begin{equation}
\label{e.poincare.logconcave}
\int_{\R^N} F(x)^2 \,d\nu(x) 
\leq 
\left( \int_{\R^N} F(x) \,d\nu(x) \right)^2 
+
\frac C\theta \left( \int_{\R^N} |\nabla F(x)|^2 \,d\nu(x) \right)^2. 
\end{equation}
To apply~\eqref{e.poincare.logconcave}, we observe that the finite volume Gibbs measure~$\mu_{L,\xi}$ can be written in the form of~$\nu$ above with $W=H$, since (as discussed above) we may identity the space of functions~$\phi:Q_L \to \R$ which vanish on $\partial Q_L$ with $\R^{Q_L^\circ}$ and hence with $\R^N$ for~$N=|Q_L^\circ|$. To check the condition~\eqref{e.WhessLB}, we recall that~$D^2H$ is given by 
\begin{equation*}
\partial_{x}\partial_y H(\phi)
=
- \indc_{\{x \sim y\}} \mathsf{V}''\left( \phi(y) - \phi(x) - \xi \cdot (y-x)\right)
+ \indc_{\{x=y\}} \sum_{e \ni x} \mathsf{V}''\left( \nabla \phi(e) -\nabla \ell_\xi (e) \right).
\end{equation*}
Thus, for every function $f:Q_L\to \R$ which vanishes on $\partial Q_L$, we have 
\begin{equation*}
(D^2H(\phi) f)(x) = \sum_{y} \partial_{x}\partial_y H(\phi)f(y) 
= 
\sum_{y\sim x} \mathsf{V}''(\phi(y) -\phi(x) - \xi \cdot (y-x))\left( f(x) - f(y) \right)
\end{equation*}
and therefore, by the (discrete) Poincar\'e inequality on $Q_L$,  
\begin{align*}
\sum_{x\in Q_L} f(x) (D^2H(\phi) f)(x)
&
= \sum_{x \in Q_L} f(x)
\sum_{y\sim x} \mathsf{V}''(\phi(y) -\phi(x) - \xi \cdot (y-x))\left( f(x) - f(y) \right)
\\ & 
=
\frac12 
 \sum_{x \in Q_L} 
\sum_{y\sim x} \mathsf{V}''(\phi(y) -\phi(x) - \xi \cdot (y-x))\left( f(x) - f(y) \right)^2
\\ & 
\geq 
\frac \lambda 2 \sum_{x \in Q_L} 
\sum_{y\sim x} \left( f(x) - f(y) \right)^2
\\ & 
\geq 
c\lambda L^{-2} \sum_{x \in Q_L} \left( f(x) \right)^2.
\end{align*}
This is~\eqref{e.WhessLB} for $\theta=c\lambda L^{-2}$. The inequality~\eqref{e.poincare.logconcave} now yields the lemma in the case that~$\mu=\mu_{L,\xi}$. The argument in the case $\mu = \mu_{L,\xi,\per}$ is similar and so we omit it. 
\end{proof}

We seek to solve~\eqref{e.LmuL} when the right-hand side~$G$ belongs to $H^{-1}(\mu)$, defined as the dual space of $H^1(\mu )$, that is, the completion of $C^\infty_c(\Omega)$ with respect to the norm
\begin{equation*}
\left\| G \right\|_{H^{-1}(\mu)}
:=
\sup\left\{ \left| \left\langle FG \right\rangle_{\mu} \right| 
\,:\,
F \in H^1(\mu), \ \left\| F \right\|_{H^{1}(\mu)} \leq 1 \right\}. 
\end{equation*}
Note that $\L_{\mu}F \in H^{-1}(\mu)$ whenever~$F\in H^1(\mu)$, thanks to~\eqref{e.GLmuLF.bnd}, and therefore we may interpret the equation~\eqref{e.LmuL} as an assertion of equality between two elements of $H^{-1}(\mu)$. Equivalently,~\eqref{e.LmuL} is satisfied if and only if
\begin{equation*}
\sum_{x\in Q} \left\langle (\partial_x F)( \partial_x w) \right\rangle_{\mu} = \left\langle Gw \right\rangle_{\mu} \quad \forall w\in H^1(\mu). 
\end{equation*}
%
In view of~\eqref{e.invariant.QL}, it is natural to expect to have the 
unique solvability of~\eqref{e.LmuL}, up to additive constants, for any right-hand side~$G\in H^{-1}(\mu)$ with zero mean. This is what we demonstrate in the next lemma.

\begin{lemma}
\label{l.solvability.L}
Let  $G\in H^{-1}(\mu)$ with $\left\langle G \right\rangle_{\mu} =0$. Then there exists a solution $F \in H^1(\mu)$ of the equation 
\begin{equation}
\label{e.LmuL.2}
-\L_{\mu } F = G.
\end{equation}
Moreover the solution~$F$ of~\eqref{e.LmuL} is unique up to additive constants, and there exists a constant $C(\data)<\infty$ such that 
\begin{equation*}
\left\| F - \left\langle F \right\rangle_{\mu} \right\|_{H^1(\mu)} 
\leq CL^2 \left\| G \right\|_{H^{-1}(\mu)}.
\end{equation*}
\end{lemma}
\begin{proof}
This result can be obtained by an application of the Lax-Milgram lemma, or, alternatively, by considering the variational problem 
\begin{equation*}
\inf
_{w\in H^1({\mu}), \, \left\langle w \right\rangle_{\mu}=0} 
\left( \frac 12 \sum_{x\in Q} \left\langle ( \partial_x w)^2 \right\rangle_{\mu} - \left\langle Gw \right\rangle_{\mu} \right). 
\end{equation*}
In either case, we just require uniform coercivity with respect to the~$H^1(\mu)$ norm, which is a direct consequence of Lemma~\ref{l.spectralgap.muL}. 
\end{proof}

Using the previous lemma and~\eqref{e.weakform.L.QL}, we obtain the following formula for the variance of an arbitrary element $F\in H^1(\mu)$: 
\begin{equation}
\label{e.varianceform1.QL}
\left\langle \left( F - \langle F \rangle_{\mu} \right)^2 \right\rangle_{\mu}
=
- \sum_{x\in Q} \left\langle ( \partial_x F ) \left( \partial_x \left(  \L_{\mu} ^{-1} \left( F - \langle F \rangle_{\mu}  \right)\right) \right) \right\rangle_{\mu}.
\end{equation}
Let us define, for each $x\in Q$ and $\phi\in \Omega$,
\begin{equation*}
\left\{ 
\begin{aligned}
& u(x,\phi):= - \partial_x \left(  \L_{\mu} ^{-1} \left( F - \langle F \rangle_{\mu}  \right)\right),
\\ & 
f(x,\phi) : = \partial_xF(\phi),
\end{aligned}
\right.
\end{equation*}
so that~\eqref{e.varianceform1.QL} can be written in the form 
\begin{equation}
\label{e.HSrep.muL}
\left\langle \left( F - \langle F \rangle_{\mu} \right)^2 \right\rangle_{\mu}
=
\sum_{x\in Q} 
\left\langle f(x,\cdot) u(x,\cdot)  \right\rangle_{\mu}.
\end{equation}
It is convenient to extend~$u$ to be defined for every~$x\in Q_L$ by setting $u(x,\phi)=0$ for~$x\in \partial Q$ and, in the case $\mu=\mu_{L,\xi,\per}$, extending the domain to $\Zd$ requiring that $u(y,\phi) = u(x,\phi)$ for $x,y\in \Zd$ with $y-x\in 2L\Zd$. 

\smallskip

It is natural to wonder whether the function~$u$ can be characterized as the solution of an equation. Set~$G:=-\L_{\mu} ^{-1} \left( F - \langle F \rangle_{\mu}  \right)$ so that~$G\in H^1(\mu)$ is the solution of 
\begin{equation}
\label{e.HSderveq}
-\L_{\mu} G =  F - \langle F \rangle_{\mu}.
\end{equation}
Formally applying $\partial_x$ to both sides of~\eqref{e.HSderveq}, we are led to the guess that, in the case $\mu=\mu_{L,\xi}$, the function~$u$ is the solution of the problem 
\begin{equation}
\label{e.HS.eqn.QL}
\left\{ 
\begin{aligned}
& -\L_{\mu_{L,\xi}}u + \nabla^* \a \nabla u 
 = f
& \mbox{in} & \ Q \times \Omega, 
\\ & 
u = 0 & \mbox{on} & \ \partial Q_L \times\Omega_0({Q_L}),
\end{aligned}
\right.
\end{equation}
where the coefficients $\a(e,\phi)$ are defined by
\begin{equation}
\label{e.a}
\a(e,\phi):= \mathsf{V}''(\nabla\phi(e)- \nabla \ell_\xi(e)), \quad e\in \mathcal{E}(Q_L),
\end{equation}
and, in the case $\mu=\mu_{L,\xi,\per}$,~\eqref{e.HS.eqn.QL} should be replaced by 
\begin{equation}
\label{e.HS.eqn.QLper}
\left\{ 
\begin{aligned}
& -\L_{\mu_{L,\xi,\per}}u + \nabla^* \a \nabla u 
 = f
& \mbox{in} & \ Q \times \Omega, 
\\ & 
u(x,\phi) = u(y,\phi) & \mbox{if} & \ x,y\in Q_L, \, \phi\in\Omega, \, x-y\in 2L\Zd,
\\ & 
u(x_0) = 0.
\end{aligned}
\right.
\end{equation}
Note that the operator $\nabla^*\a\nabla$ can be expressed explicitly as
\begin{equation*} \label{}
(\nabla^*\a\nabla u)(x,\phi) = \sum_{y\sim x} \mathsf{V}''(\phi(y)-\phi(x) - \xi \cdot (y-x)) \left( u(x) - u(y) \right).
\end{equation*}
We call the partial differential equation in~\eqref{e.HS.eqn.QL} the \emph{Helffer-Sj\"ostrand equation}, and the formula~\eqref{e.HSrep.muL} the \emph{Helffer-Sj\"ostrand representation}.

\smallskip

Let us now show that the preceding derivation of~\eqref{e.HS.eqn.QL} is actually rigorous for any $F\in H^1(\mu)$, provided that we interpret the equation as an assertion of equality between two elements of~$L^2(Q;H^{-1}(\mu))$ or, equivalently, as an equality for each fixed $x\in Q$ between elements of $H^{-1}(\mu)$. Note that~$f$ belongs to this space provided that~$F\in H^1(\mu)$---in fact it belongs to $L^2(Q;L^2(\mu))$---and 
\begin{align*}
\left\| f \right\|_{L^2(Q;H^{-1}(\mu))}^2
=
\sum_{x\in Q} \left\| \partial_x F \right\|_{H^{-1}(\mu)}^2 
\leq
\sum_{x\in Q } \left\| \partial_x F \right\|_{L^2(\mu )}^2 
\leq 
\left\| F \right\|_{H^1(\mu )}^2. 
\end{align*}
Meanwhile, for each fixed $x\in Q $ and $w\in C^\infty_c(\Omega )$, we have
\begin{align*}
\left\langle  w f(x,\cdot) \right\rangle_{\mu}
& =
\left\langle \partial_x^* w\left(  F - \langle F \rangle_{\mu} \right) \right\rangle_{\mu }
\\ & 
=
- \left\langle \partial_x^* w \L_{\mu } G \right\rangle_{\mu }
\\ & 
=
\sum_{y\in Q } 
 \left\langle \left( \partial_y \partial_x^* w \right) \left( \partial_y G \right) \right\rangle_{\mu}
\\ & 
=
\sum_{y\in Q }
 \left\langle \left(  \partial_x^* \partial_y w \right) \left( \partial_y G \right) \right\rangle_{\mu }
+ 
 \sum_{y\in Q }\left\langle \left(  \left[  \partial_y, \partial_x^* \right] w \right) \left( \partial_y G \right) \right\rangle_{\mu }.
\end{align*}
For the first sum on the right side, we observe that 
\begin{equation*}
\sum_{y\in Q }
\left\langle \left(  \partial_x^* \partial_y w \right) \left( \partial_y G \right) \right\rangle_{\mu }
=
\sum_{y\in Q }
\left\langle \left(  \partial_y w \right) \left( \partial_y u(x,\cdot) \right) \right\rangle_{\mu }
= - \left\langle w \L_{\mu } u(x,\cdot) \right\rangle_{\mu }.
\end{equation*}
For the second term, we use the commutator identity 
\begin{equation*}
 \left[  \partial_y, \partial_x^* \right] 
=
- \indc_{\{x \sim y\}}\mathsf{V}''\left( \phi(y) - \phi(x) - \xi \cdot (y-x)\right)
+ \indc_{\{x=y\}} \sum_{e \ni x} \mathsf{V}''\left( \nabla \phi(e) -\nabla \ell_\xi (e) \right)
\end{equation*}
to obtain, for each $x\in Q$, 
\begin{align*}
\sum_{y\in Q }\left\langle \left(  \left[  \partial_y, \partial_x^* \right] w \right) \left( \partial_y G \right) \right\rangle_{\mu }
&
=
\sum_{y\in Q}\left\langle \left(  \left[  \partial_y, \partial_x^* \right] w \right) u(y,\cdot)  \right\rangle_{\mu }
\\ & 
=
\sum_{y\in Q \cup\partial Q}\left\langle \left(  \left[  \partial_y, \partial_x^* \right] w \right) u(y,\cdot)  \right\rangle_{\mu }
\\ & 
=
\sum_{y\sim x} 
\left\langle w 
\mathsf{V}''(\phi(y)-\phi(x) - \xi \cdot (y-x)) 
\left( u(x,\cdot) - u(y,\cdot) \right)
\right\rangle_{\mu }
\\ & 
=
\left\langle w \left(
\nabla^*\a \nabla u\right)(x,\cdot)
\right\rangle_{\mu } . 
\end{align*}
Combining the above, we obtain, for each~$x\in Q$ and~$w\in C^\infty_c(\Omega)$,
\begin{equation*} \label{}
- \left\langle w \L_{\mu } u(x,\cdot) \right\rangle_{\mu }
+
\left\langle w \left(
\nabla^*\a \nabla u\right)(x,\cdot)
\right\rangle_{\mu}
=
\left\langle  w f(x,\cdot) \right\rangle_{\mu}. 
\end{equation*}
By density, we obtain that, in the sense of~$H^{-1}(\mu)$, for every $x\in Q$,  
\begin{equation*} \label{}
-\L_{\mu} u(x,\cdot) + \left(\nabla^*\a \nabla u\right)(x,\cdot)
= f(x,\cdot).
\end{equation*}
This completes the rigorous demonstration of~\eqref{e.HS.eqn.QL}. 

\smallskip

\subsection{Well-posedness of boundary-value problems}
\label{ss.wellpose}

We next show that the boundary-value problem~\eqref{e.HS.eqn.QL} can be solved more generally and more directly than by differentiating~\eqref{e.HSderveq}. We first introduce the appropriate function spaces and norms. For each $U\subseteq Q_L$, we let $H^1(U,\mu )$ be the Banach space of functions $w:U\times \Omega  \to \R$ with respect to the norm~$\left\| \cdot \right\|_{H^1(U,\mu)} $ defined by
\begin{equation*} \label{}
\left\| w \right\|_{H^1(U,\mu)}^2 
:=
\sum_{x\in U} \left\langle w(x,\cdot)^2 \right\rangle_{\mu } 
+
\sum_{e\in \mathcal{E}(U)} \left\langle \left( \nabla w(e,\cdot) \right)^2 \right\rangle_{\mu} 
+
\sum_{y\in Q}\sum_{x\in U}
\left\langle (\partial_yw(x,\cdot))^2 \right\rangle_{\mu }.
\end{equation*}
It is convenient to also define the seminorm~$\left\llbracket \cdot \right\rrbracket_{H^1(U,\mu)}$ by
\begin{equation}
\left\llbracket w \right\rrbracket_{H^1(U,\mu)}^2 := \sum_{e\in \mathcal{E}(U)} \left\langle \left( \nabla w(e,\cdot) \right)^2 \right\rangle_{\mu} 
+
\sum_{y\in Q}\sum_{x\in U}
\left\langle (\partial_yw(x,\cdot))^2 \right\rangle_{\mu}.
\end{equation}
We denote by $H^1_0(U,\mu )$ the subspace of $H^1(U,\mu )$ consisting of those elements of $H^1(U,\mu )$ which vanish on $\partial U \times \Omega$. We also define the dual space $H^{-1}(U,\mu )$ to be the completion of $C^\infty_c(U\times\Omega)$ with respect to 
\begin{equation*}
\left\| f \right\|_{H^{-1}(U,\mu )}
:=
\sup\left\{ 
\left| \sum_{x\in U} \left\langle f(x,\cdot) w(x,\cdot)  \right\rangle_{\mu }  \right|
\,:\,
w \in H^1_0(U,\mu ), \, \left\| w \right\|_{H^1(U,\mu )} \leq 1 \right\}. 
\end{equation*}

\smallskip

We next prove a Poincar\'e inequality for~$H^1(U,\mu)$, which is an easy consequence of Lemma~\ref{l.spectralgap.muL}. We give two statements, one for functions which vanish on the boundary of~$U$ and another for zero-mean functions in the case~$U$ is a cube. For every $U\subseteq Q_L$ and $u \in L^1(U,\mu )$, we denote the mean of $u$ by
\begin{equation}
( u )_{U,\mu } := \sum_{x\in U} \left\langle u(x,\cdot) \right\rangle_{\mu } .
\end{equation}
We sometimes write $(u)_U$ in place of $(u)_{U,\mu}$ for short. 

\begin{lemma}[{Poincar\'e inequality for $H^1(U,\mu)$}]
\label{l.spectralgap.UmuL}
There exists $C(d,\lambda)<\infty$ such that:
\begin{enumerate}
\item[(i)] For every $U\subseteq Q_L$ and $w \in H^1_0(U,\mu)$, 
\begin{equation}
\left\| w \right\|_{L^2(U,\mu)}
\leq 
CL \left\llbracket w \right\rrbracket_{H^1(U,\mu)}.
\end{equation}
\item[(ii)] For every $L\in\N$, cube $Q'\subseteq Q_L$ and $w \in H^1(Q',\mu)$, 
\begin{equation}
\left\| w - \left( w \right)_{Q',\mu} \right\|_{L^2(Q',\mu)}
\leq 
CL \left\llbracket w \right\rrbracket_{H^1(Q',\mu)}.
\end{equation}
\end{enumerate}
\end{lemma}
\begin{proof}
Denote $\overline{w}(x) := \left\langle w(x,\cdot) \right\rangle_{\mu}$. In the case of~(i), since $\overline{w}$ vanishes on $\partial U$, the (discrete) Poincar\'e inequality on $Q_L$ yields
\begin{equation*}
\sum_{x\in U} 
\overline{w}(x)^2 
\leq 
CL^2
\sum_{e\in \mathcal{E}(U)} \left( \nabla \overline{w}(e)\right)^2
\leq 
CL^2 \sum_{e\in \mathcal{E}(U)} \left\langle \left( \nabla w(e,\cdot) \right)^2 \right\rangle_{\mu}.
\end{equation*}
In the case of~(ii), we may suppose without loss of generality that~$\left( w \right)_{Q',\mu}=0$ and then apply the (discrete) Poincar\'e inequality for mean-zero functions on~$Q$ to obtain
\begin{equation*}
\sum_{x\in Q} 
\overline{w}(x)^2 
\leq 
C\diam(Q)^2 
\sum_{e\in \mathcal{E}(Q)} \left( \nabla \overline{w}(e)\right)^2
\leq 
CL^2 \sum_{e\in \mathcal{E}(Q)} \left\langle \left( \nabla w(e,\cdot) \right)^2 \right\rangle_{\mu }.
\end{equation*}
In both cases, an application of Lemma~\ref{l.spectralgap.muL} yields, for each $x\in U$, 
\begin{equation*}
\left\langle \left( w(x,\cdot) - 
\overline{w}(x) \right)^2 \right\rangle_{\mu}
\leq 
CL^2 \sum_{y\in Q} \left\langle (\partial_y w(x,\cdot) )^2 \right\rangle_{\mu}.
\end{equation*}
Summing over $x\in U$ and combining the result with the previous displays gives the lemma. 
\end{proof}

We turn to the well-posedness of the Dirichlet boundary-value problem for the Helffer-Sj\"ostrand equation, which we interpret as an assertion of equality between elements of $H^{-1}(U,\mu)$. 

\begin{lemma}
\label{l.wellposeHS.dir}
Let~$U\subseteq Q_L$ and $f \in H^{-1}(U,\mu)$. There exists a unique solution $u \in H^1(U,\mu)$ of the boundary-value problem 
\begin{equation}
\label{e.HS.eqn.QL.U}
\left\{ 
\begin{aligned}
& -\L_{\mu }u + \nabla^* \a \nabla u 
 = f
& \mbox{in} & \ U^\circ \times \Omega, 
\\ & 
u = 0 & \mbox{on} & \ \partial U \times\Omega,
\end{aligned}
\right.
\end{equation}
which satisfies, for a constant $C(d,\lambda)<\infty$, the estimate
\begin{equation}
\label{e.HSsolest}
L^{-1} \left\| u \right\|_{L^2(U,\mu)}
+
\left\llbracket u \right\rrbracket_{H^1(U,\mu)}
\leq 
CL \left\| f \right\|_{H^{-1}(U,\mu)}. 
\end{equation}
\end{lemma}
\begin{proof}
Let $f\in H^{-1}(U,\mu)$. A function $u\in H^1(U,\mu )$ is a solution of~\eqref{e.HS.eqn.QL.U} if and only if  
\begin{multline}
\label{e.HSsolbvchar}
\sum_{y\in Q} \sum_{x\in U} \left\langle (\partial_y u(x,\cdot) )(\partial_yw (x,\cdot) ) \right\rangle_{\mu}
+
\sum_{e\in \mathcal{E}(U)} 
\left\langle \nabla u(e,\cdot) \a(e) \nabla w(e,\cdot)  \right\rangle_{\mu} 
\\
=
\sum_{x\in U} \left\langle f(x,\cdot) w(x,\cdot) \right\rangle_{\mu }, \quad \forall w\in H^1_0(U,\mu ). 
\end{multline}
The symmetric bilinear form on the left side of the previous display is coercive with respect to the $H^1(U,\mu )$ norm on the subspace $H^1_0(U,\mu )$, by Lemma~\ref{l.spectralgap.UmuL}. The Lax-Milgram lemma therefore yields the existence of a unique solution~$u\in H^1_0(U,\mu )$. We see that this function satisfies~\eqref{e.HSsolest} by taking $w=u$ in~\eqref{e.HSsolbvchar} and applying part (i) of Lemma~\ref{l.spectralgap.UmuL}. 
\end{proof}

We next give the well-posedness of the Neumann problem in a cube~$Q'$, in which the boundary condition specifies the flux through boundary edges. For this purpose we introduce the set of \emph{boundary edges} $\partial \mathcal{E}(Q')$ of $Q'$, defined by 
\begin{equation}
\partial \mathcal{E}(Q') := \left\{ e\in \mathcal{E}(Q')\,:\, e=(x,y), \, x\in (Q')^\circ, \, y\in \partial Q' \right\}. 
\end{equation}

\begin{lemma}
\label{l.wellposeHS.neu}
Fix a cube~$Q'\subseteq Q_L$ and $\f\in L^2(\mathcal{E}(Q'),\mu)$ 
There exists a solution $u \in H^1(Q',\mu)$ of the boundary-value problem 
\begin{equation}
\label{e.HS.eqn.QL.U.neu}
\left\{ 
\begin{aligned}
& -\L_{\mu}u + \nabla^* \a \nabla u 
 = \nabla^*\f
& \mbox{in} & \ (Q')^\circ \times \Omega , 
\\ & 
\a\nabla u = \f & \mbox{on} & \ \partial \mathcal{E}(Q') \times\Omega ,
\end{aligned}
\right.
\end{equation}
satisfying
\begin{equation}
\label{e.neum.meanzero}
\left( u \right)_{Q',\mu } = 0.
\end{equation}
Moreover, there exists a constant~$C(d,\lambda)<\infty$ such that 
\begin{equation}
\label{e.HSsolest.neu}
L^{-1} \left\| u \right\|_{L^2(Q',\mu)}
+
\left\llbracket u \right\rrbracket_{H^1(Q,'\mu)}
\leq 
CL \left\| \f \right\|_{L^2(\mathcal{E}(Q'),\mu)}. 
\end{equation}
\end{lemma}
\begin{proof}
Let $\f\in L^2(\mathcal{E}(Q'),\mu)$. A function $u\in H^1(Q',\mu)$ is a solution of~\eqref{e.HS.eqn.QL.U.neu} if and only if  
\begin{multline}
\label{e.HSsolbvchar.neu}
\sum_{y\in Q} \sum_{x\in (Q')^\circ} \left\langle (\partial_y u(x,\cdot) )(\partial_yw (x,\cdot) ) \right\rangle_{\mu }
+
\sum_{e\in \mathcal{E}(Q')} 
\left\langle \nabla u(e,\cdot) \a(e) \nabla w(e,\cdot)  \right\rangle_{\mu} 
\\
=
\sum_{x\in \mathcal{E}(Q')} \left\langle \f(e,\cdot) \nabla w(e,\cdot) \right\rangle_{\mu}, \quad \forall w\in H^1(Q',\mu). 
\end{multline}
The symmetric bilinear form on the left side of the previous display is coercive with respect to the $H^1(Q',\mu )$ norm on the closed subspace of $H^1(Q',\mu )$ of functions satisfying~\eqref{e.neum.meanzero}, by part (ii) of Lemma~\ref{l.spectralgap.UmuL}. The Lax-Milgram lemma therefore yields the existence of a unique solution~$u\in H^1(Q',\mu )$ of~\eqref{e.HS.eqn.QL.U.neu},~\eqref{e.neum.meanzero}. We see that this function satisfies~\eqref{e.HSsolest.neu} by taking $w=u$ in~\eqref{e.HSsolbvchar}, using Cauchy's inequality and then applying Lemma~\ref{l.spectralgap.UmuL}. 
\end{proof}

\subsection{Variational characterization, scaling limits and homogenization}
\label{ss.varchar}
The solution of the Dirichlet boundary-value problem~\eqref{e.HS.eqn.QL.U} can also be characterized as the unique minimizer of the variational problem 
\begin{equation}
\label{e.HSsolbvvvchar}
\inf_{w \in H^1_0(U,\mu )} \mathsf{E}_{\mu ,U,f} \left[w\right]
\end{equation}
where we let $\mathsf{E}_{\mu ,U,f} \left[\cdot \right]$ denote the energy functional
\begin{align*}
\mathsf{E}_{\mu,U,f} \left[w\right]
& 
:=
\frac12 \sum_{y\in Q} \sum_{x\in U^\circ} \left\langle (\partial_y w(x,\cdot) )^2 \right\rangle_{\mu }
+
\frac12 \sum_{e\in \mathcal{E}(U)} 
\left\langle \a(e) (\nabla w(e,\cdot))^2  \right\rangle_{\mu } 
\\ & \qquad 
- \sum_{x\in U^\circ} \left\langle f(x,\cdot) w(x,\cdot) \right\rangle_{\mu }.
\end{align*}
Indeed, by a direct computation one checks that~\eqref{e.HSsolbvchar} is the first variation of~\eqref{e.HSsolbvvvchar}. Observe that, if~$u$ is the solution of~\eqref{e.HS.eqn.QL.U}, then by~\eqref{e.HSsolbvchar} we have that 
\begin{align*}
\mathsf{E}_{\mu ,U,f} \left[ u \right]
& =
- \frac12 \sum_{y\in Q} \sum_{x\in U^\circ} \left\langle (\partial_y u(x,\cdot) )^2 \right\rangle_{\mu }
-
\frac12 \sum_{e\in \mathcal{E}(U)} 
\left\langle \a(e) (\nabla u(e,\cdot))^2  \right\rangle_{\mu } 
\\ & 
= - \frac12 \sum_{x\in U^\circ} \left\langle f(x,\cdot) u(x,\cdot) \right\rangle_{\mu }.
\end{align*}
We may therefore give the variance of an element $F\in H^1(\mu_{L,\xi})$ a variational characterization. By the previous display and~\eqref{e.HSrep.muL}, we have
\begin{align}
\label{e.varcharvar}
\left\langle \left( F - \langle F \rangle_{\mu_{L,\xi}} \right)^2 \right\rangle_{\mu_{L,\xi}}
= -2 \mathsf{E}_{\mu_{L,\xi},Q_L,f} \left[ u \right]
= -2 \inf_{w \in H^1_0(Q_L,\mu_{L,\xi})} \mathsf{E}_{\mu_{L,\xi},Q_L,f} \left[w\right],
\end{align}
where $u \in H^1_0(Q_L,\mu_{L,\xi})$ is the unique solution of~\eqref{e.HS.eqn.QL} with $f(x,\cdot)=\partial_xF$. 

\smallskip

This variational characterization of the variance of a random variable~$F$ with respect to the Gibbs measure~$\mu_{L,\xi}$ will prove to be quite useful. Indeed, it is underlying idea behind the proof of the scaling limit of the $\nabla\phi$--model to a Gaussian free field (GFF) in the work of Naddaf and Spencer~\cite{NS}. To see this connection, consider, for $R\geq 1$, the particular random variable
\begin{equation}
F_R(\nabla \phi) := R^{-\frac d2}\sum_{e\in \mathcal{E}} f_i\left(\frac xR\right) \left( \phi(x+e_i)-\phi(x) \right), 
\end{equation}
where, for each $i\in\{1,\ldots,d\}$, we take $f_i: \Rd \to \R$ to be a compactly supported, smooth, deterministic function. We also suppose for convenience that we are working with the infinite-volume Gibbs state~$\mu_{\infty,\xi}$ that we will construct in Section~\ref{ss.infinite}, so that $\nabla \phi$ is defined on all of~$\mathcal{E}$. To prove that 
$R^{\frac d2} \nabla \phi\left( R\cdot \right)$
converges in distribution, as $R\to \infty$ to the gradient of a GFF with covariance matrix~$\ahom$, one needs to show that the random variable~$F_R$ converges in law, as~$R\to \infty$, to a normal random variable with zero mean and variance
\begin{equation}
\label{e.quantttty}
\int_{\Rd} \nabla u \cdot \ahom\nabla u\,dx, 
\end{equation}
where $u$ is the solution of the PDE
\begin{equation}
\label{e.homogttty}
-\nabla \cdot \ahom \nabla u = \nabla \cdot \f \quad \mbox{in} \ \Rd, 
\end{equation}
By an integration by parts and the variational principle for~\eqref{e.homogttty}, we have that 
\begin{equation}
\int_{\Rd} \nabla u \cdot \ahom\nabla u\,dx
=
-2 \overline{\mathsf{E}}_{\Rd,\nabla \cdot \f}\left[u \right]
:=
-2\int_{\Rd} \left( \frac12\nabla u \cdot \ahom\nabla u - \f \cdot \nabla u \right).
\end{equation}
On the other hand, by the infinite-volume analogue of~\eqref{e.varcharvar}, we have that 
\begin{equation}
\var_{\mu_{\infty,\xi}} \left[ F_R \right] 
=
-2\mathsf{E}_{\mu_{\infty,\xi},\Zd,\nabla \cdot \f} \left[u_R\right],
\end{equation}
where $u_R$ is the solution of the Helffer-Sj\"ostrand equation
\begin{equation}
\label{e.HSfullvol}
-\L_{\mu}u + \nabla^* \a \nabla u_R 
 = R^{-\frac d2} \nabla^*\f\left(\frac{\cdot}{R}\right)
\quad \mbox{in}  \ \Zd \times \Omega_\infty.
\end{equation}
One can therefore see the desired convergence of $\var_{\mu_{\infty,\xi}}\left[ F_R \right]$ to the quantity in~\eqref{e.quantttty} as the statement that the energy of the solution of~\eqref{e.HSfullvol} converges to the energy of the solution of~\eqref{e.homogttty}. This is a  manifestation of a more general \emph{homogenization principle} which states roughly that the operator on the left side of~\eqref{e.HSfullvol} ``homogenizes'' to the one on the left side of~\eqref{e.homogttty}. As observed in~\cite{NS}, an appropriate formalization of this homogenization principle is powerful enough to give the full scaling limit of~$\nabla \phi$ under~$\mu_{\infty,\xi}$, that is, the convergence in law, after the scaling above, to a gradient GFF. 

\smallskip

Since there is a strong connection between homogenization of second-order elliptic operators and invariance principles of random walks in random environments (see Section~\ref{ss.generator} below), the idea of Naddaf and Spencer can also be given a natural probabilistic interpretation in terms of the latter, which was subsequently explored in various works (see for instance~\cite{DGI, GOS, Mi} and the references therein).

\smallskip

In the periodic case $\mu=\mu_{L,\xi,\per}$, we can similarly characterize the variance of an observable $F\in H^1(\mu_{L,\xi,\per})$ in terms of the energy of the solution of the Helffer-Sj\"ostrand equation in $Q_L$ with periodic boundary conditions. Denote by $H^1_\per(Q_L,\mu_{L,\xi})$ the space of functions such that $w \in H^1(Q_L,\mu_{L,\xi})$, $w(x_0) =0$ and that $w(x) =w(y) $ if $x-y = 2L\Z$. We have,
\begin{align*}
\left\langle \left( F - \langle F \rangle_{\mu_{L,\xi, \per}} \right)^2 \right\rangle_{\mu_{L,\xi, \per}}
= -2 \mathsf{E}_{\mu_{L,\xi,\per},Q_L,f} \left[ u \right]
= -2 \inf_{w \in H^1_0(Q_L,\mu_{L,\xi,\per})} \mathsf{E}_{\mu_{L,\xi,\per},Q_L,f} \left[w\right],
\end{align*}
The Neumann boundary value problem~\eqref{e.HS.eqn.QL.U.neu} also admits a natural variational interpretation. Indeed, it is easy to check that the solution~$u$ of~\eqref{e.HS.eqn.QL.U.neu} is the unique minimizer of the problem 
\begin{equation}
\label{e.neumannvar}
\inf_{w\in H^1(Q,\mu_{L,\xi}),\, (w)_{Q,\mu_{L,\xi}}=0}
\mathsf{E}_{\mu_{L,\xi},Q,\nabla^*\f} \left[w\right].
\end{equation}

\subsection{Some functional inequalities}
In this subsection we present some basic estimates for the objects introduced above. We begin by observing that Lemma~\ref{l.wellposeHS.dir} implies the following spectral gap inequalities for the measure $\mu_{L,\xi}$ and $\mu_{L,\xi,\per}$. 

\begin{corollary}[{Spectral gap for $\mu_{L,\xi}$}]
\label{c.spectralgap}
There exists $C(d,\lambda)<\infty$ such that, for every $F\in H^1(\mu_{L,\xi})$, 
\begin{equation}
\label{e.spectralgap}
\var_{\mu_{L,\xi}} \left[ F \right] 
\leq 
CL^2 \sum_{x\in Q_{L,\per}} \left\langle (\partial_x F)^2 \right\rangle_{\mu_{L,\xi,\per}}. 
\end{equation}
\end{corollary}
\begin{proof}
Let $f(x,\phi):=\partial_x F(\phi)$ and let $u\in H^1(Q_L,\mu_{L,\xi})$ be the solution of~\eqref{e.HS.eqn.QL.U}. 
According to~\eqref{e.HSrep.muL} and Lemma~\ref{l.wellposeHS.dir}, in particular~\eqref{e.HSsolest},
\begin{align*}
\var_{\mu_{L,\xi}} \left[ F \right]
=
\sum_{x\in Q_L^\circ} 
\left\langle f(x,\cdot) u(x,\cdot)  \right\rangle_{\mu_{L,\xi}}
& 
\leq
\left\| u \right\|_{H^1(Q_L,\mu_{L,\xi})} \left\| f \right\|_{H^{-1}(Q_L,\mu_{L,\xi})}
\\ & 
\leq 
CL^2 \left\| f \right\|_{H^{-1}(Q_L,\mu_{L,\xi})}^2
\\ & 
\leq CL^2 \sum_{x\in Q_L^\circ} \left\langle (\partial_x F)^2 \right\rangle_{\mu_{L,\xi}}. \qedhere
\end{align*}
\end{proof}

\begin{corollary}[{Spectral gap for $\mu_{L,\xi,\per}$}]
\label{c.spectralgap2}
There exists $C(d,\lambda)<\infty$ such that, for every $F\in H^1(\mu_{L,\xi,\per})$, 
\begin{equation*}
\var_{\mu_{L,\xi,\per}} \left[ F \right] 
\leq 
CL^2 \sum_{x\in Q_L^\circ} \left\langle (\partial_x F)^2 \right\rangle_{\mu_{L,\xi,\per}}. 
\end{equation*}
\end{corollary}

We next present the Brascamp-Lieb inequality \cite{BL}, which is a shaper version of the previous lemma. The  proof we give is essentially the same as the one sketched in~\cite{NS}. We denote the Green function for the discrete Laplacian with zero Dirichlet boundary conditions in~$Q_L$ by $G_{Q_L}(x,y)$. 

\begin{proposition}[Brascamp-Lieb inequality for $\mu_{L,\xi}$]
\label{p.BL} 
For every $F\in H^1(\mu_{L,\xi})$,
\begin{equation}
\label{e.BL.var}
\var_{\mu_{L,\xi}} \left[ F \right] 
\leq 
\frac1\lambda 
\sum_{x,y \in Q_L^\circ}G_{Q_L}(x,y) \left\langle \left( \partial
_{x}F\right) \left( \partial _{y}F\right) \right\rangle_{\mu_{L,\xi}}.
\end{equation}
\item
For every $\psi :Q_L\to \R$, we have 
\begin{equation}
\label{e.BL.linexp}
\log \left \langle  
\exp \left( t \sum_{y\in Q_L} \phi(y)\psi (y) \right) 
\right\rangle_{\mu_{L,\xi}} 
\leq 
\frac1\lambda t^2 \sum_{x,y \in Q_L^\circ}
G_{Q_L}(x,y) \psi(x)\psi(y).
\end{equation}
\end{proposition}
\begin{proof}
\emph{Step 1.} The proof of~\eqref{e.BL.var}. 
Let $F\in H^1(\mu_{L,\xi})$. Denote 
\begin{equation}
f(x,\phi):=\partial_x F(\phi) \in L^2(Q_L,\mu_{L,\xi})
\end{equation}
By Lemma~\ref{l.wellposeHS.dir}, there exists a solution $u\in H^1_0(Q_L,\mu_{L,\xi})$ of the equation
\begin{equation*}
\left\{ 
\begin{aligned}
& -\Delta_\phi u  + \nabla^* \cdot \a \nabla  u = f
& \mbox{in} 
& \ Q^\circ_L \times \Omega_{Q_L}, \\
& u = 0
& \mbox{on} 
& \ \partial Q_L \times \Omega_{Q_L}.
\end{aligned} 
\right .
\end{equation*}
By~\eqref{e.HSrep.muL}, 
\begin{equation}
\label{e.varident.BL}
\var_{\mu_{L,\xi} }\left[ F \right] 
=
\sum_{x\in Q_L^\circ}\left\langle
u\left( x,\cdot \right) f(x,\cdot) 
\right\rangle _{\mu_{L,\xi}}.
\end{equation}
For each fixed~$\phi\in\Omega$, let $w(\cdot,\phi):= \lambda^{-1} \Delta_L^{-1} f(\cdot,\phi)$, that is, $w$ is the solution of the problem  
\begin{equation*}
\left\{ 
\begin{aligned}
& -\lambda \Delta w = f(\cdot,\phi)
& \mbox{in} 
& \ Q^\circ_L, \\
& w = 0
& \mbox{on} 
& \ \partial Q_L.
\end{aligned} 
\right .
\end{equation*}
Using the equations, and an energy comparison, we have that 
\begin{align*}
\lefteqn{
-\lambda 
\sum_{e\in \mathcal{E}(Q_L)} 
\left\langle \left( \nabla w(e,\cdot) \right)^2  
\right\rangle_{\mu_{L,\xi}}
} \  & 
\\ &
=
\lambda 
\sum_{e\in \mathcal{E}(Q_L)} 
\left\langle \left( \nabla w(e,\cdot) \right)^2 \right\rangle_{\mu_{L,\xi}}
-2\sum_{x\in Q_L^\circ} \left\langle w(x,\cdot) f(x,\cdot)  \right\rangle_{\mu_{L,\xi}}
\\ & 
\leq 
\lambda 
\sum_{e\in \mathcal{E}(Q_L)} 
\left\langle \left( \nabla u(e,\cdot) \right)^2  \right\rangle_{\mu_{L,\xi}}
-2\sum_{x\in Q_L^\circ} \left\langle u(x,\cdot) f(x,\cdot)  \right\rangle_{\mu_{L,\xi}}
\\ & 
\leq 
\sum_{y\in Q_L} 
\sum_{x\in Q_L^\circ}
\left\langle \left( \partial_y u(x,\cdot)\right)^2 \right\rangle_{\mu} 
+\sum_{e\in \mathcal{E}(Q)} 
\left\langle  \a(e,\cdot) (\nabla u(e,\cdot) )^2
\right\rangle_\mu
-
2\sum_{x\in Q_L^\circ}\left\langle u(x,\cdot) f(x,\cdot)  \right\rangle_{\mu_{L,\xi}}
\\ &
=
- \sum_{x\in Q_L^\circ}  \left\langle u(x,\cdot) f(x,\cdot)  \right\rangle_{\mu_{L,\xi}}.
\end{align*}
We deduce that 
\begin{align*}
\var_{\mu_{L,\xi} }\left[ F\right] 
&
=
\sum_{x\in Q_L^\circ} \left\langle u(x,\cdot) f(x,\cdot)  \right\rangle_{\mu_{L,\xi}}
\\ & 
\leq 
\lambda
\sum_{e\in \mathcal{E}(Q_L)}
\left\langle
\left( \nabla w(e,\cdot) \right)^2  
\right\rangle_{\mu_{L,\xi}}
=
\sum_{x\in Q_L^\circ} \left\langle w(x,\cdot) f(x,\cdot)  \right\rangle_{\mu_{L,\xi}}.
\end{align*}
This is~\eqref{e.BL.var}. 

\smallskip
\emph{Step 2.} We prove~\eqref{e.BL.linexp}
This follows from~\eqref{e.BL.var} by differentiating the quantity 
\begin{equation}
\label{e.difflog}
\frac{\partial^2}{\partial t^2} \log \left\langle \exp \left( t\sum_{y\in Q_L} \phi(y)\psi(y) \right)\right\rangle_{\mu_{L,\xi} }
= 
\var_{\mu_t}\left[\sum_{y\in Q_L} \phi(y)\psi(y) \right],
\end{equation}
where $\mu_t$ denotes the tilted Gibbs measure with Hamiltonian 
\begin{equation*}
H_{L,t} := H_L+ t\sum_{y\in Q_L} \phi(y)\psi(y).
\end{equation*}
Note that $D^2 H_{L,t} = D^2 H_L \geq \lambda \Delta_L.$ Analogues of the Poincare inequality (Lemma \ref{l.spectralgap.UmuL}) , the solvability of the Helffer-Sj\"ostrand equation (Lemma \ref{l.wellposeHS.dir}), and the variance estimate \eqref{e.BL.var} for measure $\mu_t$ in place of $\mu$, and can be obtained without any changes to the arguments.
The claim thus follows from integrating \eqref{e.BL.var} for $\mu_t$. 
\end{proof}

Let~$G_{Q_L,\per}$ denote the Green function for the discrete Laplacian with periodic boundary conditions in~$Q_L$ and zero boundary condition at $x_0$. We have a similar version of Brascamp-Lieb inequality for $\mu=\mu_{L,\xi,\per}$.

\begin{proposition}
\label{p.BL.per} 
For every $F\in H^1(\mu_{L,\xi,\per})$,
\begin{equation}
\var_{\mu_{L,\xi,\per}} \left[ F \right] 
\leq 
\frac1\lambda 
\sum_{x,y \in Q_{L,\per}}(G_{Q_{L,\per}}(x,y) \left\langle \left( \partial
_{x}F\right) \left( \partial _{y}F\right) \right\rangle_{\mu_{L,\xi,\per}}.
\end{equation}
\item
For every $\psi :Q_{L,\per}\to \R$, we have 
\begin{equation}
\label{e.BL.linexpL}
\log \left \langle  
\exp \left( t \sum_{y\in Q_{L,\per}} \phi(y)\psi (y) \right) 
\right\rangle_{\mu_{L,\xi,\per}} 
\leq 
\frac1\lambda t^2 \sum_{x,y \in Q_{L,\per}}
G_{Q_{L,\per}}(x,y) \psi(x)\psi(y).
\end{equation}
\end{proposition}

\smallskip

We next present a version of the elliptic Caccioppolli inequality.

\begin{lemma}
[Caccioppoli inequality]
\label{l.cacc}
There exists $C(\data)<\infty$ such that, for every $M\in\N$ with $2\leq M$ and $2M\leq L$ and every $u\in H^{1}\left( Q_{2M},\mu  \right)$ and $f\in H^{-1}(Q_{2M},\mu )$ satisfying
\begin{equation}
\label{l.caccpde}
-\L_{\mu } u  + \nabla^* \a \nabla  u = f
\quad \mbox{in} \ Q_{2M} \times \Omega ,
\end{equation}
we have the estimate
\begin{equation}
\label{e.cacc}
\left\llbracket u \right\rrbracket_{H^1(Q_M,\mu )}  
\leq C \left( \frac1{M} \left\| u-\left( u\right)_{Q_{2M}}\right\|
_{L^{2}\left( Q_{2M},\mu \right) }
+ 
\left\| f \right\|_{H^{-1}(Q_{2M},\mu )}
\right).
\end{equation}
\end{lemma}
\begin{proof}
The proof is almost the same as that of the standard Caccioppolli inequality, the main difference being the discrete notation. By subtracting a constant from~$u$ we may suppose that~$(u)_{Q_{2M}}=0$. Fix a cutoff function~$\eta \in C_{c}^{\infty }\left( \Rd\right)$ such that
\begin{equation}
\label{e.cutoffxi}
0\leq \eta \leq 1, \quad 
\eta \equiv 1 \ \ \mbox{on} \ Q_M, \quad
\eta \equiv 0 \ \ \mbox{on} \ \partial Q_{2M} 
\end{equation}
and, for every $x,y\in Q_{2M}$ with $x\sim y$, 
\begin{equation}
\label{e.cutoffxi2}
(\eta(x) - \eta(y))^2 \leq CM^{-2} \left( \eta(x) + \eta(y) \right).
\end{equation}
It suffices to take, for instance, $\eta:= \tilde{\eta}^2$ for any $\tilde{\eta}$ satisfying~\eqref{e.cutoffxi}.
Testing~\eqref{l.caccpde} with $\eta u$ yields
\begin{align*}
&
\frac12\sum_{x,y\in Q_{2M},\, x\sim y} 
\left\langle \a(x,y) (u(x,\cdot) - u(y,\cdot)) (\eta(x) u(x,\cdot) - \eta(y) u(y,\cdot))\right\rangle_{\mu }
\\ & \quad 
+ 
\sum_{y\in Q } \sum_{x\in Q_{2M}} 
\eta(x)\left\langle (\partial_y u(x,\cdot) )(\partial_yu (x,\cdot) ) \right\rangle_{\mu }
-\sum_{x\in Q_{2M}} \left\langle f(x,\cdot) \eta(x)u(x,\cdot) \right\rangle_{\mu }
=
0.
\end{align*}
Observe that 
\begin{align*}
\lefteqn{
\frac12\sum_{x,y\in Q_{2M},\, x\sim y} 
\left\langle \a(x,y) (u(x,\cdot) - u(y,\cdot)) (\eta(x) u(x,\cdot) - \eta(y) u(y,\cdot))\right\rangle_{\mu } 
} \quad & 
\\ &
=
\sum_{x,y\in Q_{2M},\, x\sim y} 
\left( \eta(x) + \eta(y) \right) \left\langle \a(x,y) (u(x,\cdot) - u(y,\cdot))^2\right\rangle_{\mu } 
\\ & \quad 
+
\frac12\sum_{x,y\in Q_{2M},\, x\sim y} 
\left\langle u(y,\cdot) \a(x,y) (u(x,\cdot) - u(y,\cdot)) 
\left( \eta(x) - \eta(y) \right)
\right\rangle_{\mu } 
\\ & 
\geq 
\frac12 \sum_{x,y\in Q_{2M},\, x\sim y} 
\left( \eta(x) + \eta(y) \right) \left\langle \a(x,y) (u(x,\cdot) - u(y,\cdot))^2\right\rangle_{\mu } 
\\ & \quad 
-
C\sum_{x,y\in Q_{2M},\, x\sim y} \frac{\left( \eta(x) - \eta(y) \right)^2}{\eta(x) + \eta(y)}  u(y)^2\a(x,y).
\end{align*}
Combining the previous two displays and using~\eqref{e.cutoffxi} and~\eqref{e.cutoffxi2}, we get
\begin{align*}
\left\llbracket u \right\rrbracket_{H^1(Q_M,\mu )} 
& 
\leq 
\sum_{y\in Q } \sum_{x\in Q_{2M}} 
\eta(x)\left\langle (\partial_y u(x,\cdot) )(\partial_yu (x,\cdot) ) \right\rangle_{\mu }
\\ & \quad
+
\frac12 \sum_{x,y\in Q_{2M},\, x\sim y} 
\left( \eta(x) + \eta(y) \right) \left\langle \a(x,y) (u(x,\cdot) - u(y,\cdot))^2\right\rangle_{\mu } 
\\ & 
\leq
\sum_{x\in Q_{2M}} \left\langle f(x,\cdot) \eta(x)u(x,\cdot) \right\rangle_{\mu }
+
C\sum_{x,y\in Q_{2M},\, x\sim y} \frac{\left( \eta(x) - \eta(y) \right)^2}{\eta(x) + \eta(y)}  u(y)^2
\\ & 
\leq 
C\left\| u \right\|_{H^1(Q_{2M},\mu )} \left\| f \right\|_{H^{-1}(Q_{2M},\mu )}
+
C M^{-2} \left\| u \right\|_{L^2(Q_{2M})}.
\end{align*}
Since $u$ has mean zero, we have by Lemma~\ref{l.spectralgap.UmuL} that $\left\| u \right\|_{H^1(Q_{2M},\mu )} \leq C \left\llbracket u \right\rrbracket_{H^1(Q_{2M},\mu )}$. Therefore the previous display implies~\eqref{e.cacc}.
\end{proof}

\subsection{Some special estimates for $\mu_{L,\xi,\per}$}

\smallskip

We first present an estimate on the distributional tail of the field~$\phi$ sampled by the Gibbs measure $\mu_{L,\xi,\per}$.

\begin{lemma}[Oscillation estimate, periodic fields]
\label{l.oscillation.per}
There exists~$C(\data)<\infty$, such that, for every $s \geq  C$, $\xi\in\Rd$ and $L\in\N$,
\begin{equation}
\label{e.oscillation.per}
\mu_{L,\xi,\per} \left( \left\{ 
\phi\in \Omega_\per(Q_L)\,:\,
\max_{x\in Q_L} 
| \phi(x)| 
> C s \log L
\right\} \right) 
\leq
\exp\left( -s^2 \log L \right).
\end{equation}
\end{lemma}
\begin{proof}
We will prove~\eqref{e.oscillation.per} by estimating the exponential moments of $|\phi(x)|$ for each $x\in Q_L$ and then take a union bound over~$x$. 
In view of~\eqref{e.periodic.meanzero}, 
it suffices to bound the exponential moments of $\phi(x)$, which we do by an application of the Brascamp-Lieb inequality Proposition \ref{p.BL.per}. Let~$G_{Q_L,\per}$ denote the Green function for the discrete Laplacian with periodic boundary conditions in~$Q_L$ and zero boundary condition at $x_0$.
Then we obtain, for a constant~$C(\lambda)<\infty$, and all $s \in \R$
\begin{align*}
\max_{x\in Q_L} \left\langle \exp (s \phi(x)) \right \rangle_{\mu_{L,\xi,\per}} 
&
\leq
\exp\left( \frac{s^2}{2\lambda} \max_{x\in Q_L} \left ( G_{Q_L,\per} (x,x) \right) \right) \\
&\leq \exp(Cs^2 \log L).
\end{align*}
Applying the Chebyshev inequality and optimize over $s$, we obtain, for a constant~$C_1(\lambda)< \infty$ and every $s >0$,
\begin{equation*}
\max_{x\in Q_L} \mu_{L,\xi,\per} \left\{ \phi (x) > C_1 s \log L \right\}
\leq \exp \left(- s^2 \log L \right).
\end{equation*}
The claim follows by taking a union bound over all $x$. 
\end{proof}

We denote by $ \P'_{L,\xi,\per,\phi}$ the law of the Markov process \eqref{e.dynamics.phi.QLper}, with initial condition $\phi_{\per,0} =\phi$. In what follows we consider the stationary Langevin dynamics, that is, we sample the initial condition with $\mu_{L,\xi,\per}$, so that the law of $\phi_{\per,t}$ is given by $\mu_{L,\xi,\per} \otimes \P'_{L,\xi,\per,\phi}$.
We next give an estimate on the oscillations of the dynamical field $\phi_{\per,t}$. 

\begin{lemma}
\label{l.oscillation.per.dyn}
Let $R\in [1,\infty)$. There exist $C(R,\data)<\infty$ and $L_0(R,\data)<\infty$ such that, for every $T,s\in (1,\infty)$, $\xi\in B_R$ and $L \geq L_0$,
\begin{multline}
\left( \mu_{L,\xi,\per} \otimes \P'_{L,\xi,\per,\phi}  \right) 
\left[
\max_{(t,x) \in (0,T] \times Q_L} 
| \phi_{\per,t}(x)| 
> C s \left(  \log (LT) \right)
\right]
\\
\leq
\exp\left( -s^2 \left(  \log (LT) \right)  
\right).
\end{multline}
\end{lemma}
\begin{proof}
Take $\xi \in B_R$. Since the time parameter is continuous, we prove the claim in two steps. First we discretize the time into intervals of length $(\log L)^{-1}$, and define the corresponding comb set by $\mathcal{C} := \{(t,x)\in  (0,T] \times Q_L, t\log L \in \Z \}$. A union bound over the tail estimate proved in Lemma \ref{l.oscillation.per} controls the maximum of $\phi_{\per}$ over $(t,x)\in \mathcal{C}$. Then we use continuity of the Brownian motion to bound $\phi_{\per,t}(x) - \phi_{\per,t_0}(x)$, whenever $|t-t_0| <(\log L)^{-1}$.

\smallskip

We first discuss the continuity estimates in $t$. The dynamics \eqref{e.dynamics.phi.QLper} imply, for every~$e = (x,y) \in \mathcal{E}(Q_L)$, 
\begin{align*}
& d\nabla \phi_{\per,t}(e) 
\\ & \quad
= - \left(\sum_{e\ni y}\mathsf{V}'( -\nabla\ell_\xi(e)+ \nabla \phi_{\per,t}(e)) - \sum_{e\ni x}\mathsf{V}'( -\nabla\ell_\xi(e)+ \nabla \phi_{\per,t}(e)) \right)  \,dt 
+ 2 \,dB_t(e),
\end{align*}
where $B_t(e): = \frac{1}{\sqrt{2}} (B_t(y) - B_t(x))$ is a standard Brownian motion. Let $G_t := \max_{e\in \mathcal{E}(Q_L)}| \nabla \phi_{\per,t}(e)| $ and $M :=\max_{e\in\mathcal{E}(Q_L)} \max_{t\in (0, (\log L)^{-1}]} B_t(e) $, the boundedness of $\mathsf{V}''$ implies that 
\begin{equation*}
G_t \leq 4d\Lambda \int_0^t ( G_s+ |\xi|) \,ds + 2 M.
\end{equation*}
Choose $L$ large enough such that $\Lambda R< \log L$, we apply Gronwall inequality to obtain for $t\in (0, (\log L)^{-1}]$ 
\begin{equation*}
G_t \leq 2(M+1) + 8d\Lambda \int_0^t (M+1) \exp\left( 4d\Lambda(t-s) \right) \,ds.
\end{equation*}
That is, 
\begin{equation*}
G_t \leq C (M+1)\exp\left( 4d\Lambda t\right).
\end{equation*}
We now bound $\phi_{\per,t}$ by a comparison with independent Brownian motions. Denote by $\Psi_{\per,t} := \phi_{\per,t} - (\sqrt{2} B_t +\phi_{\per,0} )$. We then have for all $x \in Q_L$, 
\begin{align*}
d\Psi_{\per,t} (x) &\leq 2d\max_{e\ni x} \mathsf{V}'( -\nabla\ell_\xi(e)+ \nabla \phi_{\per,t}(e)) \\
&\leq 2d\Lambda \max_{e\ni x} (| \nabla \phi_{\per,t}(e)| + |\xi|)
\\ &
\leq 2d \Lambda G_t + 2d\Lambda |\xi| 
\leq 
C(M+1) \exp\left( 4d\Lambda t\right)+ C\Lambda R.
\end{align*}
Integrating over $ t\in (0, (\log L)^{-1}]$, we have the following inequality in law:
\begin{equation*}
\max_{t\in (0, (\log L)^{-1}]} |\Psi_{\per,t} (x)| \leq C(M + R+1).
\end{equation*}

We are now ready to finish the proof of the Lemma. Given $t\in (0,T]$, take $t^* \in \frac{1}{\log L}\Z$ such that $t-t^* \in (0,(\log L)^{-1}] $. Using the stationarity of $\phi_{\per,t}$  in time, we have the following inequalities in law:
\begin{align}
\label{e.cont}
\lefteqn{
\max_{(t,x) \in (0,T] \times Q_L} 
| \phi_{\per,t}(x)| 
} \quad & 
\\ & \notag
\leq \max_{(t,x) \in \mathcal{C}} | \phi_{\per,t}(x)|  + \max_{(t,x) \in (0,T] \times Q_L} |\phi_{\per,t}(x)- \phi_{\per,t^*}(x)| 
\\ & \notag
\leq  \max_{(t,x) \in \mathcal{C}} | \phi_{\per,t}(x)| + \max_{(t^*,x) \in \mathcal{C}} \max_{t\in (0, (\log L)^{-1}]}|\phi_{\per,t+t^*}(x)- \phi_{\per,t^*}(x)| 
\\ & \notag
\leq  \max_{(t,x) \in \mathcal{C}} | \phi_{\per,t}(x)| + \max_{(t^*,x) \in \mathcal{C}} \max_{t\in (0, (\log L)^{-1}]} |\Psi_{\per,t} (x)|  
+ 2\max_{(t^*,x) \in \mathcal{C}} \max_{t\in (0, (\log L)^{-1}]} |B_t(x)| \notag
\\ & \notag 
\leq  \max_{(t,x) \in \mathcal{C}} | \phi_{\per,t}(x)| + C(M +R+1)+ 2\max_{(t^*,x) \in \mathcal{C}} \max_{t\in (0, (\log L)^{-1}]} |B_t(x)|.
\end{align}
Applying Lemma~\ref{l.oscillation.per} and taking a union bound over~$t\in (\log L)^{-1}\Z$  we find, for~$L > L_0(R, \data)$ 
\begin{multline*}
\left( \mu_{L,\xi,\per} \otimes \P'_{L,\xi,\per,\phi}  \right) 
\left[  
\max_{(t,x) \in \mathcal{C}} | \phi_{\per,t}(x)|
> C s \left(  \log (LT) \right)
\right]
\\
\leq T\log L \exp\left(-s^2 \log(LT)\right) 
\leq
\exp\left(-\frac{s^2}{2} \log(LT)\right).
\end{multline*}
Applying a union bound and then Doob's inequality, we obtain 
\begin{align*}
\lefteqn{ 
\left( \mu_{L,\xi,\per} \otimes \P'_{L,\xi,\per,\phi}  \right) \left[ M > s \log(LT) \right]
} \qquad & 
\\ & 
\leq |Q_L| \left( \mu_{L,\xi,\per} \otimes \P'_{L,\xi,\per,\phi}  \right) \left[ \max_{t\in (0, (\log L)^{-1}]}B_t(0) \geq \log L \right]
\\ & 
\leq 
|Q_L| \exp \left(-\frac12{(\log L)^3}\right) \leq \exp \left(-\frac13{(\log L)^3}\right). 
\end{align*}
Taking a union bound over $(t^*,x)\in \mathcal{C}$ then yields
\begin{equation*}
\left( \mu_{L,\xi,\per} \otimes \P'_{L,\xi,\per,\phi}  \right) \left[ \max_{(t^*,x) \in \mathcal{C}} \max_{t\in (0, (\log L)^{-1}]} |B_t(x)| > Cs \log(LT) \right] \leq\exp \left(-\frac14 (\log L)^3\right). 
\end{equation*}
Combining \eqref{e.cont} with the last three inequalities we conclude the lemma. 
\end{proof}

\subsection{The Helffer-Sj\"ostrand operator as a generator}
\label{ss.generator}

The operator
\begin{equation*}
\L_{\mu_{L,\xi}} - \nabla^* \a \nabla
\end{equation*}
on the left side of~\eqref{e.HS.eqn.QL.U} is the infinitesimal generator of a Markov process on the state space~$\Omega_0(Q_L) \times Q_L$. We let $\phi_t$ evolve according to~\eqref{e.dynamics.phi.QL} and augment it with the continuous-time random walk on~$Q_L$, denoted by $\{X_t\}_{t\geq 0}$, with the time dependent jump rate~$\a(t,e):= \mathsf{V}''(\nabla\phi_t(e)- \nabla \ell_\xi)$ along a edge~$e\in \mathcal{E}(\Zd)$. Then~$\{ (X_t,\phi_t) \}_{t\geq 0}$ is a Markov process on~$\Omega_0(Q_L) \times Q_L$ and its infinitesimal generator is precisely the Helffer-Sj\"ostrand operator~$\L_{\mu_{L,\xi}} - \nabla^* \a \nabla$. 

\smallskip

Given a trajectory~$\{ \phi_t\}_{t\geq 0}$, we may also view the random walk $\{X_t\}$ as a Markov process on $\Zd$ with infinitesimal generator 
$\nabla^*\a_{\{ \phi_\cdot\}} \nabla$. (We sometimes write $\a_{\{\phi_\cdot\}}(t,e)$ in place of $\a(t,e)$ if we wish to emphasize the dependence of $\a(t,e)$ on $\{\phi_t\}$.) The equivalence of these two points of view---namely, thinking of~$(X_t,\phi_t)$ as a Markov process, or alternatively thinking of $\phi_t$ as a Markov process and $X_t$ as a second Markov process which depends on $\phi_t$---gives us a convenient way to represent solutions of~\eqref{e.HS.eqn.QL.U}.

\smallskip

Let us now fix some notation. For each~$(x,\phi)\in Q_L\times \Omega_0(Q_L)$, we denote by~$\E_{L,x,\phi}$ the expectation with respect to the law of the Markov process~$(X_t,\phi_t)$, described above, with $(X_0,\phi_0)=(x,\phi)$. For each~$\phi\in \Omega_0(Q_L)$, we let~$\E'_{L,\phi}$ be the expectation with respect to the law of process $(\phi_t)$ starting from $\phi_0=\phi$. Finally, for every $x\in Q_L$ and trajectory $\{ \phi_t\}_{t\geq 0}$, which a continuous function from $(0,\infty) \to \Omega_0(Q_L)$, we let $\E''_{L,x,\{\phi_\cdot\}}$ be the expectation of the Markov process $(X_t)$ given $(\phi_t)$ and $X_0=x$. By the discussion in the previous paragraphs, it is clear that
\begin{equation}
\label{e.semidirectproduct}
\E_{L,x,\phi} = \E'_{L,\phi} \otimes \E''_{L,x,\{\phi_\cdot\}},
\end{equation}
where the $\otimes$ denotes the semidirect product. 

\begin{lemma}[Representation of the Dirichlet problem]
\label{l.represent}
Let $U \subseteq Q_L$. 
Assume that $F\in L^2(U,\mu_{L,\xi})$ is smooth. The solution~$v$ of the Dirichlet problem 
\begin{equation*}
\left\{
\begin{aligned}
& -\mathcal{L}_{\mu_{L,\xi}} v + \nabla^*\a\nabla v = F & \mbox{in} & \ U^\circ \times\Omega_0(Q_L), \\
& v = 0  & \mbox{on} & \  \partial U \times\Omega_0(Q_L),
\end{aligned}
\right.
\end{equation*}
is given by the formula 
\begin{equation}
\label{e.represent}
v(x,\phi) = \E'_{L,\phi} \left[ \int_0^\infty w\left(t,x;\{\phi_\cdot\} \right) \,dt \right] 
\end{equation}
where, for each trajectory $\{\phi_t\}$, the function $w(\cdot;\{ \phi_\cdot\})$ is the solution of the parabolic initial-value problem 
\begin{equation}
\label{e.parabolic.freezephi}
\left\{
\begin{aligned}
& \partial_t w + \nabla^*\a_{\{\phi_\cdot\}}\nabla w = 0 & \mbox{in} & \ (0,\infty) \times U^\circ,\\
& w = 0 & \mbox{on} & \ (0,\infty) \times \partial U,\\
& w = F(\cdot,\phi) & \mbox{on} & \ \{ 0\} \times U^\circ.
\end{aligned}
\right.
\end{equation}
\end{lemma}
\begin{proof}
We begin with the observation that~$v$ admits the following (Feynman-Kac-type) stochastic representation formula:  
\begin{equation}
\label{e.FK.dir}
v(x,\phi) = \E_{L,x,\phi} \left[ \int_0^{\tau_{\partial U}} F(X_s,\phi_s)\,ds \right],
\quad (x,\phi) \in U \times \Omega_0(Q_L)
\end{equation}
where~$\tau_{\partial U}:= \inf \left\{ t>0\,:\, X_t \in \partial U\right\}$ is the stopping time for the process $(X_t,\phi_t)$ to hit~$\partial U \times\Omega_0(Q_L)$. To see that~\eqref{e.FK.dir} is valid, we denote
\begin{equation}
V(t,x,\phi) := 
\E_{L,x,\phi} \left[  F(X_t,\phi_t) \indc_{\{t<\tau_{\partial U}\}}\right],
\quad (t,x,\phi) \in (0,\infty)\times U \times \Omega_0(Q_L)
\end{equation}
and observe immediately from the fact that $\mathcal{L}_{\mu_{L,\xi}} -\nabla^*\a\nabla$ is the generator of the process~$(X_s,\phi_s)$ that~$V$ is the solution of the parabolic problem 
\begin{equation}
\label{e.parabolic.V}
\left\{
\begin{aligned}
& \partial_t V -\mathcal{L}_{\mu_{L,\xi}}V+ \nabla^*\a \nabla V = 0 & \mbox{in} & \ (0,\infty) \times U^\circ \times \Omega_0(Q_L),\\
& V = 0 & \mbox{on} & \ (0,\infty) \times \partial U \times \Omega_0(Q_L),\\
& V = F & \mbox{on} & \ \{ 0\} \times U^\circ \times \Omega_0(Q_L).
\end{aligned}
\right.
\end{equation}
Therefore~\eqref{e.FK.dir} follows from Duhamel's principle, which asserts that 
\begin{equation}
v(x,\phi) = \int_0^\infty V(s,x,\phi)\,ds. 
\end{equation}
Combining now~\eqref{e.semidirectproduct} and~\eqref{e.FK.dir}, we obtain
\begin{equation*}
v(x,\phi) = \E_{L,\phi}' \left[ \E''_{L,x,\{ \phi_\cdot\}}  \left[ \int_0^{\infty} F(X_s,\phi_s) \indc_{\{s<\tau_{\partial U} \}}\,ds \right] \right].
\end{equation*}
The solution~$w(\cdot;\{\phi_\cdot\})$ of~\eqref{e.parabolic.freezephi} is given by
\begin{equation*}
w(t,x;\{ \phi_\cdot\}) = \E''_{L,x,\{ \phi_\cdot\}}  \left[ F(X_t,\phi_t) \indc_{\{t<\tau_{\partial U} \}} \right].
\end{equation*}
The previous two displays yield~\eqref{e.represent}. 
\end{proof}

\begin{lemma}[Representation of the Neumann problem]
\label{l.representN}
Given a cube $Q' \subseteq Q_L$. 
Assume that $\f\in L^2(\mathcal{E}(Q'),\mu_{L,\xi})$ is smooth. The solution~$v$ of the Neumann problem 
\begin{equation*}
\left\{
\begin{aligned}
& -\mathcal{L}_{\mu_{L,\xi}} v + \nabla^*\a\nabla v = \nabla^*\f & \mbox{in} & \ (Q')^\circ \times\Omega_0(Q_L), \\
& \a \nabla v -\f = \nabla \ell_q  & \mbox{on} & \  \partial \mathcal{E}(Q') \times\Omega_0(Q_L),
\end{aligned}
\right.
\end{equation*}
is given by the formula 
\begin{equation}
\label{e.representN}
v(x,\phi) = \E'_{L,\phi} \left[ \int_0^\infty w\left(t,x;\{\phi_\cdot\} \right) \,dt \right] 
\end{equation}
where, for each trajectory $\{\phi_t\}$, the function $w(\cdot;\{ \phi_\cdot\})$ is the solution of the parabolic initial-value problem 
\begin{equation}
\label{e.parabolic.freezephiN}
\left\{
\begin{aligned}
& \partial_t w + \nabla^*\a_{\{\phi_\cdot\}}\nabla w = 0 & \mbox{in} & \ (0,\infty) \times (Q')^\circ,\\
& \a_{\{\phi_\cdot\}} \nabla w  = 0 & \mbox{on} & \ (0,\infty) \times \partial \mathcal{E}(Q'),\\
& w = \nabla^*\f & \mbox{on} & \ \{ 0\} \times (Q')^\circ, \\
& \a_{\{\phi_\cdot\}} \nabla w -\f = \nabla \ell_q  & \mbox{on} & \ \{ 0\} \times\partial \mathcal{E}(Q').
\end{aligned}
\right.
\end{equation}
\end{lemma}
\begin{proof}
Let $F:Q' \times \Omega_0(Q_L) \to \R$ denote the function satisfying~$F(x,\phi) = \nabla^*\f(x,\phi)$ for each $x\in (Q')^\circ$ and $\a(e,\phi)\nabla F(e,\phi) = \f(e,\phi) +q$ if $e\in \partial \mathcal{E}(Q')$. We have that  
\begin{equation}
w\left(t,x,\{\phi_{\cdot}\}\right):=
\E_{L,x,\{\phi_\cdot\}}'' \left[
F(X_t,\phi_t)
\right].
\end{equation}
We then observe that we have, for any $s\in (0,\infty)$,  
\begin{equation}
V(s,x,\phi)
=
\int_0^s
\E_{L,\phi}' \left[ 
\E_{L,x,\{\phi_\cdot\}}'' \left[
F(X_t,\phi_t)
\right]
\right]
\,dt
=
\E_{L,\phi}' \left[ 
\int_0^s 
w\left(s,x,\{\phi_{\cdot}\}\right)\,dt
\right],
\end{equation}
where $V$ is the solution of the parabolic problem 
\begin{equation*}
\left\{
\begin{aligned}
& \partial_t V -\mathcal{L}_{\mu_{L,\xi}} V + \nabla^*\a\nabla V = \nabla^*\f & \mbox{in} & \ (0,\infty) \times (Q')^\circ \times\Omega_0(Q_L), \\
& \a \nabla V -\f = \nabla \ell_q  & \mbox{on} & \ (0,\infty) \times  \partial \mathcal{E}(Q') \times\Omega_0(Q_L),\\
& V = 0 & \mbox{on} & \ \{0\} \times (Q') \times\Omega_0(Q_L).
\end{aligned}
\right.
\end{equation*}
The proof will be complete once we show that $v(x,\phi) = \lim_{s\to \infty} V(s,x,\phi)$, in the sense of $L^2(Q',\mu_{L,\xi})$. To prove this, we consider the difference $\tilde{V}(x,s,\phi) := v(x,\phi) - V(s,x,\phi)$ and observe that this satisfies
\begin{equation*}
\left\{
\begin{aligned}
& \partial_t \tilde{V} -\mathcal{L}_{\mu_{L,\xi}} \tilde{V} + \nabla^*\a\nabla \tilde{V} = 0 & \mbox{in} & \ (0,\infty) \times (Q')^\circ \times\Omega_0(Q_L), \\
& \a \nabla \tilde{V}  = 0  & \mbox{on} & \ (0,\infty) \times  \partial \mathcal{E}(Q') \times\Omega_0(Q_L),\\
& \tilde{V} = v & \mbox{on} & \ \{0\} \times Q' \times\Omega_0(Q_L)
\end{aligned}
\right.
\end{equation*}
and then compute, using Lemma~\ref{l.spectralgap.UmuL}, 
\begin{align*}
\lefteqn{
\partial_t \left\langle \sum_{x\in (Q')^\circ} \tilde{V}^2(t,x,\cdot) \right\rangle_{\mu_{L,\xi}}
} \qquad & 
\\ &
=
2\left\langle \sum_{x\in (Q')^\circ} \tilde{V}(t,x,\cdot) \left( \mathcal{L}_{\mu_{L,\xi}} - \nabla^*\a\nabla \right)\tilde{V}(t,x,\cdot)  \right\rangle_{\mu_{L,\xi}}
\\ & 
\leq 
-2\lambda \left\langle \sum_{e \in \mathcal{E}(Q')} \left( \nabla \tilde{V}(t,e,\cdot) \right)^2   \right\rangle_{\mu_{L,\xi}}
-2 \left\langle \sum_{y\in Q } \sum_{x \in Q'} \left( \partial_y \tilde{V}(t,x,\cdot) \right)^2   \right\rangle_{\mu_{L,\xi}}
\\ & 
\leq 
-cL^{-2} \left\langle \sum_{x\in (Q')^\circ} \tilde{V}^2(t,x,\cdot) \right\rangle_{\mu_{L,\xi}}.
\end{align*}
This implies that 
\begin{align*}
\left\langle \sum_{x\in (Q')^\circ} \tilde{V}^2(t,x,\cdot) \right\rangle_{\mu_{L,\xi}}
&
\leq 
\left\| v \right\|_{L^2(Q',\mu_{L,\xi})}^2 
\exp\left( -c \frac{t}{L^2} \right) 
\\ &
\leq 
C \left( |q|^2+ \left\| \f \right\|_{L^2(Q',\mu_{L,\xi})}^2 \right) \exp\left( -c \frac{t}{L^2} \right).
\end{align*}
In particular, we have that $v(x,\phi) = \lim_{s\to \infty} V(s,x,\phi)$ in the sense of $L^2(Q',\mu_{L,\xi})$. The proof is complete.
\end{proof}

\section{Dynamical coupling and localization} \label{s.couple}

It was observed by Funaki and Spohn~\cite{FS} and later used by Miller~\cite{Mi} that the Langevin dynamics provide a convenient way to construct couplings between different gradient Gibbs measures, for example between $\mu_L$ and $\mu_{M}$ for different $M,L\in\N$, or between measures with slightly different choices of the potential~$\mathsf{V}$. In~\cite{FS} this coupling technique was used in order to prove the uniqueness of infinite-volume measures with a given slope. In~\cite{Mi} it was used to prove the CLT in finite volume, by comparing to the known CLT in infinite volume. 

\smallskip

The basic idea is that we can couple the measures by driving the dynamics in~\eqref{e.dynamics.phi.QL} with the same family~$\{  B_t(x) \}$ of Brownian motions and estimating the difference of the solutions of the system of SDEs with the aid of parabolic esimates (e.g., the De Giorgi-Nash $C^{0,\beta}$--type estimate for solutions of uniformly parabolic equations). In this section we will use this technique to obtain estimates not only on the difference of the~$\nabla\phi$ fields corresponding to different underlying Gibbs measures, but also on the closeness of the solutions of respective Helffer-Sj\"ostrand equations.

\subsection{Coupling the Dirichlet and periodic Gibbs measures}

In order to obtain useful estimates for our couplings, we must control the oscillation of the gradient fields samples by our Gibbs measures. The Brascamp-Lieb inequality provides estimate on the \emph{fluctuations} of our fields, but we need to also estimate the spatial oscillations of their expectations. In order to obtain such a bound for the field under the measure~$\mu_{L,\xi}$, we construct a coupling between~$\mu_{L,\xi}$ and~$\mu_{L,\xi,\per}$ (defined by \eqref{e.muper.def}) and apply Lemma \ref{l.oscillation.per.dyn}.

\smallskip

\begin{lemma}[{Dynamical coupling of~$\mu_{L,\xi}$ and $\mu_{L,\xi,\per}$}]
\label{l.coupling.Dir-Per}
Let $\mu_{L,\xi}$ and $\mu_{L,\xi,\per}$ be defined as above. There exists a random element $(\phi, {\phi}_\per)$ of $C(\R^+;\Omega_0(Q_L))\times C(\R^+;\Omega_\per(Q_L))$ with law $\Theta$ such that:
\begin{equation}
\label{e.lawyes.Dir}
\mbox{the law of~$\phi$ is $\mu_{L,\xi} \otimes \P'_{L,\xi,\phi} $,}
\end{equation}
\begin{equation}
\label{e.lawyes.Per}
\mbox{the law of~${\phi}_\per$ is $\mu_{L,\xi,\per} \otimes \P'_{L,\xi,\per,\phi}$,}
\end{equation}
and, for any $A<\infty$ and $\xi \in B_R$, there exists a constant $C(A,\data)<\infty$, such that for all $L>L_0(R, \data)$,
\begin{equation}
\label{e.couplingbound.DirPer}
\Theta \left(  
\left\{ 
\sup_{t\in[0,AL^2 \log L]}\sup_{x \in Q_{L}} \left| \phi_t(x) - {\phi}_{\per,t}(x) \right| 
> C s \log L 
\right\} 
\right)
\leq
C \exp\left( -s^2 \log L \right).
\end{equation}
\end{lemma}
\begin{proof}
As in Subsection~\ref{ss.generator}, let $\P'_{L,\xi,\phi_0}$ and ${\P}'_{L,\xi,\per,\tilde\phi_0}$ respectively denote the laws of the processes $\{ \phi_t\}$ and $\{ \phi_{\per,t} \}$ defined by
\begin{equation*}
\left\{
\begin{aligned}
& d\phi_t(x) 
= \sum_{y\sim x} \mathsf{V}'( -\xi\cdot (y-x)+ \phi_t(y)-\phi_t(x) ) \,dt + \sqrt{2} \,dB_t(x), && x\in Q_L^\circ, 
\\ & 
\phi_t(x) = 0, && x \in \partial Q_L,
\end{aligned} 
\right.  
\end{equation*}
and
\begin{equation*}
\left\{
\begin{aligned}
& d\phi_{\per,t}(x) 
= \sum_{y\sim x} \mathsf{V}'( -\xi\cdot (y-x)+ \phi_{\per,t}(y)-\phi_{\per,t}(x) ) \,dt 
\\ & \qquad\qquad\qquad
+ \sqrt{2} \,dB_t(x), \quad x\in [-L,L)^d\cap \Zd \setminus \{x_0\}, 
\\ &
\phi_{\per,t}(x_0) = 0,
\\  &
\phi_{\per,t}(x) 
= \phi_{\per,t}(y), \quad x \in \partial Q_L \cap (-L,L]^d,\,  x -y \in 2L\Zd
\end{aligned} 
\right.  
\end{equation*}
starting from $\phi_0\in \Omega_0(Q_L)$ and $\tilde{\phi}_0\in \Omega_\per(Q_{L})$, respectively. We may couple these measures by requiring that the family $\{  B_t(x) \,:\, x\in Q_{L} \}$ of Brownian motions driving the dynamics are the same. Note that the periodized dynamics have more Brownian motions than the Dirichlet dynamics, namely corresponding to the points~$[-L,L)^d\cap \Zd\setminus \left(  Q_L^\circ \cup\{ x_0 \}\right)$. 

\smallskip

We let $\P^*_{(\phi_0,\tilde\phi_0)}$ be the resulting coupled measure of the joint process $\{ (\phi_t,\phi_{t,\per}) \}$. We sample the initial data with $\mu_{L,\xi} \times \mu_{L,\xi,\per}$ by setting 
\begin{equation}
\Theta':= \left( \mu_{L,\xi} \times \mu_{L,\xi,\per} \right) \otimes \P^*_{(\phi,\tilde\phi)}. 
\end{equation}
In other words, $\Theta'$ is the law of the pair $(\phi_t,{\phi}_{\per,t})$ of trajectories obtained by first sampling $\phi_0$ and $\tilde{\phi}_0$ according to the measures~$\mu_{L,\xi}$ and~$ \mu_{L,\xi,\per}$, respectively, and then running the dynamics above. 

\smallskip

It is clear, by the invariance of the Gibbs measures with respect to the dynamics, that at any time $t$, the law of $\phi_t$ is $\mu_{L,\xi}$ and the law of $ \phi_{\per,t}$ is $\mu_{L,\xi,\per}$. We will eventually take the measure~$\Theta$ as in the statement of the proposition to be the law of~$\left(\phi_{t}(\cdot),\phi_{\per,t}(\cdot) \right)$. This ensures that~\eqref{e.lawyes.Dir} and~\eqref{e.lawyes.Per} are satisfied. It remains to check that the inequality~\eqref{e.couplingbound.DirPer} is satisfied. 

\smallskip

Consider the difference 
\begin{equation}
u(t,x) := \phi_t(x) - \phi_{\per,t}(x), \quad (t,x) \in (0,\infty) \times Q_L.
\end{equation}
Observe that $u$ satisfies the parabolic equation
\begin{equation}
\label{e.parab.DirPer}
\partial_tu + \nabla^* \cdot \hat{\a}\nabla u =  0
\quad 
\mbox{in} \ (0,\infty) \times Q_L^\circ,
\end{equation}
where $\hat{\a}$ is defined by
\begin{equation*}
\hat{\a}(t,e)
:=
\int_0^1 \mathsf{V}''\left(-\nabla \ell_\xi(e) + s \nabla \phi_t(e) + (1-s) \nabla {\phi}_{\per,t}(e) \right)\,ds.
\end{equation*}  
The maximum principle then implies that, for every $t\in (0,\infty)$, 
\begin{align}
\label{e.max}
\sup_{r\in[0,t]}\sup_{x\in Q_L} 
\left| u(r,x) \right|
& 
\leq 
\sup_{(r,x) \in (0,t] \times \partial Q_L} 
\left| \phi_{\per,r}(x) \right| 
+ 
C \exp \left( -\frac{t}{CL^2} \right) 
\sup_{x\in Q_L} \left| u(0,x) \right|.
\end{align}
Therefore, if we take~$t :=AL^2 \log L$ in~\eqref{e.max}, then we have for all large $L$
\begin{equation*}
\sup_{r\in[0,AL^2 \log L]} \, \sup_{x\in Q_L} 
\left| u(r,x) \right| \leq
2 \sup_{(r,x) \in (0,AL^2 \log L] \times \partial Q_L} 
\left| \phi_{\per,r}(x) \right| 
\end{equation*}
Applying Lemma~\ref{l.oscillation.per.dyn} with $T:= AL^2 \log L$ now yields~\eqref{e.couplingbound.DirPer}. 
\end{proof}

The previous lemma combined with the parabolic Nash estimate (see~\cite{SZ} for a discrete version which applies here), applied to the equation~\eqref{e.parab.DirPer}, yields that, for some exponent $\alpha(\data)>0$, 
\begin{multline}
\label{e.DirPer.grads}
\Theta \left(  
\left\{ 
\sup_{t\in[\frac12AL^2 \log L ,AL^2 \log L]}\sup_{x \in Q_{L/2}} \left| \nabla \phi_t(x) - \nabla {\phi}_{\per,t}(x) \right| 
> Cs L^{-\alpha} \log L 
\right\} 
\right)
\\
\leq
C \exp\left( -s^2 \log L \right).
\end{multline}
By~\eqref{e.perexpvanish}, this implies that, for some $\alpha(\data)>0$ and $C(\data)<\infty$,
\begin{equation}
\label{e.Dir.edges.0}
\sup_{e\in \mathcal{E}(Q_{L/2})} 
\left| \left\langle \nabla \phi(e) \right\rangle_{\mu_{L,\xi}} \right| 
\leq 
CL^{-\alpha}.
\end{equation}
Another easy consequence of Lemma~\ref{l.coupling.Dir-Per} is the following bound on the oscillation of field~$\phi_{L,\xi}$ sampled by (Dirichlet) finite-volume measure~$\mu_{L,\xi}$.

\begin{corollary}[{Oscillation estimate for~$\mu_{L,\xi}$}]
\label{c.oscillation}
Given $A<\infty, \xi \in B_R$, there exists $C=C(A, \data)<\infty$ such that for all $L>L_0(R,\data)$, 
\begin{equation}
\label{e.oscillation.muLxi}
\mu_{L,\xi} \otimes \P'_{L,\xi,\phi} \left( \left\{ 
\max_{t\in[0,AL^2 \log L]}\max_{x\in Q_L} 
| \phi(t, x)| 
> C s \log L
\right\} \right) 
\leq
\exp\left( -s^2 \log L \right).
\end{equation}
\end{corollary}
\begin{proof}
This is an immediate consequence of Lemmas~ \ref{l.oscillation.per.dyn} and~\ref{l.coupling.Dir-Per}. 
\end{proof}

\subsection{Coupling Dirichlet Gibbs measures with different~$\mathsf{V}$'s and~$L$'s}

In this subsection we give the main coupling result needed in the proof of Theorem~\ref{t.surfacetension}. 

\smallskip

Throughout this subsection, we consider two interaction potentials~$\mathsf{V}$ and~$\tilde{\mathsf{V}}$, each satisfying the assumptions stated in the introduction. We also fix $L,\tilde{L}\in\N$ with $L,\tilde{L}\geq 2$ and $\xi,\tilde{\xi} \in \Rd$ and let $\mu_{L,\xi}$ and $\tilde\mu_{L,\tilde{\xi}}$ denote, respectively, the measures defined in~\eqref{e.defmuL} with~$\mathsf{V}$ and $\tilde{\mathsf{V}}$.

\begin{proposition}[{Dynamic coupling of $\mu_{L,\xi}$ and $\tilde{\mu}_{\tilde{L}, \tilde{\xi}}$}]
\label{p.coupling}
Let $R\in (1,\infty)$ and fix a $\xi \in B_R$. Let $K,L,\tilde L \in\N$ with  $2\leq K \leq L\leq \tilde L$ and let the finite volume measures $\mu_{L,\xi}$ and $\tilde{\mu}_{\tilde{L},\tilde{\xi}}$ be defined as above.  There exists a random element $(\nabla \phi,\nabla \tilde{\phi})$ of $C(\R^+;\Omega_0(Q_L))\times C(\R^+;\Omega_0(Q_{\tilde{L}}))$ with law $\Theta$ such that:
\begin{equation}
\label{e.lawyes}
\mbox{the law of~$\nabla \phi$ is $\mu_{L,\xi} \otimes \P'_{L,\xi,\phi} $,}
\end{equation}
\begin{equation}
\label{e.lawyestilde}
\mbox{the law of~$\nabla \tilde{\phi}$ is $\tilde{\mu}_{\tilde{L},\tilde{\xi}}\otimes \P'_{\tilde{L},\xi,\tilde{\phi}}$,}
\end{equation}
and a constant ~$\beta(\data)\in \left(0,\tfrac12 \right]$, such that for all $A<\infty$, $K>K_0(R, \data)$ and any cube $x_0+Q_{2K} \subseteq Q_L$, there exists $C= C(A,\data)<\infty$ such that 
\begin{equation}
\label{e.couplingbound}
\E_\Theta\left[ \sup_{t\in [0,AK^2 \log K]}\sup_{e\in \mathcal{E}(x_0+Q_K)} \left| \nabla \phi(t,e) - \nabla \tilde{\phi}(t,e) \right| \right]
\leq
CK^{-\beta}+CK^{1-\beta} \left\| \mathsf{V}' - \tilde{\mathsf{V}}' \right\|_{L^\infty(\R)}.
\end{equation}
\end{proposition}
\begin{proof}
 As in Subsection~\ref{ss.generator}, let $\P'_{L,\phi_0}$ and $\tilde{\P}'_{\tilde{L},\tilde\phi_0}$ respectively denote the law of the processes defined in~\eqref{e.dynamics.phi.QL} with respect to~$(L,\mathsf{V}(\xi+\cdot))$ and~$(\tilde{L},\tilde{\mathsf{V}}(\tilde{\xi}+\cdot))$, starting from $\phi_0\in \Omega_0(Q_L)$ and $\tilde{\phi}_0\in \Omega_0(Q_{\tilde{L}})$, respectively. We may couple these measures by requiring that the family $\{  B_t(x) \,:\, x\in Q_{\tilde{L}} \}$ of Brownian motions driving the dynamics are the same. We let $\P^*_{(\phi_0,\tilde\phi_0)}$ be the resulting coupled measure of the joint process $(\phi,\tilde\phi)$. In other words, $\P^*_{(\phi_0,\tilde\phi_0)}$ is the law on the set of trajectories $(\phi_t,\tilde{\phi}_t)$ satisfying the coupled set of equations
\begin{equation}
\label{e.dynamics.coupling}
\left\{
\begin{aligned}
& d\phi_t(x) 
= \sum_{y\sim x} \mathsf{V}'( \phi_t(y)-\phi_t(x) ) \,dt  + \sqrt{2} \,dB_t(x), && x\in Q_L^\circ, 
\\  & 
d\tilde{\phi}_t(x) 
= \sum_{y\sim x} \tilde{\mathsf{V}}'( \tilde{\phi}_t(y)-\tilde{\phi}_t(x) ) \,dt + \sqrt{2} \,dB_t(x), && x\in Q_{\tilde{L}}^\circ,
\\ &
\phi_t(x) = 0, && x \in \partial Q_L,
\\ & 
\tilde{\phi}_t(x) = 0, && x \in \partial Q_{\tilde{L}},
\end{aligned} 
\right.  
\end{equation}
with initial data $(\phi_0,\tilde{\phi}_0)\in\Omega_0(Q_L)\times\Omega_0(Q_{\tilde{L}})$. Let us sample the initial data with $\mu_L \times \tilde{\mu}_{\tilde{L}}$ itself by setting 
\begin{equation}
\Theta':= \left( \mu_{L,\xi} \times \tilde{\mu}_{\tilde{L}, \tilde{\xi}} \right) \otimes \P^*_{(\phi,\tilde\phi)}. 
\end{equation}
In other words, $\Theta'$ is the law of the pair $(\phi_t,\tilde{\phi}_t)$ of trajectories obtained by first sampling $\phi_0$ and $\tilde{\phi}_0$ according to the measures~$\mu_{L,\xi}$ and~$\tilde{\mu}_{\tilde{L}, \tilde{\xi}}$, respectively, and then running the dynamics~\eqref{e.dynamics.coupling}. 

\smallskip

It is clear, by the invariance of the Gibbs measures with respect to the dynamics, that at any time $t$, the law of $\nabla \phi_t$ is $\mu_{L,\xi}$ and the law of $\nabla \tilde{\phi}_t$ is $\tilde{\mu}_{\tilde{L}, \tilde{\xi}}$. We will eventually take the measure~$\Theta$ as in the statement of the proposition to be the law of~$(\nabla \phi_t (\cdot),\nabla \tilde{\phi}_t(\cdot))$ in $t \in [t_*, t_* +AK^2 \log K]$, where a given time~$t_*$ will be selected below. This ensures that~\eqref{e.lawyes} and~\eqref{e.lawyestilde} are satisfied. It remains therefore to show that we can select $t$ in such a way that the bound~\eqref{e.couplingbound} is satisfied. 

\smallskip

Consider the difference 
\begin{equation}
u(t,x) := \phi_t(x) - \tilde{\phi}_t(x), \quad (t,x) \in (0,\infty) \times Q_L.
\end{equation}
Observe that $u$ satisfies the parabolic equation
\begin{equation}
\partial_tu + \nabla^* \cdot \hat{\a}\nabla u =  \nabla^* \mathbf{f} 
\quad 
\mbox{in} \ (0,\infty) \times Q_L,
\end{equation}
where $\hat{\a}$ and $\mathbf{f}$ are defined by
\begin{equation}
\left\{
\begin{aligned}
& 
\hat{\a}(t,e)
:=
\int_0^1 \mathsf{V}''\left(s \nabla \phi_t(e) + (1-s) \nabla \tilde{\phi}_t(e) \right)\,ds,
\\ & 
\mathbf{f} (t,e):= \mathsf{V}'(\nabla \tilde{\phi}_t(e))  - \tilde{\mathsf{V}}'(\nabla \tilde{\phi}_t(e)) .
\end{aligned}
\right. 
\end{equation} 
Denote, for $r>0$ and $(t,x)\in (0,\infty) \times \Zd$,  the parabolic cylinder 
\begin{equation*}
W_r (t,x) 
:= 
(t,x) + (-r^2,0] \times Q_r.
\end{equation*}
By the parabolic Nash estimate (see~\cite{SZ}), there exist~$\beta(\data)\in \left(0,\tfrac12 \right]$ and~$C(\data)>0$ such that, for every $t\in [K^2,\infty)$,
\begin{align*}
&
\left\| u \right\|_{L^\infty(W_{\lceil K/2\rceil} (t,0))}
+
K^{\beta} \left[ u \right]_{C^{0,\beta}\left(W_{\lceil K/2\rceil} (t,0) \right)} 
\\ & \qquad 
\leq 
C \left(  
\left\| u - \left( u \right)_{ W_{K} (t,0)} \right\|_{\underline{L}^2\left( W_{K} (t,0)  \right)}
+ K \left\| \mathbf{f} \right\|_{L^\infty(W_{K} (t,0))} \right).
\end{align*}
In particular, 
\begin{align*}
\lefteqn{
\sup_{s\in[t-K^2/4,t]}\sup_{e\in \mathcal{E}(Q_{\lceil K/2\rceil})} 
\left| \nabla \phi(s,e) - \nabla \tilde{\phi}(s,e) \right| 
} \qquad & 
\\ &
\leq
\left[ u \right]_{C^{0,\beta}\left(W_{\lceil K/2\rceil} (t,0) \right)} 
\\ & 
\leq 
CK^{-\beta} \left( \left\| \phi- ({\phi})_{W_K(t,0)} \right\|_{\underline{L}^2( W_K(t,0) )} + \left\| \tilde{\phi} -(\tilde{\phi})_{W_K(t,0)} \right\|_{\underline{L}^2( W_K(t,0) )} \right) 
\\ & \qquad 
+ K^{1-\beta} \left\| \mathsf{V}' - \tilde{\mathsf{V}}' \right\|_{L^\infty(\R)} .
\end{align*}
Applying Corollary~\ref{c.oscillation}, we obtain for some $C(A, \data)<\infty$ and all $K>K_0(R,\data)$
\begin{multline*}
\P_\Theta \left[ \left\| \phi- ({\phi})_{W_K(t,0)} \right\|_{\underline{L}^2( W_K(t,0) )} + \left\| \tilde{\phi} -(\tilde{\phi})_{W_K(t,0)} \right\|_{\underline{L}^2( W_K(t,0) )} > Cs\log K  \right]
\\ 
\leq \exp(-cs^2 \log K).
\end{multline*}
Thus for all $K>K_0(R,\data)$
\begin{align*}
\P_\Theta \left[  \sup_{s\in[t-K^2/4,t]}\sup_{e\in \mathcal{E}(Q_{\lceil K/2\rceil})} 
\left| \nabla \phi(s,e) - \nabla \tilde{\phi}(s,e) \right| -  K^{1-\beta} \left\| \mathsf{V}' - \tilde{\mathsf{V}}' \right\|_{L^\infty(\R)} > cs \log(Kt) \right] \\
\leq
\exp(-cs^2 \log (Kt))
\end{align*}
Taking $t_*= K^2$, repeatedly apply the last inequality for $t = t_* + mK^2/4, m=0,1,\cdots, 4A\log K$, and take a union bound over $ t\in [K^2, K^2 (1+A\log K)]$, we conclude the proposition. 
\end{proof}

\subsection{Infinite-volume Gibbs measures}
\label{ss.infinite}
The coupling results from the previous subsections allow us to construct an infinite-volume $\nabla\phi$-Gibbs state. We denote by $\Omega_\infty$ the set of gradient fields on $\mathcal{E}(\Zd)$. That is, $p\in \Omega_\infty$ if there exists $\phi:\Zd \to \R$ such that $p(e) = \nabla \phi(e)$ for every $e\in \mathcal{E}(\Zd)$. We may identify $\phi\in \Omega_0(Q_L)$ with an element of $\Omega_\infty$ by extending $\phi$ to be zero outside of $Q_L^\circ$ and then taking its gradient. We may also identify $\phi\in\Omega_\per(Q_L)$ with~$\nabla \phi\in \Omega_\infty$. 

\smallskip

For each~$\xi \in \R^d$, we now construct an infinite volume, translation invariant and ergodic $\nabla \phi$ Gibbs measure $\mu_{\infty,\xi}$, such that $\left\langle \nabla \phi (e)\right\rangle_{\mu_{\infty,\xi}} = 0$ for all $e\in\mathcal{E}(\Zd)$.  In view of Proposition~\ref{p.coupling}, for any $F \in C_c^\infty(\Omega_\infty)$, the sequence of random variable $\left\langle F(\nabla \phi)\right\rangle_{\mu_{L,\xi}}$ converges as $L\to \infty$. 
It follows that the law of $\nabla \phi$ under $\mu_{L,\xi}$, viewed as an element of $\Omega_\infty$, converges weakly as $L\to \infty$. 
By Lemma~\ref{l.coupling.Dir-Per}, in particular~\eqref{e.DirPer.grads}, the law of $\nabla \phi$ under $\mu_{L,\xi,\per}$ also converges weakly to the same limiting measure.  
We denote the (unique) limiting law by $\mu_{\infty,\xi}$. In view of~\eqref{e.perexpvanish}, we have that, for every $e \in \mathcal{E}(\Zd)$, 
\begin{equation}
\left\langle \nabla \phi(e) \right\rangle_{\mu_{\infty,\xi}} = 0. 
\end{equation}
Moreover, since the law of $\nabla \phi$ under $\mu_{L,\xi,\per}$ is invariant with respect to $\Zd$--translations, the same is true of $\mu_{\infty,\xi}$. 
That is,
\begin{equation}
\label{e.muinfty.stationary}
\mbox{the law of $\nabla \phi(z+\cdot)$ with respect to $\mu_{\infty,\xi}$ does not depend on $z\in\Zd$.}
\end{equation}
The stationarity property~\eqref{e.muinfty.stationary} is very convenient to work with, and for this reason we often work with the measure~$\mu_{\infty,\xi}$.

\subsection{Comparing Helffer-Sj\"ostrand solutions}

In this section, we use the couplings we have constructed in the previous subsections together with the representation formulas of Subsection~\ref{ss.generator} to obtain estimates on the continuous dependence, with respect to various parameters, of the solutions of the Helffer-Sj\"ostrand equation.

\begin{lemma}
\label{l.coupling.D}
Let $K,L,M\in\N$ with  $2\leq K \leq L\leq M$ and let the finite volume measures $\mu_{L,\xi}$ and $\tilde{\mu}_{M,\tilde{\xi}}$ be defined as above. Also let $R\in (1,\infty)$ and fix a $|\xi| <R$. There exist
$\beta(\data)>0$, and a random element $(\nabla \phi,\nabla \tilde{\phi})$ of $\Omega_0(Q_L)\times \Omega_0(Q_M)$ with law $\Theta$ such that:
\begin{equation}
\mbox{the law of~$\nabla \phi$ is $\mu_{L,\xi}$,}
\end{equation}
\begin{equation}
\mbox{the law of~$\nabla \tilde{\phi}$ is $\tilde{\mu}_{M,\tilde{\xi}}$, }
\end{equation}
and, for all $K>K_0(R,\data) $ and any cube $x_0+Q_{2K} \subseteq Q_L$, the solutions to the Dirichlet problems \begin{equation}
\label{e.Dircoupl}
\left\{
\begin{aligned}
& -\mathcal{L}_{\mu_L} v + \nabla^*\a\nabla v =  \nabla^*\f & \mbox{in} & \ \left( x_0+Q_K \right)^\circ \times\Omega_0(Q_L), \\
& v = 0  & \mbox{on} & \  \partial \left( x_0+Q_K \right) \times\Omega_0(Q_L),
\end{aligned}
\right.
\end{equation}
and
\begin{equation}
\label{e.Dircoupltilde}
\left\{
\begin{aligned}
& -\mathcal{L}_{\tilde\mu_M} \tilde{v} + \nabla^*\tilde\a\nabla v = \nabla^*\tilde{\f} & \mbox{in} & \ \left( x_0+Q_K \right)^\circ \times\Omega_0(Q_M), \\
& \tilde{v} = 0  & \mbox{on} & \  \partial \left( x_0+Q_K \right) \times\Omega_0(Q_M),
\end{aligned}
\right.
\end{equation}
satisfy the following estimate with $C=C(\mathsf M, \data)$
\begin{align*}
\lefteqn{
\E_\Theta \left[ \left\|\nabla  v-\nabla \tilde{v} \right\|_{\underline{L}^2(x_0+Q_K)}^2 \right] 
} \quad & 
\\ & 
\leq 
C \log K \bigg(
\E_\Theta \left[ \left\| \f(\cdot,\phi) - \tilde{\f}(\cdot,\tilde\phi) \right\|_{\underline{L}^2(x_0+Q_K)}^2 \right] 
\\ & \quad 
+
C\left( K^{-\beta} + K^{1-\beta} \left\| \mathsf{V}' - \tilde{\mathsf{V}}' \right\|_{L^\infty(\R)} 
+\left\| \mathsf{V}'' - \tilde{\mathsf{V}}'' \right\|_{L^\infty(\R)} \right) 
\E_\Theta \left[  \left\| \tilde{\f}(\cdot,\tilde{\phi}) \right\|_{\underline{L}^2(x_0+Q_K)}^4 \right]^{\frac12} \bigg).
\end{align*}
\end{lemma}
\begin{proof}
As in the proof of Proposition~\ref{p.coupling}, for each $(\phi,\tilde\phi)\in \Omega_0(Q_L)\times\Omega_0(Q_M)$, we let~$\P^*_{(\phi,\tilde\phi)}$ denote the law of the trajectories of the system~\eqref{e.dynamics.coupling} of SDEs starting from the initial data $(\phi,\tilde\phi)$. We let $\E^*_{(\phi,\tilde\phi)}$ denote the corresponding expectation. Throughout the argument, we set $U:= x_0+Q_K$. 

\smallskip

For each pair $(\phi_t,\tilde{\phi}_t)$ of trajectories, we let the function $w(\cdot;\{ \phi_\cdot\})$ be the solution of the parabolic initial-value problem~\eqref{e.parabolic.freezephi} with $F= \nabla^*\f$ and $\tilde{w} (\cdot;\{ \tilde{\phi}_\cdot\})$ be the solution of
\begin{equation}
\label{e.parabolic.freezephitilde}
\left\{
\begin{aligned}
& \partial_t \tilde{w} + \nabla^*\a_{\{\tilde{\phi}_\cdot\}}\nabla \tilde{w} = 0 & \mbox{in} & \ (0,\infty) \times U^\circ,\\
& \tilde{w} = 0 & \mbox{on} & \ (0,\infty) \times \partial U,\\
& \tilde{w} =  \nabla^*\f(\cdot,\tilde\phi_0) & \mbox{on} & \ \{ 0\} \times U^\circ.
\end{aligned}
\right.
\end{equation}
Applying Lemma~\ref{l.represent}, we find that, for every $(\phi,\tilde\phi)\in \Omega_0(Q_L)\times\Omega_0(Q_M)$,
\begin{equation}
\label{e.vdiff}
v(x,\phi) - \tilde v(x,\tilde\phi)
= 
\E^*_{(\phi,\tilde\phi)}\left[
 \int_0^\infty W(t,x;\{\phi_\cdot, \tilde{\phi}_\cdot\}) \,dt
\right],
\end{equation}
where we define 
\begin{equation*}
W(t,x;\{\phi_\cdot, \tilde{\phi}_\cdot\}):=
w\left(t,x;\{\phi_\cdot\} \right)
-
\tilde{w}\left(t,x;\{\tilde\phi_\cdot\} \right).
\end{equation*}
Notice that $W(\cdot ;\{\phi_\cdot, \tilde{\phi}_\cdot\})$ satisfies the initial-value problem 
\begin{equation}
\label{e.parabolic.freezediff}
\left\{
\begin{aligned}
& \partial_t W + \nabla^*\a_{\{\phi_\cdot\}}\nabla W = \nabla^*\h & \mbox{in} & \ (0,\infty) \times U^\circ,\\
& W = 0 & \mbox{on} & \ (0,\infty) \times \partial U,\\
& W = 
\nabla^*\!\left( \f(\cdot,\phi) - \tilde{\f}(\cdot,\tilde\phi)\right)  & \mbox{on} & \ \{ 0\} \times U^\circ,
\end{aligned}
\right.
\end{equation}
where for convenience we have defined
\begin{equation*}
\h(t,e):= 
 \left( \tilde{\a}_{\{\tilde\phi_\cdot\}}(t,e) - \a_{\{\phi_\cdot\}}(t,e) \right) \nabla\tilde{w}\left(t,e;\{\tilde\phi_\cdot\} \right).
\end{equation*}
By~\eqref{e.vdiff}, we have 
\begin{equation*}
\left\| \nabla v - \nabla\tilde{v} \right\|_{L^2(U)}^2
= \sum_{e\in \mathcal{E}(x_0+Q_K)} \E^*_{(\phi,\tilde\phi)}\left[
 \int_0^\infty \nabla W(t,e;\{\phi_\cdot, \tilde{\phi}_\cdot\}) \,dt
\right]^2.
\end{equation*}
For every $t>0$ and $e\in \mathcal{E}(x_0+Q_K)$, 
\begin{align*}
&
\E^*_{(\phi,\tilde\phi)}\left[
 \int_t^{2t} \nabla W(t,e;\{\phi_\cdot, \tilde{\phi}_\cdot\}) \,dt \right] 
\leq 
t^{\frac12} 
\left( \E^*_{(\phi,\tilde\phi)}\left[
 \int_t^{2t} \left| \nabla W(t,e;\{\phi_\cdot, \tilde{\phi}_\cdot\}) \right|^2 \,dt \right] \right)^{\frac12}.
%
\end{align*}
Splitting the time into dyadic time intervals~$I_k:=\left[2^{k},2^{k+1}\right)$ and using the Cauchy-Schwarz inequality, we deduce that
\begin{align*}
\left\| \nabla v - \nabla\tilde{v} \right\|_{L^2(U)}^2
& 
\leq
\sum_{k,l\in \Z}
\sum_{e\in \mathcal{E}(x_0+Q_K)}
2^{\frac k2+\frac l2} 
\left( \E^*_{(\phi,\tilde\phi)}\left[
\int_{I_k} \left| \nabla W(t,e;\{\phi_\cdot, \tilde{\phi}_\cdot\}) \right|^2 \,dt \right] \right)^{\frac12}
\\ & \qquad\qquad\qquad\times \notag
\left( \E^*_{(\phi,\tilde\phi)}\left[
\int_{I_l} \left| \nabla W(t,e;\{\phi_\cdot, \tilde{\phi}_\cdot\}) \right|^2 \,dt \right] \right)^{\frac12}
\\ & \notag
\leq 
\sum_{k,l\in \Z} 2^{\frac k2+\frac l2}
\left( \E^*_{(\phi,\tilde\phi)}\left[ \sum_{e\in \mathcal{E}(x_0+Q_K)}
\int_{I_k} \left| \nabla W(t,e;\{\phi_\cdot, \tilde{\phi}_\cdot\}) \right|^2 \,dt\right] 
\right)^{\frac12} 
\\ & \qquad\qquad\qquad\times \notag
\left(  \E^*_{(\phi,\tilde\phi)}\left[ \sum_{e\in \mathcal{E}(x_0+Q_K)}
\int_{I_l} \left| \nabla W(t,e;\{\phi_\cdot, \tilde{\phi}_\cdot\}) \right|^2 \,dt 
\right]\right)^{\frac12}
\\ & \notag
= 
\left( \sum_{k\in\Z} 2^{\frac k2} 
\left( \E^*_{(\phi,\tilde\phi)}\left[ \sum_{e\in \mathcal{E}(x_0+Q_K)}
\int_{I_k} \left| \nabla W(t,e;\{\phi_\cdot, \tilde{\phi}_\cdot\}) \right|^2 \,dt\right] 
\right)^{\frac12} 
\right)^2.
\end{align*}
We next estimate, for each $k\in\Z$, the term 
\begin{equation*}
2^k\E^*_{(\phi,\tilde\phi)} \!\left[ \sum_{e\in \mathcal{E}(x_0+Q_K)}
\int_{I_k} \left| \nabla W(t,e;\{\phi_\cdot, \tilde{\phi}_\cdot\}) \right|^2 dt 
\right].
\end{equation*}
We break into two cases, depending the size of~$2^k$ relative to $K$.

\smallskip

For small times (i.e., $k\leq 0$), we apply Lemmas~\ref{l.CD.decay} and discreteness (which implies trivially that $\left\| \nabla U \right\|_{L^2(U)} \leq C \left\| U \right\|_{L^2(U)}$ for any function $U\in L^2(U)$) to immediately obtain
\begin{align*}
\left\| \nabla W \right\|_{L^2((0,1)\times U)}^2
&
\leq 
\left\| W \right\|_{L^2((0,1)\times U)}^2
\\ &
\leq 
C\left\| \f(\cdot,\phi) - \tilde{\f}(\cdot,\tilde\phi) \right\|_{L^2(U)}^2
+
\int_0^1 \left\| \h(t,\cdot) \right\|_{L^2(U)}^2\,dt.
\end{align*}
For each~$T\in [1,\infty)$ we test the equation~\eqref{e.parabolic.freezediff} with $\eta W$, where $\eta(t)$ is a cutoff in time satisfying $\indc_{(T,2T)}\leq \eta \leq \indc_{\left(\frac12T,2T\right)}$ and $\left| \eta' \right| \leq 4T^{-1}$, and then apply  Lemma~\ref{l.CD.decay}, to obtain
\begin{align}
\label{e.readyforW}
\lefteqn{
\left\| \nabla W \right\|_{L^2((T,2T)\times U)}^2
} \quad & 
\\ &\notag
\leq 
\frac{C}{T} \int_{\frac12T}^{2T}  \left\| W(s, \cdot) \right\|_{L^2(U)}^2 \,ds
+
C\int_{\frac 12T}^{2T} \left\| \h(s,\cdot) \right\|_{L^2(U)}^2\,dt
\\ & \notag
\leq 
\frac{C}T
\left\| \f(\cdot,\phi) - \tilde{\f}(\cdot,\tilde\phi) \right\|_{L^2(U)}^2
\exp\left( -\frac{T}{CK^2} \right)
\\ & \quad \notag
+\frac{C}{T} \int_{\frac12T}^{2T}  
\int_0^{s} 
\left\| \h(t'-s,\cdot) \right\|_{L^2(U)}^2
\exp\left( -\frac{s}{CK^2} \right)\,ds \,dt'
+
C\int_{\frac 12T}^{2T} \left\| \h(s,\cdot) \right\|_{L^2(U)}^2\,dt.
\end{align}
To estimate the terms involving~$\h$, we first apply the regularity assumption on~$\mathsf{V}$  (see~\eqref{e.V.C2gamma}) to obtain
\begin{align*}
\left\|  \tilde{\a}_{\{\tilde\phi_s\}} - \a_{\{\phi_s\}} \right\|_{L^\infty(U)} 
&
\leq 
\left\|  \tilde{\a}_{\{\tilde\phi_s\}} - \a_{\{\tilde\phi_s\}} \right\|_{L^\infty(U)} 
+
\left\|  \a_{\{\tilde\phi_s\}} - \a_{\{\phi_s\}}\right\|_{L^\infty(U)} 
\\ 
&\leq \left\|  \tilde{\mathsf{V}}'' - \mathsf{V}'' \right\|_{L^\infty(\R)}   + \mathsf M \| \nabla \tilde\phi_s - \nabla \phi_s \|^\gamma_{L^\infty(U)}.
\end{align*}
We next apply Lemma~\ref{l.CD.decay} and the parabolic Caccioppolli inequality  to~\eqref{e.parabolic.freezephitilde} to obtain, for every $T\in[1,\infty)$, 
\begin{align}
\label{e.gradwt}
\left\| \nabla\tilde{w}(\cdot,\tilde{\phi}) \right\|_{L^2(\left(\frac12T,2T\right)\times U)}^2
\leq 
C\left(1+ T\right)^{-2} \left\| \tilde{\f}(\cdot,\tilde{\phi}) \right\|_{L^2(U)}^2 \exp\left( -\frac{T}{CK^2} \right).
\end{align}
Fix a large constant $A<\infty$. For every $T\in\left[1, AK^2 \log K\right]$, we can apply Proposition~\ref{p.coupling} to find that 
\begin{align*}
\lefteqn{
\E_\Theta \left[ \left\| \h(s,\cdot) \right\|_{L^2(\left(\frac12T,2T\right)\times U)}^2 \right]
} \qquad &
\\ & 
\leq 
(1+T)^{-2}  \left\|  \tilde{\mathsf{V}}'' - \mathsf{V}'' \right\|_{L^\infty(\R)}^2 
\E_\Theta \left[ \left\| \tilde{\f}(\cdot,\tilde{\phi}) \right\|^2_{L^2(\left(\frac12T,2T\right)\times U)} \right]
\\ & \quad
+ (1+T)^{-2} \mathsf M \E_\Theta \left[ \| \nabla \tilde\phi_s - \nabla \phi_s \|_{L^\infty(U)} \right]^{\frac\gamma2}
\E_\Theta \left[  \left\| \tilde{\f}(\cdot,\tilde{\phi}) \right\|_{L^2(U)}^4 \right]^{\frac12}
\\ & 
\leq 
\left(1+T\right)^{-2} \mathsf{D},
\end{align*}
where we have defined for some $C=C(\mathsf M, \data)$
\begin{equation}
\label{e.D}
\mathsf{D}
:=
C\left( K^{-\beta} + K^{1-\beta} \left\| \mathsf{V}' - \tilde{\mathsf{V}}' \right\|_{L^\infty(\R)} 
+\left\| \mathsf{V}'' - \tilde{\mathsf{V}}'' \right\|_{L^\infty(\R)} \right) 
\E_\Theta \left[  \left\| \tilde{\f}(\cdot,\tilde{\phi}) \right\|_{L^2(U)}^4 \right]^{\frac12}.
\end{equation}
Inserting this into~\eqref{e.readyforW} and integrating, we obtain, for every $T\in\left[1, AK^2 \log K\right]$,
\begin{align*}
T\E_\Theta \left[ \left\| \nabla W \right\|_{L^2((T,2T)\times U)}^2 \right]
\leq 
C\left( \E_\Theta \left[ \left\| \f(\cdot,\phi) - \tilde{\f}(\cdot,\tilde\phi) \right\|_{L^2(U)}^2 \right]
+ \mathsf{D} \right). 
\end{align*}
Likewise, for small times, we get
\begin{equation}
\E_\Theta \left[ \left\| \nabla W \right\|_{L^2((0,1)\times U)}^2 \right]
\leq 
C\left( \E_\Theta \left[ \left\| \f(\cdot,\phi) - \tilde{\f}(\cdot,\tilde\phi) \right\|_{L^2(U)}^2 \right]
+ \mathsf{D} \right). 
\end{equation}
For large times $T\in\left[A K^2\log K,\infty\right)$, we  instead use the boundedness of~$\a$ and~$\tilde{\a}$ together with~\eqref{e.gradwt} to obtain
\begin{align*}
\E_\Theta \left[ \left\| \h \right\|_{L^2((t,2t)\times U)}^2 \right]
&
\leq 
CT^{-2}  \E_\Theta\left[ \left\| \tilde{\f}(\cdot,\tilde{\phi}) \right\|_{L^2(U)}^2 \right]
\exp\left(-\frac{T}{CK^2}\right)
\\ & 
\leq 
CT^{-100} \E_\Theta\left[ \left\| \tilde{\f}(\cdot,\tilde{\phi}) \right\|_{L^2(U)}^2 \right].
\end{align*}
We deduce that, for every~$T\in\left[A K^2\log K,\infty\right)$,
\begin{align*}
T\E_\Theta \left[ \left\| \nabla W \right\|_{L^2((T,2T)\times U)}^2 \right]
&
\leq 
CT^{-99} \left( \E_\Theta \left[\left\| \f(\cdot,\phi) - \tilde{\f}(\cdot,\tilde\phi) \right\|_{L^2(U)}^2 \right] + 
 \E_\Theta\left[ \left\| \tilde{\f}(\cdot,\tilde{\phi}) \right\|_{L^2(U)}^2 \right]
\right)
\\ & 
\leq 
CT^{-99} \left(  \E_\Theta \left[\left\| \f(\cdot,\phi) - \tilde{\f}(\cdot,\tilde\phi) \right\|_{L^2(U)}^2 \right] + \mathsf{D} \right).
\end{align*}
Summing these bounds for $T=2^k$ over $k\in\N$, we obtain
\begin{equation*}
\left\| \nabla v - \nabla\tilde{v} \right\|_{L^2(U)}^2
\leq C\log K\left( \E_\Theta \left[ \left\| \f(\cdot,\phi) - \tilde{\f}(\cdot,\tilde\phi) \right\|_{L^2(U)}^2 \right]
+ \mathsf{D} \right)^{\frac12}. \qedhere
\end{equation*}
\end{proof}

We next give an analogue of the previous lemma for the Neumann problem. 

\begin{lemma}
\label{l.coupling.N}
Let $K, L,M\in\N$ with  $2\leq K \leq L\leq M$ and let the finite volume measures $\mu_{L,\xi}$ and $\tilde{\mu}_{M,\tilde{\xi}}$ be defined as above. Also let $R\in (1,\infty)$ and fix a $|\xi |<R$. There exist
$\beta(\data)>0$ and a random element $(\nabla \phi,\nabla \tilde{\phi})$ of $\Omega_0(Q_L)\times \Omega_0(Q_M)$ with law $\Theta$ such that:
\begin{equation}
\mbox{the law of~$\nabla \phi$ is $\mu_{L,\xi}$,}
\end{equation}
\begin{equation}
\mbox{the law of~$\nabla \tilde{\phi}$ is $\tilde{\mu}_{M,\tilde{\xi}}$, }
\end{equation}
and, for all $K>K_0(R,\data) $ and any cube $x_0+Q_{2K} \subseteq Q_L$, the solutions to the Neumann problems 

\begin{equation}
\label{e.Neucoupl}
\left\{
\begin{aligned}
& -\mathcal{L}_{\mu_L} v + \nabla^*\a\nabla v = \nabla^*\f & \mbox{in} & \ (x_0+Q_K)^\circ \times\Omega_0(Q_L), \\
& \a \nabla v - \f = \nabla \ell_q  & \mbox{on} & \  \partial (x_0+Q_K) \times\Omega_0(Q_L),
\end{aligned}
\right.
\end{equation}
and
\begin{equation}
\label{e.Neucoupltilde}
\left\{
\begin{aligned}
& -\mathcal{L}_{\tilde\mu_M} \tilde{v} + \nabla^*\tilde\a\nabla v =\nabla^*\tilde{\f} & \mbox{in} & \ (x_0+Q_K)^\circ \times\Omega_0(Q_M), \\
& \ \tilde\a \nabla \tilde{v} - \tilde{\f} = \nabla \ell_q  & \mbox{on} & \  \partial (x_0+Q_K) \times\Omega_0(Q_M),
\end{aligned}
\right.
\end{equation}
we have the following estimate with $C=C(\mathsf M, \data)$
\begin{align*}
\lefteqn{
\E_\Theta \left[ \left\|\nabla  v-\nabla \tilde{v} \right\|_{\underline{L}^2(x_0+Q_K)}^2 \right] 
} \ \  & 
\\ & 
\leq 
C \log K \bigg(
\E_\Theta \left[ \left\| \f(\cdot,\phi) - \tilde{\f}(\cdot,\tilde\phi) \right\|_{\underline{L}^2(x_0+Q_K)}^2 \right] 
\\ & \quad 
+
C\left( K^{-\beta} + K^{1-\beta} \left\| \mathsf{V}' - \tilde{\mathsf{V}}' \right\|_{L^\infty(\R)} 
+\left\| \mathsf{V}'' - \tilde{\mathsf{V}}'' \right\|_{L^\infty(\R)} \right) 
\E_\Theta \left[  \left\| \tilde{\f}(\cdot,\tilde{\phi}) \right\|_{\underline{L}^2(x_0+Q_K)}^4 \right]^{\frac12} \bigg).
\end{align*}
\end{lemma}
\begin{proof}
The proof of this lemma is completely analogous to Lemma~\ref{l.coupling.D} in the Dirichlet case. For each $(\phi,\tilde\phi)\in \Omega_0(Q_L)\times\Omega_0(Q_M)$, we let~$\P^*_{(\phi,\tilde\phi)}$ denote the law of the trajectories of the system~\eqref{e.dynamics.coupling} of SDEs starting from the initial data $(\phi,\tilde\phi)$. We let $\E^*_{(\phi,\tilde\phi)}$ denote the corresponding expectation. As in the previous argument, we set $U:=x_0+Q_K$. 

\smallskip

For each pair $(\phi_t,\tilde{\phi}_t)$ of trajectories, we let the function $w(\cdot;\{ \phi_\cdot\})$ be the solution of the parabolic initial-value problem~\eqref{e.parabolic.freezephi} with $F= \nabla^*\f$ and $\tilde{w} (\cdot;\{ \tilde{\phi}_\cdot\})$ be the solution of
\begin{equation}
\label{e.parabolic.freezephitildeN}
\left\{
\begin{aligned}
& \partial_t \tilde{w} + \nabla^*\a_{\{\tilde{\phi}_\cdot\}}\nabla \tilde{w} = 0 & \mbox{in} & \ (0,\infty) \times U^\circ,\\
& \a_{\{\tilde{\phi}_\cdot\}}\nabla \tilde{w} = 0 & \mbox{on} & \ (0,\infty) \times \partial \mathcal{E}\left(U\right),\\
& \tilde{w} =  \nabla^*\f(\cdot,\tilde\phi_0) & \mbox{on} & \ \{ 0\} \times U^\circ, \\
& 
\a_{\{\tilde{\phi}_\cdot\}} \nabla \tilde{w} - \f(\cdot,\tilde\phi_0) = \nabla \ell_q & 
\mbox{on} & \ \{ 0 \} \times \partial \mathcal{E} \left(  U \right).
\end{aligned}
\right.
\end{equation}
Applying Lemma~\ref{l.represent}, we find that, for every $(\phi,\tilde\phi)\in \Omega_0(Q_L)\times\Omega_0(Q_M)$,
\begin{equation*}
v(x,\phi) - \tilde v(x,\tilde\phi)
= 
\E^*_{(\phi,\tilde\phi)}\left[
 \int_0^\infty W(t,x;\{\phi_\cdot, \tilde{\phi}_\cdot\}) \,dt
\right],
\end{equation*}
where we define 
\begin{equation*}
W(t,x;\{\phi_\cdot, \tilde{\phi}_\cdot\}):=
w\left(t,x;\{\phi_\cdot\} \right)
-
\tilde{w}\left(t,x;\{\tilde\phi_\cdot\} \right).
\end{equation*}
Notice that $W(\cdot ;\{\phi_\cdot, \tilde{\phi}_\cdot\})$ satisfies the Cauchy-Neumann problem 
\begin{equation}
\left\{
\begin{aligned}
& \partial_t W + \nabla^*\a_{\{\tilde{\phi}_\cdot\}} \nabla W =\nabla^* \h & \mbox{in} & \ (0,\infty) \times U^\circ,\\
& \a_{\{\tilde{\phi}_\cdot\}} \nabla W  = \h & \mbox{on} & \ (0,\infty) \times \partial U,\\
& W = \nabla^*(\f - \tilde \f) & \mbox{on} & \ \{ 0\} \times U^\circ, \\
& \a_{\{\tilde{\phi}_\cdot\}} \nabla W -(\f- \tilde \f) = \h  & \mbox{on} & \ \{ 0\} \times\partial \mathcal{E}(U ).
\end{aligned}
\right.
\end{equation}
Following exactly the same proof as Lemma \ref{l.coupling.D}, and applying the Cauchy-Neumann estimate Lemma \ref{l.CN.decay} instead of Lemma \ref{l.CD.decay}, we obtain

\begin{align*}
T\E_\Theta \left[ \left\| \nabla W \right\|_{L^2((T,2T)\times U)}^2 \right]
&
\leq 
C\left( \E_\Theta \left[ \left\| \f(\cdot,\phi) - \tilde{\f}(\cdot,\tilde\phi) \right\|_{L^2(U)}^2 + \left\| \h \right\|_{L^2(\partial U)}^2 \right]
+ \mathsf{D} \right) 
\\ &
\leq 
C\left( \E_\Theta \left[ \left\| \f(\cdot,\phi) - \tilde{\f}(\cdot,\tilde\phi) \right\|_{L^2(U)}^2  \right]
+ \mathsf{D} \right),
\end{align*}
where $\mathsf{D}$ is defined in \eqref{e.D}. The second inequality of the previous display follows from  Proposition~\ref{p.coupling}, which gives that
\begin{align*}
\lefteqn{
\E_\Theta \left[ \left\| \h(s,\cdot) \right\|_{L^2(U)}^2 \right]
}  &
\\ & 
\leq 
  \left\|  \tilde{\mathsf{V}}'' - \mathsf{V}'' \right\|_{L^\infty(\R)}^2 
\E_\Theta \left[ \left\| \tilde{\f}(\cdot,\tilde{\phi}) \right\|^2_{L^2( U)} \right]
+ \E_\Theta \left[ \| \nabla \tilde\phi_s - \nabla \phi_s \|_{L^\infty(U)} \right]^{\frac\gamma2}
\E_\Theta \left[  \left\| \tilde{\f}(\cdot,\tilde{\phi}) \right\|_{L^2(U)}^4 \right]^{\frac12}
\\ & 
\leq 
\mathsf{D}.
\end{align*}
We then conclude the Lemma by summing over all dyadic time scales as in the proof of Lemma~\ref{l.coupling.D}.  
\end{proof}

\section{Subadditive quantities and basic properties}
\label{s.subadditive}

In this section we introduce two subadditive energy quantities related to the variational formulation of the Helffer-Sj\"ostrand equation described in Section~\ref{ss.varchar}. These quantities are analogous to the ones introduced in~\cite[Chapter 2]{AKMbook}. They represent, respectively, the energy of the Dirichlet problem with affine boundary data $\ell_p(x):=p\cdot x$ and that of the Neumann problem with boundary flux~$\nabla \ell_q$. These quantities are \emph{subadditive} and therefore converge as the side length of the cube~$Q$ becomes large. Also, they are in some sense \emph{dual} to each other in a convex analytic sense: in particular, we will discover that they converge to a pair of convex conjugate functions as the size of the cube~$Q$ becomes large. It is this duality that makes it possible to obtain quantitative results. The main focus of this and the next section is to implement a multiscale iteration procedure to obtain an estimate of the convergence rate of this limit.

\smallskip

Throughout the section and the next one, we fix a parameter~$\mathsf{K}_0\in [0,\infty)$ and a H\"older continuous $\f:\R\to \R^d$, with coordinates~$\f=(\f_i)_{i\in\{1,\ldots,d\}}$, satisfying
\begin{equation} 
\label{e.figamma}
\sup_{i\in \{1,\ldots,d\}} 
\left\| \f_i \right\|_{C^{0,\gamma}(\R)} 
\leq 
\mathsf{K}_0. 
\end{equation}
Abusing notation, we write $\f(e,\phi) := \f_i(\nabla\phi(e))$, where $e = (x+e_i,x)$ is an edge in the~$i$th coordinate direction. For convenience we take the H\"older exponent~$\gamma$ to be the same as the one the assumption for $\mathsf{V}''$. We remark that our analysis in the case~$\f=0$ will suffice to prove Theorem~\ref{t.surfacetension}. Our reason for including a general~$\f$ is because the analysis here has broader implications than just the proof of Theorem~\ref{t.surfacetension} because it essentially yields quantitative homogenization estimates the Helffer-Sj\"ostrand equation. This will be investigated in more detail a forthcoming work, and the need for the inclusion of a more general~$\f$ will be apparent there.

\smallskip

We note that the cube~$Q_L:= \left[ 0,L \right]^d \cap \Zd$ has $(L+1)^d$ vertices. The set $\mathcal{E}(Q_L)$ of interior edges of $Q_L$ has exactly $dL(L-1)^{d-1}$ elements: there are $L(L-1)^{d-1}$ edges in each of the~$d$ unit directions. It is convenient to think of this number as the ``volume'' of $Q_L$, and for this reason we denote, by abuse notation,
\begin{equation*} \label{}
\left|Q_L \right| := L(L-1)^{d-1}. 
\end{equation*}
Recalling that $\ell_p$ is the affine function $\ell_p(x):=p\cdot x$, this allows us to write, for example, for every $p,q\in\Rd$,  
\begin{equation} 
\label{e.goodnormies}
\frac1{|Q_L|} 
\sum_{e\in\mathcal{E}(Q_L)} 
\nabla \ell_q(e) \nabla \ell_p(e) 
= p\cdot q.
\end{equation}
We say that $Q\subseteq
\Zd$ is a \emph{cube} if~$Q=z+Q_L$ for some $z\in\Zd$ and $L\in\N$ with $L\geq 2$. In this case we denote $|Q|:=|Q_L|$.

\smallskip

We now define subadditive quantities with respect to the infinite volume Gibbs measure $\mu_{\infty,\xi}$. Throughout the rest of the section we fix $R\in [1,\infty)$ and a $\xi\in B_R$, and will hide the dependence of $\xi$ from notation. We also denote, for short, 
\begin{equation*}
\mu:= \mu_{\infty,\xi}.
\end{equation*}
For every cube $Q \subseteq\Zd$ and $u,v\in H^1(Q,\mu)$, we define 
\begin{equation}
\label{e.defBU}
\mathsf{B}_{Q}\left[u,v\right]
:= 
\frac{1}{|Q|}
\sum_{y\in \Zd} 
\sum_{x\in Q^\circ}
\left( \partial_y u(x,\cdot), \partial_y v(x,\cdot)\right)_{\mu} 
+ \frac1{|Q|}\sum_{e\in \mathcal{E}(Q)} \left\langle 
 \nabla u(e,\cdot) \a(e,\cdot) \nabla v(e,\cdot) 
\right\rangle_\mu
\end{equation}
and
\begin{align*}
\mathsf{E}_{Q,\f} \left[u\right]
& 
:=
\frac12 \mathsf{B}_{Q}\left[u,u\right]
- \frac1{|Q|}\sum_{e\in \mathcal{E}(Q)}
\left\langle \f(e,\phi) \nabla w(e,\cdot)  \right\rangle_{\mu} .
\end{align*}
%
The subadditive energy quantities are defined, for every cube~$Q$ and~$p,q\in\Rd$, by
\begin{equation}
\label{e.defnu}
\nu(Q,\mathbf{f},p):= 
\inf_{v \in \ell_p + H^1_0(Q,\mu)} 
\mathsf{E}_{Q,\f} \left[ v \right]. 
\end{equation}
and
\begin{equation}
\label{e.defnustar}
\nu^*(Q,\mathbf{f},q):=
\sup_{u \in H^1(Q,\mu)} 
\left( 
\frac{1}{|Q|} \sum_{e\in \mathcal{E}(Q)} 
 \nabla \ell_q(e)
\left\langle  \nabla u(e,\cdot) \right\rangle_\mu
-
\mathsf{E}_{Q,\f} \left[ u \right]
\right).
\end{equation}
It is clear that the infimum in~\eqref{e.defnu} and the supremum in~\eqref{e.defnustar} are attained, by the coercivity of the energy functionals with respect to~$H^1(Q,\mu)$, a consequence of the Poincar\'e inequalities in Lemma~\ref{l.spectralgap.UmuL}. The optimizing functions are unique---up to an additive constants in the case of~$\nu^*$.
We may therefore denote by $v(\cdot, Q,\mathbf{f},p)$
the minimizer of $\nu(Q,\mathbf{f},p)$
and by $u(\cdot,Q,\mathbf{f},q)$ 
the maximizer of $\nu^*(Q,\mathbf{f},q)$. We choose the additive constant for $u(\cdot,Q,\mathbf{f},q)$ in such a way that
\begin{equation} 
\label{e.meanzerou}
\sum_{x\in Q} \left\langle u(x,\cdot,Q,\mathbf{f},q) \right\rangle_{\mu}
= 0. 
\end{equation}
These quantities are translation invariant: for every cube $Q$, $p,q\in\Rd$ and $z\in\Zd$, 
\begin{equation} 
\label{e.stationarity}
\nu(Q,\f,p) = \nu(z+Q,\f,p) 
\quad \mbox{and} \quad 
\nu^*(Q,\f,q) = \nu^*(z+Q,\f,q) 
\end{equation}
This is a consequence of the translation invariance of the infinite-volume measure~$\mu$.

We devote the rest of this section to collecting some basic properties of the quantities $\nu$ and $\nu^*$. We show first that they are bounded above and below by quadratic functions (see~\eqref{e.nusbounds}) and that they satisfy a Fenchel-type inequality (see~\eqref{e.fenchel}).

\begin{lemma}[Boundedness and Fenchel inequality]
\label{l.nusbounds}
Let $Q\subseteq\Zd$ be a cube. For every~$p,q\in\Rd$, 
\begin{equation} 
\label{e.fenchel}
\nu(Q,\f,p)+\nu^*(Q,\f,q) \geq p\cdot q
\end{equation}
and, for a constant~$C(\data)<\infty$,
\begin{equation} 
\label{e.nusbounds}
\left\{
\begin{aligned}
&
\frac1C |p|^2 -C\mathsf{K_0}|p|
\leq 
\nu(Q,\f,p) 
\leq 
C|p|^2 + C\mathsf{K}_0|p|,
\\ & 
\frac1C |q|^2 -C\mathsf{K}_0 |q|
\leq 
\nu^*(Q,\f,q) 
\leq C(|q|+\mathsf{K}_0 )^2.
\end{aligned}
\right.
\end{equation}

\end{lemma}
\begin{proof}
It is clear by~\eqref{e.goodnormies} that any function~$w\in \ell_p+H^1_0(Q,\mu)$ satisfies, for every $\phi\in\Omega$ and~$q\in\Rd$, 
\begin{equation} 
\label{e.stokesie}
\frac1{|Q|} 
\sum_{e\in \mathcal{E}(Q)} 
\nabla \ell_q(x)\nabla w(e,\phi) = p\cdot q. 
\end{equation}
Therefore by testing the definition of~$\nu^*(Q,\f,q)$ with the minimizer~$v(\cdot,Q,\f,p)$ of~$\nu(Q,\f,p)$, we obtain~\eqref{e.fenchel}. 
Testing the definition of~$\nu(Q,\f,p)$ with~$\ell_p$ yields
\begin{align*} \label{}
\nu(Q,\f,p) 
&
\leq 
\frac1{|Q|} \sum_{e\in \mathcal{E}(Q)}
\left(
\frac12 \left\langle\a(e,\cdot)\right\rangle_\mu (\nabla \ell_p(e))^2
- \nabla \ell_p(e) \left\langle \f(e,\cdot\right\rangle_\mu \right)
\leq 
\frac12 \Lambda |p|^2 + d\mathsf{K}_0 |p|.
\end{align*}
To get the upper bound for $\nu^*$, we use Cauchy's inequality to find, for every $w\in H^1(Q,\mu)$,
\begin{align*}
\lefteqn{
\frac{1}{|Q|} 
\left( 
\sum_{e\in \mathcal{E}(Q)} 
 \nabla \ell_q(e)
\left\langle  \nabla w(e,\cdot) \right\rangle_\mu
-
\mathsf{E}_{Q,\f} \left[ w \right]
\right).
} \quad & 
\\ & 
\leq 
\frac{1}{|Q|} 
\sum_{e\in \mathcal{E}(Q)} 
\left(
\frac1{2\lambda}\left\langle (\nabla\ell_q(e) + \f(e,\cdot) )^2 \right\rangle_\mu
+
\frac12\lambda \left\langle \nabla u(e))^2 \right\rangle_\mu
-\frac12 \left\langle\a(e) (\nabla u(e))^2 \right\rangle_\mu
\right)
\\ & 
\leq 
\frac{d}{2\lambda} \left( |q| + \mathsf{K}_0 \right)^2
\leq C(|q|+\mathsf{K}_0 )^2.
\end{align*}
Taking the supremum over $w\in H^1(Q,\mu)$ yields the desired upper bound for~$\nu^*$. 
The lower bounds for $\nu$ and $\nu^*$ follow from the upper bounds and~\eqref{e.fenchel}. We have 
\begin{align*}
\nu(Q,\f,p) 
\geq \sup_{q\in\Rd} \left( p\cdot q - \nu^*(Q,\f,q) \right)
& 
\geq 
\sup_{q\in\Rd} \left( p\cdot q - C\left(|q|+\mathsf{K}_0\right)^2 \right)
\\ & 
\geq \frac1C |p|^2 - C\mathsf{K}_0 |p|
\end{align*}
and
\begin{align*}
\nu(Q,\f,p) 
\geq \sup_{p\in\Rd} \left( p\cdot q - \nu(Q,\f,p) \right)
&
\geq \sup_{p\in\Rd} \left( p\cdot q - C|p|^2 - C\mathsf{K}_0|p| \right)
\\ & 
\geq 
\frac1C |q|^2 - C\mathsf{K}_0|q|. \qedhere
\end{align*}
\end{proof}

We show next that~$\nu$ and $\nu^*$ are actually quadratic polynomials and compute the first and second variations of their defining optimization problems. The following lemma is analogous to~\cite[Lemma 2.2]{AKMbook}. 

\begin{lemma}
[Basic properties of $\nu$ and $\nu^*$]
\label{l.basicprops}
Fix a cube~$Q \subseteq \Zd$ and~$\f\in L^\infty(\Zd \times\Omega)$. The quantities $\nu(Q,\f,p)$ and $\nu^*(Q,\f,q)$ and their respective optimizing functions $v(\cdot,Q,\f,p)$ and $u(\cdot,Q,\f,q)$ satisfy the following properties.

\begin{itemize}

\item \emph{Quadratic representation.}
There exist symmetric matrices $\ahom(Q), \ahom_*(Q) \in \R^{d\times d}$, vectors $\fhom(Q,\f), \fhom_*(Q,\f) \in \Rd$ and  $\chom(Q,\f), \chom_*(Q,\f) \in [0,\infty)$ such that
\begin{equation}
\label{e.quadrep}
\left\{
\begin{aligned}
& \nu(Q,\f,p)
= \frac12 p\cdot \ahom(Q) p - \fhom(Q,\f) \cdot p - \chom(Q,\f) \quad \forall p\in\Rd, 
\\ & 
\nu^*(Q,\f,q) 
= \frac 12 \left( q+\fhom_*(Q,\f) \right) \cdot \ahom_*^{\,-1}(Q) \left( q+\fhom_*(Q,\f) \right)
+ \chom_*(Q,\f)  \quad \forall q\in\Rd. 
\end{aligned}
\right.
\end{equation}
These are characterized by the identities, which are valid for~$p,p',q,q'\in\Rd$: 
\begin{equation}
\label{e.coeffs}
\left\{
\begin{aligned}
& p' \cdot \ahom(Q) p
=
\mathsf{B}_Q\left[ \ell_{p'}, v(\cdot,Q,\mathbf{0},p)\right],
\\ & 
\fhom(Q,\f)\cdot p = 
-  \mathsf{B}_Q \left[ \ell_p, v(\cdot,Q,\f,0)  \right]
+\frac1{|Q|} \sum_{e\in \mathcal{E}(Q)} 
\nabla \ell_p(e) \left\langle 
 \mathbf{f}(e,\cdot)
\right\rangle_\mu,
\\ &
\chom(Q,\f) = -\nu(Q,\f,0),
\end{aligned}
\right.
\end{equation}
and
\begin{equation}
\label{e.mus.coeffs}
\left\{
\begin{aligned}
& q' \cdot \ahom_*^{\,-1}(Q) q
=
\frac1{|Q|} \sum_{e\in \mathcal{E}(Q)}
 \nabla\ell_{q'}(e) \left\langle \nabla u(e,\cdot,Q,\mathbf{0},q) \right\rangle_\mu
\\ & 
\ahom_*^{\,-1} q' \cdot \fhom_*(Q,\f) = 
\frac1{|Q|} \sum_{e\in \mathcal{E}(Q)}
 \nabla\ell_{q'}(e) \left\langle \nabla u(e,\cdot,Q,\mathbf{f},0) \right\rangle_\mu
\\ &
\chom_*(U,\f) = \nu^*(Q,\f,-\overline{\f}_*(Q,\f)).
\end{aligned}
\right.
\end{equation}

\item \emph{First variation.} The optimizing functions are characterized as follows: $v(\cdot,Q,\f,p)$ is the unique element of $\ell_p+H^1_0(Q,\mu)$ satisfying 
\begin{align}
\label{e.firstvar.nu}
\mathsf{B}_Q \left[ v(\cdot,Q,\f,p) ,w \right]
=
\frac1{|Q|} \sum_{e\in \mathcal{E}(Q)} 
\left\langle 
\nabla w(e,\cdot) \mathbf{f}(e,\cdot)
\right\rangle_\mu, \quad \forall w \in H^1_0(Q,\mu);
\end{align}
$u(\cdot,Q,\f,q)$ is the unique element of $H^1(U,\mu)$ satisfying~\eqref{e.meanzerou} and
\begin{align}
\label{e.firstvar.nustar}
\lefteqn{
\mathsf{B}_Q\left[ u(\cdot,Q,\f,q),w \right]
} \quad & \\ &  \notag
=
\frac1{|Q|} \sum_{e\in \mathcal{E}(Q)} 
\left\langle
\nabla w(e,\cdot) \mathbf{f}(e,\cdot) 
\right\rangle_\mu 
+ \frac1{|Q|} \sum_{e\in \mathcal{E}(Q)} 
\nabla \ell_q(e)
\left\langle
\nabla w(e,\cdot) 
\right\rangle_\mu,
 \quad \forall w \in H^1(Q,\mu).
\end{align}

\item \emph{Second variation.}
For every~$w\in \ell _{p}+H_{0}^{1}\left( Q,\mu\right)$, 
\begin{equation}
\label{e.quadresp.nu}
\mathsf{E}_{Q,\f}\left[ w \right] -  \nu(Q,\f,p)
= 
\frac12 \mathsf{B}_Q \left[ v(\cdot,Q,\f,p) -w, v(\cdot,Q,\f,p) -w\right] 
\end{equation}
and, for every $w\in H^1(Q,\mu)$, 
\begin{align}
\label{e.quadresp.nustar}
\lefteqn{
 \nu^*(Q,\f,q)- 
\left( 
\frac1{|Q|}  \sum_{e\in \mathcal{E}(Q)} 
\left\langle  \nabla \ell_q(e) \nabla w(e,\cdot) \right\rangle_\mu
- 
\mathsf{E}_{Q,\f}[w] 
\right)
} \qquad\qquad \qquad  & \\ & \notag
=
\frac12 \mathsf{B}_Q \left[ u(\cdot,Q,\f,q) -w, u(\cdot,Q,\f,q) -w\right].
\end{align}
\end{itemize}
\end{lemma}

\begin{proof}
We fix a cube $Q$ and $p,q\in\Rd$ and set $v:=v(\cdot,Q,\f,p)$, $u:=u(\cdot,Q,\f,q)$, $v_0:=v(\cdot,Q,\mathbf{0},p)$ and $u_0:= u(\cdot,Q,\mathbf{0},q)$ to ease the notation. We also write $\mathsf{E}_\f:=\mathsf{E}_{Q,\f}$ and $\mathsf{B}:= \mathsf{B}_{Q}$, and so forth. 

\smallskip

\emph{Step 1.} We prove the first and second variation formulas. Observe that, for~$t\in\R$ and $w\in H^1_0(Q,\mu)$, 
\begin{align*}
\mathsf{E}_\f \left[ v+tw \right]
&
=
\mathsf{E}_\f \left[ v \right]
+ t\left( \mathsf{B}[v,w] - \frac1{|Q|} \sum_{e\in \mathcal{E}(Q)}
\left\langle \f(e,\cdot) \nabla w(e,\cdot)  \right\rangle_{\mu}
\right) 
+\frac12t^2 \mathsf{B}\left[ w,w \right].
\end{align*}
Using that~$\mathsf{E}_\f \left[ v+tw \right]
\leq \mathsf{E}_\f \left[ v \right]$ for all $t\in\R$, we may divide by $|t|$ in the previous display and send $t\to 0+$ and $t\to 0-$ to find that the coefficient of $t$ in the previous display vanishes. This yields~\eqref{e.firstvar.nu} as well as 
\begin{align*}
\mathsf{E}_\f \left[ v+tw \right]
&
=
\mathsf{E}_\f \left[ v \right]
+\frac12t^2 \mathsf{B}\left[ w,w \right].
\end{align*}
As $|Q|\nu(Q,\f,p) = \mathsf{E}_\f \left[ v \right]$, this also gives~\eqref{e.quadresp.nu}. The argument for~\eqref{e.firstvar.nustar} and~\eqref{e.quadresp.nustar} is similar and we omit it.

\smallskip

\emph{Step 2.} We prove the quadratic representation formula for~$\nu(Q,\f,p)$. 
We first define
$\chom(Q,\f):= - \nu(Q,\f,0)$. Observe that, by~\eqref{e.firstvar.nu}, 
\begin{equation*}
\chom(Q,\f) = -\frac12 \mathsf{B} \left[ v(\cdot,Q,\f,0),v(\cdot,Q,\f,0) \right].
\end{equation*}
It is clear from the first variation characterization of~$v(\cdot,Q,\f,p)$ that  the map
\begin{equation}
\label{e.tildelinearmapping}
p\mapsto 
v(\cdot,Q,\mathbf{0},p) = 
v(\cdot,Q,\f,p) - v(\cdot,Q,\f,0) \quad \mbox{is linear.}
\end{equation}
Moreover, by~\eqref{e.firstvar.nu}, 
\begin{equation}
\label{e.tildeqqq}
\mathsf{B}\left[ v(\cdot,Q,\mathbf{0},p), w \right] = 0, \quad \forall w\in H^1_0(Q,\mu)
\end{equation}
and, by~\eqref{e.quadresp.nu} and~\eqref{e.firstvar.nu},
\begin{align*}
\lefteqn{
 \nu(Q,\f,p) 
} \quad & 
\\ &
=
\mathsf{E}_{\f}\left[ v(\cdot,Q,\mathbf{0},p) \right]
- 
\frac12
\mathsf{B}\left[ v(\cdot,Q,\f,0), v(\cdot,Q,\f,0) \right]
\\ & 
=
\frac12 \mathsf{B}\left[ v(\cdot,Q,\mathbf{0},p), v(\cdot,Q,\mathbf{0},p) \right]
- \frac1{|Q|} \sum_{e\in \mathcal{E}(Q)} 
\left\langle 
\mathbf{f}(e,\cdot) \nabla v(e,Q,\mathbf{0},p)
\right\rangle_\mu
+\chom(Q,\f).
\end{align*}
In view of~\eqref{e.tildelinearmapping}, there exists a symmetric $d$-by-$d$ matrix $\ahom(Q)$ and a vector $\fhom(Q,\f)\in\Rd$ such that, for every $p,p'\in\Rd$, 
\begin{equation*}
\left\{ 
\begin{aligned}
& p'\cdot \ahom(Q)p =   \mathsf{B}_{Q}\left[ v(\cdot,Q,\mathbf{0},p'), v(\cdot,Q,\mathbf{0},p)\right], \\
& \fhom(Q,\f)\cdot p 
=
\frac1{|Q|} 
\sum_{e\in \mathcal{E}(Q)} 
\left\langle 
\mathbf{f}(e,\cdot) \nabla v(e,Q,\mathbf{0},p)
\right\rangle_\mu.
\end{aligned}
\right. 
\end{equation*}
With these definitions we obtain the identity for~$\nu(Q,\f,p)$ in~\eqref{e.quadrep}. To show that~$\ahom(Q)$ and~$\fhom(Q,\f)$ satisfy~\eqref{e.coeffs}, we first observe by~\eqref{e.tildeqqq} that 
\begin{equation*}
\mathsf{B} \left[ v(\cdot,Q,\mathbf{0},p')- \ell_{p'}, v(\cdot,Q,\mathbf{0},p) \right] = 0. 
\end{equation*}
This yields the formula for~$\ahom(Q)$. Using~\eqref{e.firstvar.nu} and~\eqref{e.tildeqqq} a second time, we find 
\begin{align*}
0 
&
=
\mathsf{B} \left[ v(\cdot,Q,\f,0) , v(\cdot,U,\mathbf{0},p) - \ell_p \right]
-
\frac1{|Q|} \sum_{e\in \mathcal{E}(Q)} 
\left\langle 
\left( \nabla v(e,Q,\mathbf{0},p) - \nabla \ell_p(e) \right)\mathbf{f}(e,\cdot)
\right\rangle_\mu
\\ & 
= 
-\mathsf{B}  \left[ v(\cdot,Q,0) , \ell_p \right]
-
\frac1{|Q|}  \sum_{e\in \mathcal{E}(Q)} 
\left\langle 
\left( \nabla v(e,Q,\mathbf{0},p) - \nabla \ell_p(e) \right)\mathbf{f}(e,\cdot)
\right\rangle_\mu.
\end{align*}
Rearranging this expression gives the desired formula for~$\fhom(Q,\f)$. 

\smallskip

\emph{Step 3.} We prove the quadratic representation formula for~$\nu^*(U,\f,q)$. Similar to Step~2 above, we observe that  from the first variation characterization of~$u(\cdot,U,\f,q)$ that  the map
\begin{equation}
\label{e.tildelinearmapping.q}
q\mapsto 
u(\cdot,Q,\mathbf{0},q) = 
u(\cdot,Q,\f,q) - v(\cdot,Q,\f,0) \quad \mbox{is linear.}
\end{equation}
Therefore, in view of the lower bound for $\nu^*$ in~\eqref{e.nusbounds}, there is a symmetric, invertible, $d\times d$ matrix $\ahom_*^{\,-1}(Q)$ such that 
\begin{equation*}
q'\cdot \ahom_*^{\,-1}(Q)q = \mathsf{B}\left[ u(\cdot,Q,\mathbf{0},q'), u(\cdot,Q,\mathbf{0},q)\right].
\end{equation*}
By the first variation formula~\eqref{e.firstvar.nustar}, we see that this is equivalent to the characterization of~$\ahom_*$ in~\eqref{e.mus.coeffs}. As the mapping 
\begin{equation}
\label{e.tildelinearmapping.f}
\f \mapsto u(\cdot,Q,\f,0) \quad \mbox{is linear,}
\end{equation}
we see that there exists $\fhom_*(Q,\f)$ satisfying the identity in~\eqref{e.mus.coeffs}. We then define $\chom_*$ by the identity in~\eqref{e.mus.coeffs}.

\smallskip

To check the second identity in~\eqref{e.quadrep}, we compute, using~\eqref{e.quadresp.nustar} with $w=0$ and~\eqref{e.tildelinearmapping.q} and~\eqref{e.tildelinearmapping.f}, 
\begin{align*}
\lefteqn{
 \nu^*(Q,\f,q-\fhom_*) 
} \quad & 
\\ &
=
\frac12 \mathsf{B} 
\left[ u(\cdot,Q,\f,q-\fhom_*) , u(\cdot,Q,\f,q-\fhom_*) \right]
\\ & 
=
\frac12 
\mathsf{B} 
\left[ u(\cdot,Q,\mathbf{0},q),  u(\cdot,Q,\mathbf{0},q)\right]
+
\frac12 
\mathsf{B}
\left[ u(\cdot,Q,\f,-\fhom_*),u(\cdot,Q,\f,-\fhom_*)\right]
\\ & \quad
+
\mathsf{B}
\left[ u(\cdot,Q,\f,-\fhom_*), u(\cdot,Q,\mathbf{0},q)\right]
\\ & 
=
\frac12q\cdot \ahom_*^{\,-1}q + \chom_*.
\end{align*}
Here we also used the fact that, 
\begin{align*} \label{}
\mathsf{B}
\left[ u(\cdot,Q,\f,-\fhom_*), u(\cdot,Q,\mathbf{0},q)\right] 
&
=\frac1{|Q|} \sum_{e\in \mathcal{E}(Q)} \nabla\ell_q(e) \left\langle \nabla u(\cdot,Q,\f,-\fhom_*) \right\rangle_\mu
= 0,
\end{align*}
which a consequence of~\eqref{e.firstvar.nustar} and the identity
\begin{equation*} \label{}
q' \cdot \ahom_*^{\,-1} \left(q+\f_*\right) 
=
\frac1{|Q|}  \sum_{e\in \mathcal{E}(Q)} \nabla\ell_{q'}(e) \left\langle \nabla u(\cdot,Q,\f,q) \right\rangle_\mu,
\end{equation*}
which is obtained by summing the first two lines of~\eqref{e.mus.coeffs}. 
This completes the proof of the second line of~\eqref{e.quadrep}. 
\end{proof}

By the previous lemma and ~\eqref{e.nusbounds}, for any cube~$Q$ we have the bounds
\begin{equation}
\label{e.coeff.bounds}
\left\{
\begin{aligned}
& \frac1C \Id \leq \ahom_*(Q) \leq \ahom(Q) \leq C\Id, \\
& \left| \fhom(Q,\f) \right| + \left|\fhom_*(Q,\f) \right| \leq C\mathsf{K}_0, \\
& 0\leq \chom(Q,\f) \leq \chom_*(Q,\f) \leq \mathsf{K}_0^2.
\end{aligned}
\right. 
\end{equation}
The optimizing functions~$v(\cdot,Q,\f,p)$ and~$u(\cdot,Q,\f,q)$ in the definitions of~$\nu(Q,\f,p)$ and~$\nu^*(Q,\f,q)$, respectively, can be characterized as the solutions of boundary value problems. Indeed, the first variations~\eqref{e.firstvar.nu} and~\eqref{e.firstvar.nustar} assert that~$v(\cdot,Q,\f,p)$ and $u(\cdot,Q,\f,q)$ satisfy, respectively, the Dirichlet and Neumann problems
\begin{equation}
\label{e.BVP.v}
\left\{ 
\begin{aligned}
& \left( -\L_{\mu} + \nabla^* \a \nabla \right) v(\cdot,Q,\f,p) 
 = \nabla^* \f
& \mbox{in} & \ Q^\circ \times \Omega, 
\\ & 
v(\cdot,Q,\f,p) - \ell_p = 0& \mbox{on} & \ \partial Q \times\Omega,
\end{aligned}
\right.
\end{equation}
and
\begin{equation}
\label{e.BVP.u}
\left\{ 
\begin{aligned}
& \left( -\L_{\mu} + \nabla^* \a \nabla \right) u(\cdot,Q,\f,q) 
 = \nabla^* \f
& \mbox{in} & \ Q^\circ \times \Omega, 
\\ & 
\a \nabla u(\cdot,Q,\f,q) - \f  = \nabla \ell_q & \mbox{on} & \ \partial\mathcal{E}( Q) \times\Omega.
\end{aligned}
\right.
\end{equation}
Compare to the discussion in the last two paragraphs of Section~\ref{ss.wellpose}, in particular~\eqref{e.HSsolbvvvchar} and~\eqref{e.neumannvar}.

\smallskip

We next show that the quantities $\nu$ and~$\nu^*$ are approximately subadditive.

\begin{lemma}[{Subadditivity of $\nu$ and $\nu^*$}]
\label{l.nussubadd}
There exists $C(\data)<\infty$ such that, for every $L,m\in\N$ with $L\geq2$ and $p,q\in\Rd$, 
\begin{equation}
\label{e.nu.subadd}
\nu(Q_{mL},\f,p) 
\leq 
\nu(Q_L,\f,p)
+
C\left( |p| + \mathsf{K}_0 \right)^2L^{-1}
\end{equation}
and
\begin{equation}
\label{e.nus.subadd}
\nu^*(Q_{mL},\f,q) 
\leq 
\nu^*(Q_L,\f,q)
+
C\left( |q|+\mathsf{K}_0 \right)^2L^{-\frac12}.
\end{equation}
\end{lemma}
\begin{proof}
\emph{Step 1.} The proof of~\eqref{e.nu.subadd}. Define a function $v\in \ell_p+H^1_0(Q_{mL},\mu)$ 
\begin{equation*} \label{}
v(x,\phi):= v(x,\phi,z+Q_L,\f,p),\quad x\in z+Q_L, \ z\in L\Zd\cap Q_{mL}, \ \phi\in\Omega. 
\end{equation*}
Note that some $x\in Q_{mL}$ belong to two different cubes of the form $z+Q_L$ for $z\in L\Zd\cap Q_{mL}$, namely those points that lie on the boundaries of the subcubes. However for such points~$x$ the two possible definitions of $v(x,\phi)$ above agree and equal $\ell_p(x)$. Testing the definition of $\nu(Q_{mL},\f,p)$ with $v$ yields 
\begin{align*}
\lefteqn{
\nu(Q_{mL},\f,p)
} \quad & \\ 
& 
\leq 
|Q_{mL}| \mathsf{E}_{Q_{mL},\f} \left[ v \right]
\\ &
=
|Q_{mL}|  \sum_{z\in L\Zd\cap Q_{mL}}
\mathsf{E}_{z+Q_{L},\f} \left[ v \right]
+
\sum_{e\in \mathcal{E}'} \left\langle \a(e) (\nabla \ell_p(e))^2 \right\rangle_\mu
- \sum_{e\in \mathcal{E}'} \left\langle \f(e,\cdot) \nabla \ell_p(e) \right\rangle_\mu,
\end{align*}
we denote $\mathcal{E}': = \mathcal{E}(Q_{mL}) \setminus \cup_{z\in L\Zd\cap Q_{mL}} \mathcal{E}(z+Q_L)$. Here we also used that $\partial_yv(x,\cdot) = \partial_y \ell_p(x)= 0$ for every $x\in z+\partial Q_L$ with $z\in L\Zd\cap Q_{mL}$. Since the number of elements of $\mathcal{E}'$ is at most $Cm^dL^{d-1}$, we have that 
\begin{align*} \label{}
\sum_{e\in \mathcal{E}'} \left\langle \a(e) (\nabla \ell_p(e))^2 \right\rangle_\mu
- \sum_{e\in \mathcal{E}'} \left\langle \f(e,\cdot) \nabla \ell_p(e) \right\rangle_\mu
&
\leq 
Cm^dL^{d-1}
\left( |p|^2 + \mathsf{K}_0|p| \right)
\\ & 
\leq Cm^dL^{d-1} \left( |p| + \mathsf{K}_0 \right)^2.
\end{align*}
By the definition of $v$ and~\eqref{e.stationarity}, we also have 
\begin{equation*} \label{}
\sum_{z\in L\Zd\cap Q_{mL}}
\mathsf{E}_{z+Q_{L},\f} \left[ v \right]
=
m^d \nu(Q_L,\f,p).
\end{equation*}
Combining these, we obtain 
\begin{equation*} \label{}
\nu(Q_{mL},\f,p)
\leq \frac{m^d |Q_L|}{|Q_{mL}|}\nu(Q_L,\f,p) + C\left( |p| + \mathsf{K}_0 \right)^2L^{-1}.
\end{equation*}
Since $m^d|Q_L|\leq |Q_{mL}|$, we obtain~\eqref{e.nu.subadd}.

\smallskip

\emph{Step 2.} The proof of~\eqref{e.nus.subadd}.
Testing the definition of $\nu^*(z+Q_L,\f,q)$ with the function $u(\cdot,Q_{mL},\f,q)$ and summing over~$z\in L\Zd\cap Q_{mL}$ yields, in view of~\eqref{e.stationarity}, 
\begin{align*}
\lefteqn{
m^d |Q_L| \,\nu^*\!(Q_L,\f,q) 
} \quad & 
\\ &
= 
\sum_{z\in L\Zd\cap Q_{mL}}
|Q_L| \,\nu^*\!(z+Q_L,\f,q) 
\\ & 
\geq
\sum_{z\in L\Zd\cap Q_{mL}}
\left( 
\sum_{e\in \mathcal{E}(z+Q_L)} 
 \nabla \ell_q(e)
\left\langle  \nabla u(e,\cdot,Q_{mL},\f,q) \right\rangle_\mu
-
|Q_{L}| \mathsf{E}_{z+Q_L,\f} \left[ u(\cdot,Q_{mL},\f,q) \right]
\right)
\\ & 
\geq |Q_{mL}| \,\nu^*\!(Q_{mL},\f,q) 
- \sum_{e\in \mathcal{E}'}
\left\langle \left(  \nabla \ell_q(e)+ \f(e,\cdot) \right)  \nabla u(e,\cdot,Q_{mL},\f,q) \right\rangle_\mu,
\end{align*}
where~$\mathcal{E}'$ is as in Step~1. Since the number of elements of $\mathcal{E}'$ is at most $Cm^dL^{d-1}$, we find that 
\begin{align*}
\lefteqn{
\left|
\sum_{e\in\mathcal{E}'}
\left\langle \left(  \nabla \ell_q(e)+ \f(e,\cdot) \right)  \nabla u(e,\cdot,Q_{mL},\f,q) \right\rangle_\mu
\right|
} \quad & 
\\ & 
\leq
C\left( |q|+\mathsf{K}_0 \right) 
\left( m^dL^{d-1} \right)^{\frac12} 
\left( 
\sum_{e\in \mathcal{E}(Q_{mL})} \left\langle (\nabla u(e,\cdot,Q_{mL},\f,q) )^2 \right\rangle_\mu
\right)^{\frac12}
\\ & 
\leq 
C\left( |q|+\mathsf{K}_0 \right)^2 
\left| Q_{mL} \right|^{\frac12} 
\left( m^dL^{d-1} \right)^{\frac12} 
= CL^{-\frac12} \left( |q|+\mathsf{K}_0 \right)^2 
\left| Q_{mL} \right|.
\end{align*}
Combining these and using that $m^d|Q_L|\leq |Q_{mL}|$, we obtain
\begin{align*} \label{}
\nu^*\!(Q_{mL},\f,q) 
&
\leq 
\frac{m^d |Q_L|}{|Q_{mL}|} 
\,\nu^*\!(Q_L,\f,q) 
+
C\left( |q|+\mathsf{K}_0 \right)^2L^{-\frac12}
\\ &
\leq
\nu^*\!(Q_L,\f,q) 
+
C\left( |q|+\mathsf{K}_0 \right)^2L^{-\frac12}. \qedhere
\end{align*}
\end{proof}

The proof of the previous lemma combined with and~\eqref{e.quadresp.nu} and~\eqref{e.quadresp.nustar} also yields the following estimates:
\begin{multline}
\label{e.subaddresp.nu}
\sum_{z\in L\Zd\cap Q_{mL}} \!\!\!\!
\mathsf{B}_{z+Q_m} \!\left[ v(\cdot,Q_{mL},\f,p) -  v(\cdot,z+Q_{L},\f,p), v(\cdot,Q_{mL},\f,p) -  v(\cdot,z+Q_{L},\f,p)\right] 
\\
\leq 
C \left( \nu(Q_L,\f,p) - \nu(Q_{mL},\f,p) \right)
+ C \left( |p|+\mathsf{K}_0 \right)^2 L^{-1}
\end{multline}
and
\begin{multline}
\label{e.subaddresp.nustar}
\sum_{z\in L\Zd\cap Q_{mL}} \!\!\!\!
\mathsf{B}_{z+Q_m} \!\left[ u(\cdot,Q_{mL},\f,q) -  u(\cdot,z+Q_{L},\f,q), u(\cdot,Q_{mL},\f,q) -  u(\cdot,z+Q_{L},\f,q)\right] 
\\
\leq 
C \left( \nu^*(Q_L,\f,p) - \nu^*(Q_{mL},\f,p) \right)
+ C \left( |q|+\mathsf{K}_0 \right)^2 L^{-\frac12}.
\end{multline}
We will often find it convenient to work with triadic cubes, and for this reason we denote, for every $m\in\N$ with $m\geq 1$, 
\begin{equation}
\label{e.triad}
\cu_m:= \left[ -3^m, 3^m \right]^d\cap \Zd.  
\end{equation}
As a consequence of the monotonicity in Lemmas~\ref{l.nussubadd}, the following limits exist:  
\begin{equation}
\label{e.homs}
\left\{
\begin{aligned}
& \overline{\nu}(\f,p) := \lim_{m\to \infty} \nu(\cu_m,\f,p),
\\
& \overline{\nu}^*\!(\f,q) := \lim_{m\to \infty} \nu^*\!(\cu_m,\f,p)
\end{aligned}
\right.
\end{equation}
By Lemma~\ref{l.nusbounds}, they satisfy 
\begin{equation} 
\label{e.fenchel.bar}
\overline{\nu}(\f,p)+\overline{\nu}^*\!(\f,q) \geq p\cdot q
\end{equation}
and
\begin{equation} 
\label{e.nusbounds.bar}
\left\{
\begin{aligned}
&
\frac1C |p|^2 -C\mathsf{K_0}|p|
\leq 
\overline{\nu}(Q,\f,p) 
\leq 
C|p|^2 + C\mathsf{K}_0|p|,
\\ & 
\frac1C |q|^2 -C\mathsf{K}_0 |q|
\leq 
\overline{\nu}^*(Q,\f,q) 
\leq C(|q|+\mathsf{K}_0 )^2
\end{aligned}
\right.
\end{equation}
By~\eqref{e.homs}, the following limits also exist:
\begin{equation}
\label{e.coeff.homs}
\left\{ 
\begin{aligned}
& \ahom := \lim_{m\to \infty} \ahom(\cu_m), \
\fhom := \lim_{m\to \infty} \f(\cu_m,\f), \
\chom := \lim_{m\to \infty} \chom(\cu_m,\f), \\
& \ahom_* := \lim_{m\to \infty} \ahom_*(\cu_m), \
\fhom_* := \lim_{m\to \infty} \f_*(\cu_m,\f), \
\chom_* := \lim_{m\to \infty} \chom_*(\cu_m,\f),
\end{aligned}
\right. 
\end{equation}
and we have the formulas 
\begin{equation}
\label{e.quadrep.homs}
\overline{\nu}(\f,p) = \frac12p\cdot \ahom p -\fhom\cdot p -\chom 
\quad \mbox{and} \quad 
\overline{\nu}^*\!(\f,q) = \frac12(q+\fhom_*)\cdot \ahom_*^{\,-1} (q+\fhom_*) + \chom_*.
\end{equation}
We will prove below in Proposition~\ref{p.convergence} that in fact $p\mapsto \overline\nu(\f,p)$ and $q\mapsto \overline{\nu}^*\!(\f,q)$ are convex dual functions. Since 
\begin{equation}
\label{e.duals}
\sup_{p\in\Rd} \left( p\cdot q - \left( \frac12p\cdot \ahom p -\fhom\cdot p -\chom \right) \right) 
=
\frac12(q+\fhom)\cdot \ahom^{\,-1} (q+\fhom) + \chom,
\end{equation}
this is equivalent to the statement that the two sets of limiting coefficients are equal:~$(\ahom,\fhom,\chom) = (\ahom_*,\fhom_*,\chom_*)$. The reason for defining the coefficients the way we did in Lemma~\ref{e.quadrep} is due to~\eqref{e.duals} and the fact that we expect~$\nu(Q,\f,p)$ and~$\nu^*(Q,\f,q)$ to converge to a pair of convex dual functions in the large-cube limit.

\section{Quantitative convergence of the subadditive quantities}
\label{s.convergence}

The main purpose of this section is to prove the following result concerning the convergence rate of the subadditive quantities to their limits. It can be compared to~\cite[Theorem 11.4]{AKMbook}. Recall that~$\cu_m$ is the triadic cube defined in~\eqref{e.triad}. As in the previous section, throughout let $R\in [1,\infty)$ and $\xi \in B_R$ be fixed throughout this section and we denote $\mu:= \mu_{\infty,\xi}$ and $\mu_{L} := \mu_{L,\xi}$ and so forth. 

\begin{proposition}[{Convergence of $\nu$ and $\nu^*$}]
\label{p.convergence}
There exist an exponent~$\beta(\data)\in \left(0,\tfrac12\right]$ and a constant~$C(\mathsf M, R,\data)<\infty$ such that, for every $p,q\in\Rd$ and $m\in\N$,
\begin{equation}
\label{e.convergence}
\left|\nu(\cu_m,\f,p)  - \overline{\nu}(\f,p) \right| 
+ \left| \nu^*(\cu_m,\f,q) - \overline{\nu}^*(\f,q) \right| 
\leq C \left( |q|+\mathsf{K}_0\right)^2 3^{-m\beta}. 
\end{equation}
Moreover, the coefficients defined in~\eqref{e.coeff.homs} satisfy $\ahom=\ahom_*$, $\fhom=\fhom_*$ and $\chom=\chom_*$. 
\end{proposition}

The proof of Proposition~\ref{p.convergence} is based on ideas first developed in~\cite{AS},  closely following the presentation in~\cite[Chapter 2]{AKMbook}. For convenience, we define the quantity
\begin{equation*} \label{}
J(Q,p,q):= \nu(Q,\f,p) + \nu^*(Q,\f,q) - p\cdot q.
\end{equation*}
Note that $J(Q,p,q)\geq 0$ by~\eqref{e.fenchel}. This quantity measures the sharpness of the inequality~\eqref{e.fenchel}. If this inequality is sharp---in the sense that for every $p$ there exists a $q$ such that we have equality in~\eqref{e.fenchel}---then $\nu$ and $\nu^*$ are a pair of convex dual functions. We will prove Proposition~\ref{p.convergence} by arguing that, for each $p\in\Rd$,
\begin{equation}
\label{e.convexdualdefectest}
\inf_{q\in\Rd} J(Q_L,p,q) 
\leq 
C \left(|p|+\mathsf{K}_0\right)^2 L^{-\beta}. 
\end{equation}
Moreover, we will show that the infimum is achieved for~$q$ close to~$\ahom p - \f$ (see Lemma~\ref{l.iterateup} below). 
The estimate~\eqref{e.convergence} and the full statement of Proposition~\ref{p.convergence} follows easily from this. We will prove~\eqref{e.convexdualdefectest}, following~\cite[Chapter 2]{AKMbook}, by an iteration of the scales. As we pass to larger scales, we essentially show that the size of $J(Q,p,q)$, for an appropriate choice of~$q$ depending on~$p$, must contract by a multiplicative factor less than one. 

\smallskip

We define the \emph{subadditivity defect at scale~$3^m$} by 
\begin{equation*}
\tau_m :=
\sup_{p \in B_1} 
\left( \nu(\cu_{m},\f,p) - \nu(\cu_{m+1},\f,p) \right)_+ 
+ 
\sup_{q \in B_1} \left( \nu^*(\cu_{m},\f,q) - \nu^*(\cu_{m+1},\f,q) \right)_+.
\end{equation*}
Observe that, for any $p,q \in\Rd$, we have that 
\begin{equation*}
J(\cu_m,\f,p,q) - J(\cu_{m+1},\f,p,q) 
\leq 
C \left( |p|+|q|+\mathsf{K}_0\right)^2 \tau_m .
\end{equation*}
For the rest of this section, we fix a function $\f : \R\to \Rd$ satisfying~\eqref{e.figamma} and allow ourselves to drop dependence on~$\f$ from the notation (e.g., we write $\nu(\cu_m,p)$ in place of $\nu(\cu_m,\f,p)$).

\smallskip

In the following lemma, we use the coupling between~$\mu$ and~$\mu_L$ and the spectral gap for~$\mu_L$ (Lemma~\ref{l.spectralgap.muL}) to control the fluctations of the functions in the~$\phi$ variable. 

\begin{lemma}
\label{l.verticalosc}
There exist $\beta(\data)\in \left(0,\tfrac12\right]$, $m_0(R,\mathsf{M},\data)<\infty$,  $C(\data)<\infty$ and $C'(R,\mathsf{M},\data)<\infty$ such that, for every $p,q\in\Rd$ and $m\in\N$ with $m\geq m_0$, 
\begin{multline}
\label{e.verticalosc.v}
\frac1{|\cu_m|} 
\sum_{x\in \cu_m} 
\left\langle 
\left( v(x,\cdot,\cu_m,p) - \left\langle v(x,\cdot,\cu_m,p) \right\rangle_{\mu} \right)^2 \right\rangle_\mu 
\\
\leq 
C  \left(|p| + \mathsf{K}_0\right)^23^{2m}
\left( C' 3^{-\beta m} +\sum_{n=0}^m 3^{-\beta(m-n)}\tau_n \right)
\end{multline}
and
\begin{multline}
\label{e.verticalosc.u}
\frac1{|\cu_m|} 
\sum_{x\in \cu_m} 
\left\langle 
\left( u(x,\cdot,\cu_m,q) - \left\langle u(x,\cdot,\cu_m,q) \right\rangle_{\mu} \right)^2 \right\rangle_\mu 
\\
\leq 
C 3^{2m} \left(|q| + \mathsf{K}_0\right)^2 \left( 
C'3^{-\beta m} +\sum_{n=0}^m 3^{-\beta(m-n)}\tau_n  \right).
\end{multline}
\end{lemma}
\begin{proof}
We will give the proof only of~\eqref{e.verticalosc.u}, since the one for~\eqref{e.verticalosc.v} is essentially identical. To shorten the expressions below, we drop~$q$ from the notation it plays no role in the argument, writing for instance $u_L(x,\phi,\cu_m)$ in place of $u_L(x,\phi,\cu_m,q)$. We also write $C$ in place of $C(|q|+\mathsf{K}_0)^2$. Throughout we let $C$ and $C'$ denote constants which may vary in each occurrence and depend, respectively, on $(\data)$ and $(R,\mathsf{M},\data)$. 

\smallskip

We first work with the finite-volume measure~$\mu_L$ for~$L\in\N$ with $L\geq 3^{m+1}$ to be selected below. By the triangle inequality, 
\begin{align*}
\lefteqn{
\left( \frac1{|\cu_m|} 
\sum_{x\in \cu_m} 
\var_{\mu_L} \left[ u_L(x,\cdot,\cu_m) \right] \right)^{\frac12}   
} \ \ & 
\\ &
\leq 
\left( \frac1{|\cu_m|} 
\sum_{z\in 3^{m-1} \Zd\cap \cu_m}
\sum_{x\in z+\cu_{m-1}} 
\var_{\mu_L} \left[ u_L(x,\cdot,\cu_m) - u_L(x,\cdot,z+\cu_{m-1}) \right] \right)^{\frac12} 
\\ & \quad 
+
\left( \frac1{|\cu_m|} 
\sum_{z\in 3^{m-1} \Zd\cap \cu_m}
\sum_{x\in z+\cu_{m-1}} 
\var_{\mu_L} \left[ u_L(x,\cdot,z+\cu_{m-1}) \right] \right)^{\frac12}.
\end{align*}
By the Poincar\'e inequality for $\mu_L$ (Lemma~\ref{l.spectralgap.muL}) and~\eqref{e.subaddresp.nustar}, we have 
\begin{align*}
\lefteqn{
\frac1{|\cu_m|} 
\sum_{z\in 3^{m-1} \Zd\cap \cu_m}
\sum_{x\in z+\cu_{m-1}} 
\var_{\mu_L} \left[ u_L(x,\cdot,\cu_m) - u_L(x,\cdot,z+\cu_{m-1}) \right] 
} \  & 
\\ &
\leq 
CL^2 \frac1{|\cu_m|} 
\sum_{z\in 3^{m-1} \Zd\cap \cu_m}
\sum_{x\in z+\cu_{m-1}} 
\sum_{y\in Q_L^\circ} \left\langle \left(  \partial_y u_L(x,\cdot,\cu_m) - \partial_y u_L(x,\cdot,z+\cu_{m-1})  \right)^2\right\rangle_{\mu_L}
\\ & 
\leq CL^2 \left( \nu^*_L(\cu_m) - \nu^*_L(\cu_{m-1}) \right) + CL^{\frac 32}
\\ & 
\leq CL^2\left( \tau_m + L^{-2\beta} \right). 
\end{align*}
Combining the previous two displays, we obtain
\begin{align*}
& \left( \frac1{|\cu_m|} 
\sum_{x\in \cu_m} 
\var_{\mu_L} \left[ u_L(x,\cdot,\cu_m) \right] \right)^{\frac12}   
\\ & \qquad 
\leq 
\left( \frac1{|\cu_m|} 
\sum_{z\in 3^{m-1} \Zd\cap \cu_m}
\sum_{x\in z+\cu_{m-1}} 
\var_{\mu_L} \left[ u_L(x,\cdot,z+\cu_{m-1}) \right] \right)^{\frac12}
+
CL\left( \tau_m^{\frac12} + L^{-\beta} \right). 
\end{align*}
Here is where we invoke the coupling results from the previous section. Applying Lemma~\ref{l.coupling.N} and taking $L=3^{m+1}$, we deduce that there exists $C'<\infty$ and $m_0(R,\mathsf{M}, \data)$ such that, if $m\geq m_0$, then
\begin{align*}
\lefteqn{
\left( \frac1{|\cu_m|} 
\sum_{x\in \cu_m} 
\var_{\mu} \left[ u(x,\cdot,\cu_m) \right] \right)^{\frac12} 
} \qquad  & 
\\ & 
\leq 
\left( \frac1{|\cu_{m-1}|} 
\sum_{x\in \cu_{m-1}} 
\var_{\mu} \left[ u(x,\cdot,\cu_{m-1}) \right] \right)^{\frac12}
+
C3^m\left( \tau_m^{\frac12} +C'3^{-m\beta} \right).
\end{align*}
Iterating the previous inequality and using that 
\begin{equation*}
\left( \frac1{|\cu_{1}|} 
\sum_{x\in \cu_{1}} 
\var_{\mu} \left[ u(x,\cdot,\cu_{1}) \right] \right)^{\frac12} \leq C, 
\end{equation*}
we obtain that 
\begin{equation*}
\left( \frac1{|\cu_m|} 
\sum_{x\in \cu_m} 
\var_{\mu} \left[ u(x,\cdot,\cu_m) \right] \right)^{\frac12} 
\leq 
C3^m \sum_{n=0}^m 3^{-(m-n)}\left( \tau_n^{\frac12} + C'3^{-n\beta} \right).
\end{equation*}
Squaring both sides, we obtain
\begin{align*}
\frac1{|\cu_m|} 
\sum_{x\in \cu_m} 
\var_{\mu} \left[ u(x,\cdot,\cu_m) \right]
&
\leq 
C3^{2m} \sum_{n=0}^m 3^{-(m-n)}\left( \tau_n + C'3^{-2n\beta} \right)
\\ & 
\leq 
C3^{2m} \left( C'3^{-2m \beta} 
+ \sum_{n=0}^m 3^{-(m-n)} \tau_n
\right).
\end{align*}
This completes the proof of~\eqref{e.verticalosc.u}.  
\end{proof}

In the following lemma, we compare the spatial averages of $ \left\langle\nabla u(\cdot,\cu_m,q) \right\rangle_\mu$ and~$ \left\langle\nabla v(\cdot,\cu_m,p) \right\rangle_\mu$ on different scales. For the rest of this section, we denote
\begin{equation*}
\left\{ 
\begin{aligned}
& \ahom_{m} := \ahom(\cu_m), \
\fhom_{m} := \fhom(\cu_m), \
\chom_{m}:= \chom(\cu_m),  \\
& \ahom_{*,m} := \ahom_*(\cu_m), \
\fhom_{*,m} := \fhom_*(\cu_m), \
\chom_{*,m}:= \chom_*(\cu_m)
\end{aligned}
\right. 
\end{equation*}
Recall that, by~\eqref{e.mus.coeffs}, for every $m\in\N$,
\begin{equation}
\label{e.blargh}
\ahom_{*,m}^{\,-1}(q-\overline{\f}_{*,m})
=
\frac{1}{|\cu_m|}
\sum_{e \in \mathcal{E}(\cu_m)} 
\left\langle \nabla u(e,\cdot,\cu_m,q) \right\rangle_\mu
\end{equation}
and, by~\eqref{e.stokesie}, 
\begin{equation}
\label{e.blargh.v}
p = \frac{1}{|\cu_m|}
\sum_{e \in \mathcal{E}(\cu_m)} 
\left\langle \nabla v(e,\cdot,\cu_m,p) \right\rangle_\mu.
\end{equation}

\begin{lemma}
\label{l.spatavg}
There exist $\beta(\data)\in \left(0,\tfrac12\right]$ and $C(\data)<\infty$ such that, for every $p,q\in\Rd$ and $m,n\in\N$ with $n<m$,
\begin{equation}
\label{e.spatavgscales}
\left| 
\ahom_{*,m}^{\,-1}(q-\overline{\f}_{*,m})
- 
\ahom_{*,n}^{\,-1}(q-\overline{\f}_{*,n})
\right|^2
\leq 
C\left(|q| + \mathsf{K}_0\right)^2 
\left(  3^{-n\beta} + C \sum_{k=n}^m \tau_k \right),
\end{equation}
\begin{multline}
\label{e.spatavg.u}
\frac{|\cu_n|}{|\cu_m|}
\sum_{z\in 3^n\Zd\cap \cu_m}
\left| 
\frac{1}{|\cu_n|}
\sum_{e \in \mathcal{E}(z+\cu_n)} 
\left\langle \nabla u(e,\cdot,\cu_m,q) 
\right\rangle_\mu
- \ahom_{*,n}^{\,-1}(q-\overline{\f}_{*,n})
\right|^2
\\
\leq 
C\left(|q| + \mathsf{K}_0\right)^2 
\left( 
 3^{-n\beta} + C \sum_{k=n}^m \tau_k \right)
\end{multline}
and
\begin{multline}
\label{e.spatavg.v}
\frac{|\cu_n|}{|\cu_m|}
\sum_{z\in 3^n\Zd\cap \cu_m}
\left| 
\frac{1}{|\cu_n|}
\sum_{e \in \mathcal{E}(z+\cu_n)} 
\left\langle \nabla v(e,\cdot,\cu_m,p) 
\right\rangle_\mu
- p
\right|^2
\\
\leq 
C\left(|p| + \mathsf{K}_0\right)^2 \left(
3^{-n\beta} + C \sum_{k=n}^m \tau_k \right).
\end{multline}
\end{lemma}
\begin{proof}
Observe that, by~\eqref{e.subaddresp.nustar}, 
\begin{align*}
\lefteqn{
\frac{|\cu_n|}{|\cu_m|}
\sum_{z\in 3^n\Zd\cap \cu_m}
\left| 
\frac{1}{|\cu_n|}
\sum_{e \in \mathcal{E}(z+\cu_n)} 
\left\langle \nabla u(e,\cdot,\cu_m,q) 
\right\rangle_\mu
- \ahom_{*,n}^{\,-1}(q-\overline{\f}_{*,n})
\right|^2
} \quad & 
\\ & 
=
\frac{|\cu_n|}{|\cu_m|}
\sum_{z\in 3^{n} \Zd\cap \cu_m}
\left| 
\frac{1}{|\cu_n|}
\sum_{e \in \mathcal{E}(z+\cu_n)} 
\left( \left\langle \nabla u(e,\cdot,\cu_m,q) \right\rangle_\mu
-
\left\langle \nabla u(e,\cdot,\cu_n,q)  \right\rangle_\mu\right)
\right|^2
\\& 
\leq
\frac{|\cu_n|}{|\cu_m|}
\sum_{z\in 3^{n} \Zd\cap \cu_m}
\frac{1}{|\cu_n|}
\sum_{e \in \mathcal{E}(z+\cu_n)} 
 \left\langle \left|  \nabla u(e,\cdot,\cu_m,q) - \nabla u(e,\cdot,\cu_n,q) 
\right|^2
\right\rangle_\mu
\\ & 
\leq
C\left(|q| + \mathsf{K}_0\right)^2 \left( 3^{-\frac n2} + \sum_{k=n}^m \tau_k \right).
\end{align*}
Using the previous display, we also obtain 
\begin{align*}
\lefteqn{
\left| 
\ahom_{*,m}^{\,-1}(q-\overline{\f}_{*,m})
- 
\ahom_{*,n}^{\,-1}(q-\overline{\f}_{*,n})
\right|^2
} \qquad & 
\\ &
=
\left| 
\frac{1}{|\cu_m|}
\sum_{e \in \mathcal{E}(\cu_m)} 
\left\langle \nabla u(e,\cdot,\cu_m,q) \right\rangle_\mu
-
\ahom_{*,n}^{\,-1}(q-\overline{\f}_{*,n})
\right|^2
\\ &
\leq
\frac{|\cu_n|}{|\cu_m|}
\sum_{z\in 3^{n} \Zd\cap \cu_m}
\left| 
\frac1{|\cu_n|}
\sum_{e \in \mathcal{E}(z+\cu_n)} 
\left\langle \nabla u(e,\cdot,\cu_m,q) \right\rangle_\mu
-
\ahom_{*,n}^{\,-1}(q-\overline{\f}_{*,n})
\right|^2
\\ & \qquad 
+ C3^{-n}\left(|q| + \mathsf{K}_0\right)^2
\\ & 
\leq 
C\left(|q| + \mathsf{K}_0\right)^2 \left( 3^{-\frac n2} + \sum_{k=n}^m \tau_k \right).
\end{align*}
This is~\eqref{e.spatavgscales}. Note that in the third line of the previous display we used that the set $\mathcal{E}'_{m,n} :=  \mathcal{E}(\cu_m) \setminus \left( \cup_{z\in 3^n\Zd\cap\cu_m} \mathcal{E}(z+\cu_n) \right)$ has at most $C3^{-n}|\cu_m|$ elements and thus, by the H\"older inequality, 
\begin{align*}
\frac{1}{|\cu_m|} 
\sum_{e \in \mathcal{E}'_{m,n}} 
\left| 
\left\langle \nabla u(e,\cdot,\cu_m,q) \right\rangle_\mu
\right| 
\leq
C 3^{-\frac n2} \left( \sum_{e\in \mathcal{E}(\cu_m)} 
\left\langle \left( \nabla u(e,\cdot,\cu_m,q) \right)^2
\right\rangle_\mu \right)^{\frac12}
  \leq C|q| 3^{-\frac n2} .
\end{align*}
The combination of the previous two displays also gives~\eqref{e.spatavg.u}. The proof of~\eqref{e.spatavg.v} is similar: we substitute~\eqref{e.blargh.v} in the place of~\eqref{e.blargh}.
\end{proof}

We next show that $x\mapsto \left\langle v(x,\cdot,\cu_m,\f,p) \right\rangle_\mu$ and $x\mapsto \left\langle u(x,\cdot,\cu_m,\f,p) \right\rangle_\mu$ are close to affine functions. 

\begin{lemma}
\label{l.L2bracketosc}
There exist $\beta(\data)\in \left(0,\tfrac12\right]$ and $C(R,\mathsf{M},\data)<\infty$ such that, for every $p,q\in \Rd$ and $m\in\N$,
\begin{multline}
\label{e.L2control.u}
\frac{1}{|\cu_m|}
\sum_{x \in \cu_m } \!\!
\left( 
\left(
\left\langle u(x,\cdot,\cu_m,q) \right\rangle_\mu
-
\ahom_{*,m}^{\,-1}(q+\overline{\f}_{*,m}) \cdot x
\right)^2
\right)
\\
\leq 
C3^{2m}
\left(|q| + \mathsf{K}_0\right)^2  \left(  3^{-m\beta} + \sum_{n=0}^m 3^{-2(m-n)} \tau_n \right)
\end{multline}
and
\begin{multline}
\label{e.L2control.v}
\frac{1}{|\cu_m|}
\sum_{x \in \cu_m } \!\!
\left(
\left\langle v(x,\cdot,\cu_m,p) \right\rangle_\mu
-
p \cdot x
\right)^2
\\
\leq 
C3^{2m}
\left(|p| + \mathsf{K}_0\right)^2 \left(  3^{-m\beta} + \sum_{n=0}^m 3^{-2(m-n)} \tau_n \right).
\end{multline}

\end{lemma}
\begin{proof}
By the multiscale Poincar\'e inequality (Proposition \ref{p.MP}), 
\begin{align*}
\lefteqn{
\frac{1}{|\cu_m|}
\sum_{x \in \cu_m } 
\left( 
\left\langle u(x,\cdot,\cu_m,q) \right\rangle_\mu
-
\ahom_{*,m}^{\,-1}(q+\overline{\f}_{*,m}) \cdot x
\right)^2
} \quad & 
\\ & 
\leq
C \frac{1}{|\cu_m|}
\sum_{e\in \mathcal{E}(\cu_m)} 
\left|
\left\langle \nabla u(e,\cdot,\cu_m,q) 
\right\rangle_\mu
- \ahom_{*,m}^{\,-1}(q+\overline{\f}_{*,m})
\right|^2
\\ & \quad
+ C 
\sum_{n=0}^m
3^{2n}
\frac{|\cu_n|}{|\cu_m|}
\sum_{z\in 3^n\Zd\cap \cu_m}
\left| 
\frac{1}{|\cu_n|}
\sum_{e \in \mathcal{E}(z+\cu_n)} 
\left\langle \nabla u(e,\cdot,\cu_m,q) 
\right\rangle_\mu
- \ahom_{*,m}^{\,-1}(q+\overline{\f}_{*,m})
\right|^2.
\end{align*}
We bound the first term on the right crudely:
\begin{align*}
\lefteqn{
\frac{1}{|\cu_m|}\sum_{e\in \mathcal{E}(\cu_m)} 
\left|
\left\langle \nabla u(e,\cdot,\cu_m,q) 
\right\rangle_\mu
- \ahom_{*,m}^{\,-1}(q+\overline{\f}_{*,m})
\right|^2
} \qquad & 
\\ &
\leq
 \frac{2}{|\cu_m|}\sum_{e\in \mathcal{E}(\cu_m)} 
\left\langle \left| \nabla u(e,\cdot,\cu_m,q) \right|^2
\right\rangle_\mu
+2 \left|\ahom_{*,m}^{\,-1}(q+\overline{\f}_{*,m})\right|^2
\leq
C \left(|q| + \mathsf{K}_0\right)^2.
\end{align*}
We estimate the second term using Lemma~\ref{l.spatavg}:
\begin{align*}
& 
\sum_{n=0}^m
3^{2n}
\frac{|\cu_n|}{|\cu_m|}
\sum_{z\in 3^n\Zd\cap \cu_m}
\left| 
\frac{1}{|\cu_n|}
\sum_{e \in \mathcal{E}(z+\cu_n)} 
\left\langle \nabla u(e,\cdot,\cu_m,q) 
\right\rangle_\mu
- \ahom_{*,m}^{\,-1}(q+\overline{\f}_{*,m})
\right|^2.
\\ &  \qquad\qquad\qquad 
\leq
C \left(|q| + \mathsf{K}_0\right)^2 3^{2m}
\sum_{n=0}^m 3^{2n-2m} \left(  3^{-n\beta} + \sum_{k=n}^m \tau_k \right)
\\ & \qquad\qquad\qquad 
\leq C \left(|q| + \mathsf{K}_0\right)^2 3^{2m} \left( 3^{-m\beta} + \sum_{n=0}^m 3^{2(n-m)} \tau_n \right).
\end{align*}
This completes the proof of~\eqref{e.L2control.u}. The proof of~\eqref{e.L2control.v} is similar. 
\end{proof}

Combining Lemmas~\ref{l.verticalosc} and~\ref{l.L2bracketosc} with the triangle inequality, we obtain the following statement. 

\begin{corollary}
\label{c.L2osc}
There exist $\beta(\data)\in \left(0,\tfrac12\right]$, $m_0(R,\mathsf{M},\data)<\infty$,  $C(\data)<\infty$ and $C'(R,\mathsf{M},\data)<\infty$ such that, for every $p,q\in\Rd$ and $m\in\N$ with $m\geq m_0$, 
\begin{multline}
\label{e.L2osc}
\frac{1}{|\cu_m|}
\sum_{x \in \cu_m } 
\left\langle \left(
u(x,\cdot,\cu_m,\ahom_{*,m}p-\overline{\f}_{*,m})
-
v(x,\cdot,\cu_m,p) \right)^2
\right\rangle_\mu
\\
\leq 
C3^{2m} \left(|p| + \mathsf{K}_0\right)^2
\left( C' 3^{-m\beta} + \sum_{n=0}^m 3^{-(m-n)} \tau_n  \right).
\end{multline}
\end{corollary}

We are now in a position to compare minimizers of~$\nu$ to maximizers of~$\nu^*$, enabling us to compare the two quantities. 
The result is summarized in the following lemma. 

\begin{lemma}
\label{l.Jminimalset}
There exist $\beta(\data)\in \left(0,\tfrac12\right]$, $m_0(R,\mathsf{M},\data)<\infty$,  $C(\data)<\infty$ and $C'(R,\mathsf{M},\data)<\infty$ such that, for every $p\in\Rd$ and $m\in\N$ with $m\geq m_0$, 
\begin{equation}
\label{e.Jminimalset}
J(\cu_m,p,\ahom_{*,m}p-\overline{\f}_{*,m}) 
\leq
C \left(|p|+\mathsf{K}_0 \right)^2\left(C' 3^{-m\beta} + \sum_{n=0}^m 3^{-2(m-n)} \tau_n \right). 
\end{equation}
\end{lemma}
\begin{proof}
Fix $m\in\N$, $p\in\Rd$ and let
\begin{equation*} \label{}
w:= u(x,\cdot,\cu_{m+1},\ahom_{*,m}p-\overline{\f}_{*,m})
-
v(x,\cdot,\cu_{m+1},p).
\end{equation*}
Observe that $w$ satisfies the equation
\begin{equation*} \label{}
-\L_{\mu_L}w + \nabla^* \a \nabla w 
 = 0
\quad \mbox{in}  \ \cu_{m+1}^\circ \times \Omega.
\end{equation*}
By the Caccioppoli inequality (Lemma~\ref{l.cacc}) and Corollary~\ref{c.L2osc}, 
\begin{align*}
\mathsf{B}_{\cu_m} \left[ w,w \right] 
& 
\leq 
\frac{C3^{-2m}}{|\cu_{m+1}|}
\sum_{x \in \cu_{m+1} } 
\left\langle \left(
w(x,\cdot) \right)^2
\right\rangle_\mu
\leq 
C  
\left(|p| + \mathsf{K}_0\right)^2 \left( C'3^{-m\beta} + \sum_{n=0}^m 3^{-(m-n)} \tau_n \right).
\end{align*}
To complete the proof, it therefore suffices to show that 
\begin{align} 
\label{e.Jboundwts}
J(\cu_m,p,\ahom_{*,m}p-\overline{\f}_{*,m}) 
\leq
C  \mathsf{B}_{\cu_m} \left[ w,w \right]
+
C\left(|p|+\mathsf{K}_0\right)^2\left( 3^{-m\beta} +  \tau_m  \right).
\end{align}
We define, for each $z\in\Zd$, 
\begin{equation*} \label{}
\tilde{w}_z:= u(x,\cdot,z+\cu_{m},\ahom_{*,m}p-\overline{\f}_{*,m})
-
v(x,\cdot,z+\cu_{m},p)
\end{equation*}
so that by~\eqref{e.quadresp.nustar} we have  
\begin{equation*}
J(z+\cu_m,p,\ahom_{*,m}p-\overline{\f}_{*,m}) 
=
\frac{1}{2}  \mathsf{B}_{z+\cu_m} \left[ \tilde{w}_z,\tilde{w}_z \right].
\end{equation*}
By~\eqref{e.subaddresp.nu} and~\eqref{e.subaddresp.nustar}, we have that \begin{align*}
\lefteqn{
 \mathsf{B}_{\cu_{m+1}} \left[ w - \tilde{w}_0,w - \tilde{w}_0 \right]
} \quad & 
\\ & 
\leq 
\sum_{z\in 3^m\Zd\cap\cu_{m+1}}
\mathsf{B}_{z+\cu_m} \left[ w - \tilde{w}_z,w - \tilde{w}_z \right]
\\ & 
\leq 
J(\cu_m,p,\ahom_{*,m}p-\overline{\f}_{*,m}) - 
J(\cu_{m+1},p,\ahom_{*,m}p-\overline{\f}_{*,m}) + C\left(|p|+\mathsf{K}_0\right)^2 3^{-\frac m2}
\\ & 
\leq 
C\left(|p|+\mathsf{K}_0\right)^2 \left( 3^{-\frac m2} + \tau_m \right).
\end{align*}
We also have that 
\begin{align*}
 \mathsf{B}_{\cu_m}  \left[ w + \tilde{w}_0,w + \tilde{w}_0 \right]
& 
\leq 
C J(\cu_m,p,\ahom_{*,m}p-\overline{\f}_{*,m}) + C J(\cu_{m+1},p,\ahom_{*,m}p-\overline{\f}_{*,m})
\\ & 
\leq
C J(\cu_m,p,\ahom_{*,m}p-\overline{\f}_{*,m}) + C(|p|+\mathsf{K}_0)^2 3^{-\frac m2}.
\end{align*}
Therefore, 
\begin{align*}
\lefteqn{
\left| \mathsf{B}_{\cu_m}  \left[ w - \tilde{w}_0,w+\tilde{w}_0 \right] \right| 
} \quad & \\
&
\leq 
\left( \mathsf{B}_{\cu_m}  \left[ w - \tilde{w}_0,w - \tilde{w}_0 \right]\right)^{\frac12}\left( \mathsf{B}_{\cu_m}  \left[ w + \tilde{w}_0,w + \tilde{w}_0 \right]\right)^{\frac12}
\\ &
\leq 
C\left( |p|+\mathsf{K}_0\right)\left( \tau_m+3^{-\frac m2} \right)^{\frac12} \left( J(\cu_m,p,\ahom_{*,m}p-\overline{\f}_{*,m}) + (|p|+\mathsf{K}_0)^2 3^{-\frac m2} \right)^{\frac12}
\\ & 
\leq 
C\left( |p|+\mathsf{K}_0\right)
\left(  \tau_m^{\frac12} J(\cu_m,p,\ahom_{*,m}p-\overline{\f}_{*,m})^{\frac12} + (|p|+\mathsf{K}_0) 3^{-\frac m2} \right).
\end{align*}
and hence
\begin{align*}
\lefteqn{
2J(\cu_m,p,\ahom_{*,m}p-\overline{\f}_{*,m}) 
} \quad & \\
& 
=
  \mathsf{B}_{\cu_m} \left[ \tilde{w}_0,\tilde{w}_0 \right]
\\ & 
= 
 \mathsf{B}_{\cu_m} \left[ w,w \right]
-\mathsf{B}_{\cu_m}  \left[ w - \tilde{w}_0,w+\tilde{w}_0 \right]
\\ & 
\leq 
 \mathsf{B}_{\cu_m} \left[ w,w \right]
+
C\left( |p|+\mathsf{K}_0\right)
\left(  \tau_m^{\frac12} J(\cu_m,p,\ahom_{*,m}p-\overline{\f}_{*,m})^{\frac12} + (|p|+\mathsf{K}_0) 3^{-\frac m2} \right).
\end{align*}
This implies that 
\begin{equation*}
J(\cu_m,p,\ahom_{*,m}p-\overline{\f}_{*,m}) 
\leq 
\frac{1}{2 }  \mathsf{B}_{\cu_m} \left[ w,w \right]
+
C(|p|+\mathsf{K}_0)^2 \left( \tau_m  + 3^{-\frac m2} \right),
\end{equation*}
which yields~\eqref{e.Jboundwts}. 
\end{proof}

We next iterate the result of the previous lemma to obtain an algebraic rate of convergence of $J(\cu_m,p,\ahom p-\overline{\f})$ to zero. 

\begin{lemma}
\label{l.iterateup}
There exist $\beta(\data)\in \left(0,\tfrac12\right]$ and $C(\mathsf M, R,\data)<\infty$ such that, for every $p\in\Rd$ and $m\in\N$,
\begin{equation}
\label{e.iteratedup}
J(\cu_m,p,\ahom_* p-\overline{\f}_*) 
\leq
C\left(|p|+\mathsf{K}_0 \right)^23^{-m\beta}. 
\end{equation}
\end{lemma}
\begin{proof}
As usual, we let $\{e_1,\ldots,e_d\}$ denote the standard basis and for convenience we set $e_0=0$. We also let $\beta$, $c$, $C$ and $C'$ denote constants depending only on~$(\data)$, ~$(\data)$,~$(\data)$ and~$(R,\mathsf{M},\data)$, respectively, which may vary in each occurrence and $m_0$ be the constant in Lemma~\ref{l.Jminimalset}. 
For every $i\in\{0,\ldots,d\}$ and $m\in\N$, define
$q_i:= \ahom_* e_i+\fhom_*$ and $q_{i,m} := \ahom_{*,m} e_i + \fhom_{*,m}$ as well as 
\begin{equation}
F_m := \sum_{i=0}^d \sum_{n=0}^m 3^{-\frac12\beta(m-n)} 
J(\cu_n,e_i,q_{i,n}),
\end{equation}
where $\beta>0$ is the smaller of the exponents in Lemmas~\ref{l.spatavg} and~\ref{l.Jminimalset}. Observe that, for every $i\in\{0,\ldots,d\}$ and $m,n\in\N$,
we have that $|q_i|+|q_{i,m}| \leq C(1+\mathsf{K}_0)$ and, by~\eqref{e.spatavgscales}, 
\begin{equation}
\label{e.qnqm}
\left| q_{i,n} - q_{i,m} \right|^2 
\leq 
C\left(1 + \mathsf{K}_0\right)^2 
\left(
3^{-n\beta} + \sum_{k=n}^m \tau_k
\right).
\end{equation}

\emph{Step 1.} We show that there exists $c(\data)>0$ such that 
\begin{equation}
\label{e.Jincretaum}
\sum_{i=0}^d 
\left( 
J(\cu_n,e_i,q_{i,n}) - J(\cu_{n+1},e_i,q_{i,n+1})
\right) 
\geq c \tau_m. 
\end{equation}
Using Lemma~\ref{l.basicprops}, we see that, for each $p\in\Rd$, 
the mapping $q' \mapsto J(Q,p,q')$ achieves its minimal at the point~$q' = \ahom_{*}(Q)p + \fhom_*(Q)$. Therefore, 
\begin{align*}
\lefteqn{
\sum_{i=0}^d \left( 
J(\cu_n,e_i,q_{i,n}) - J(\cu_{n+1},e_i,q_{i,n+1}) \right) 
} \quad & 
\\ &
\geq
\sum_{i=0}^d \left( 
J(\cu_n,e_i,q_{i,n}) - J(\cu_{n+1},e_i,q_{i,n})
\right) 
\\ & 
=
\sum_{i=0}^d 
\left( \nu(\cu_n,e_i,0) - \nu(\cu_{n+1},e_i,0) \right)
+
\sum_{i=0}^d 
\left( \nu^*(\cu_n,0,q_{i,n}) - \nu^*(\cu_{n+1},0,q_{i,n}) \right)
\\ & 
\geq c\tau_n. 
\end{align*}

\smallskip

\emph{Step 2.} We show that there exist~$\theta(\data)\in \left[\frac12,1\right)$ and~$C'(\mathsf{M},R,\data)<\infty$ such that, for every $m\geq m_0$, 
\begin{equation}
\label{e.preiterFm}
F_{m+1}  \leq \theta F_m + C'\left( |p|+\mathsf{K}_0 \right)^2 3^{-\frac12\beta m}.
\end{equation}
We have by Lemma~\ref{l.nusbounds} that $F_{m_0} (p) \leq C\left( |p|+\mathsf{K}_0 \right)^2$. Thus by~\eqref{e.Jincretaum},
\begin{align*}
\lefteqn{
F_{m}  - F_{m+1} 
}  & 
\\ &
=  
\sum_{i=0}^d \left(  \sum_{n=m_0}^m 3^{-\frac12\beta(m-n)} J(\cu_n,e_i,q_{i,n})
-
\sum_{n=m_0}^{m+1}  3^{-\frac12\beta(m+1-n)} J(\cu_n,e_i,q_{i,n})
\right)
\\ & 
= 
\sum_{i=0}^d \left( 
\sum_{n=m_0}^m 
3^{-\frac12 \beta(m-n) } 
\left( J(\cu_n,e_i,q_{i,n}) - J(\cu_{n+1},e_i,q_{i,n+1}) \right)
- 
3^{-\frac12 \beta(m+1)} J(\cu_{m_0},e_i,q_{i,0})
\right)
\\ & 
\geq 
c\sum_{n=m_0}^m 
3^{-\frac12 \beta(m-n) } 
\tau_n
- C'\left( 1+\mathsf{K}_0 \right)^2 3^{-\frac12\beta m}.
\end{align*}
Lemma~\ref{l.Jminimalset} asserts that 
\begin{equation}
F_m  
\leq  
C \sum_{n=0}^m 3^{-\frac12 \beta(m-n) } 
\tau_n
+C' \left( |p|+\mathsf{K}_0 \right)^2 3^{-\frac12\beta m}.
\end{equation}
Combining the previous two displays, we obtain 
\begin{align*}
F_{m}  - F_{m+1} 
\geq 
c F_m  
-
C'\left(1+\mathsf{K}_0 \right)^2 3^{-\frac12\beta m}.
\end{align*}
This is~\eqref{e.preiterFm}. 

\smallskip

\emph{Step 3.} The conclusion. 
By an iteration of~\eqref{e.preiterFm}, using that $F_{m_0}(p) \leq C\left( |p|+\mathsf{K}_0 \right)^2$, we obtain, for every $m\geq m_0$, 
\begin{equation}
F_m \leq C' \theta^{m-m_0} \left( 1+\mathsf{K}_0 \right)^2.
\end{equation}
Taking $\alpha(\data)>0$ such that $\theta = 3^{-\alpha}$, we obtain that 
\begin{equation}
F_m  \leq C' 3^{\alpha m_0} 3^{-m\alpha} \left(1+\mathsf{K}_0 \right)^2.
\end{equation}
Since $F_n \leq C'\left(1+\mathsf{K}_0 \right)^2$ for $n\leq m_0$, by enlarging the constant $C'$ we obtain, for every $m\in\N$, 
\begin{equation}
F_m  \leq C' 3^{-m\alpha} \left(1+\mathsf{K}_0 \right)^2.
\end{equation}
By~\eqref{e.Jincretaum}, from this we obtain also that 
\begin{equation}
\label{e.tauwhipped}
\tau_m  \leq C'3^{-m\alpha} \left( 1 +\mathsf{K}_0 \right)^2.
\end{equation}
By Lemma~\ref{l.spatavg} and~\eqref{e.tauwhipped}, we get
\begin{equation}
\left| q_{i,m} - q_i \right|^2 \leq C'3^{-m\alpha}  \left( 1+\mathsf{K}_0 \right)^2.
\end{equation}
As noted above, $q' \mapsto J(\cu_m,p,q')$ achieves its minimum at $q'=q_m$. In view of Lemma~\ref{l.basicprops} and the boundedness of $\ahom_*^{\,-1}$ in~\eqref{e.coeff.bounds}, we deduce that
\begin{align*}
J(\cu_m,e_i,q_i) 
\leq 
J(\cu_m,e_i,q_{i,m}) + C\left|q_i-q_{i,m}\right|^2  
\leq
F_m  + C\left|q_i-q_{i,m}\right|^2 
\leq 
C'3^{-m\alpha}  \left( 1+\mathsf{K}_0 \right)^2.
\end{align*}
This implies~\eqref{e.iteratedup}. 
\end{proof}

We now complete the proof of Proposition~\ref{p.convergence} by showing that the previous lemma implies the bounds~\eqref{e.convergence}.

\begin{proof}[{Proof of Proposition~\ref{p.convergence}}]
Observe that, by~\eqref{e.fenchel},~\eqref{e.nus.subadd} and~\eqref{e.iteratedup},
\begin{align*}
\nu(\cu_m,\f,p) 
& 
= - \nu^*(\cu_m,\f,\ahom_* p-\overline{\f}_*) + p\cdot \left(\ahom_* p-\overline{\f}_* \right) + J(\cu_m,\f,p,\ahom_* p-\overline{\f}_*)
\\ & 
\leq - \nu^*(\cu_m,\f,\ahom_* p-\overline{\f}_*) + p\cdot \left(\ahom_* p-\overline{\f}_* \right) + C\left(|p|+\mathsf{K}_0 \right)^23^{-m\beta}
\\ & 
\leq 
-\overline{\nu}^*(\f,\ahom_* p-\overline{\f}_*) + p\cdot \left(\ahom_* p-\overline{\f}_* \right) + C\left(|p|+\mathsf{K}_0 \right)^23^{-m\beta}
\\ & 
\leq 
\overline{\nu}(\f,p) + C\left(|p|+\mathsf{K}_0 \right)^23^{-m\beta}.
\end{align*}
Combining this with~\eqref{e.nu.subadd}, we obtain, for $\beta(\data)>0$ and $C(R,\mathsf{M},\data)\in [1,\infty)$, 
\begin{equation*}
\left|\nu(\cu_m,\f,p)  - \overline{\nu}(\f,p) \right| 
\leq C\left(|p|+\mathsf{K}_0 \right)^23^{-m\beta}.
\end{equation*}
The bound for $\nu^*$ is also implied by the previous two displays, since they imply that all inequalities are sharp, up to $C\left(|p|+\mathsf{K}_0 \right)^23^{-m\beta}$. This tells us that 
\begin{align}
\label{e.oneconseq}
\left| \nu^*(\cu_m,\ahom_* p - \f_*) 
+ \overline{\nu}^*(\f,\ahom_* p-\overline{\f}_*)  \right| 
\leq C \left( |p|+\mathsf{K}_0\right)^2 3^{-m\beta}
\end{align}
as well as
\begin{align}
\label{e.twoconseq}
\left| \overline{\nu}(\f,p) +\overline{\nu}^*(\f,\ahom_* p-\overline{\f}_*) - p\cdot \left(\ahom_* p-\overline{\f}_* \right)  \right| 
\leq 
C\left(|p|+\mathsf{K}_0 \right)^23^{-m\beta}.
\end{align}
Sending $m\to \infty$ in~\eqref{e.twoconseq}, we get
\begin{align*}
\overline{\nu}(\f,p)
&
=
-\overline{\nu}^*(\f,\ahom_* p-\overline{\f}_*) + p\cdot \left(\ahom_* p-\overline{\f}_* \right)
= \frac 12 p\cdot \ahom_* p - \fhom_*\cdot p -\chom_*. 
\end{align*}
Comparing this with~\eqref{e.quadrep.homs}, we deduce that 
\begin{equation*}
\left(\ahom_*,\overline{\f}_*,\overline{c}_*\right) = 
\left(\ahom,\overline{\f},\overline{c}\right).
\end{equation*}
Taking $p=\ahom^{\,-1}(q+\overline{\f})$, which satisfies $\left( |p|+\mathsf{K}_0\right)^2 \leq C(|q| + \mathsf{K}_0)^2$ by~\eqref{e.coeff.bounds}, we obtain from~\eqref{e.oneconseq} that
\begin{equation*}
\left| \nu^*(\cu_m,q) 
- \overline{\nu}^*(\f,q) \right| 
\leq C \left( |q|+\mathsf{K}_0\right)^2 3^{-m\beta}. 
\end{equation*}
This completes the proof.
\end{proof}

The subadditive quantities $\nu$ and $\nu^*$ are defined in Section \ref{s.subadditive} with respect to the infinite volume Gibbs measure $\mu_\xi$. Since we are interested in the finite volume surface tension, we also study the subadditive quantities associated with finite volume Gibbs measures, defined by 
\begin{equation*}
\nu_L(Q,\mathbf{f},p):= 
\inf_{v \in \ell_p + H^1_0(Q,\mu_{L,\xi})} 
\mathsf{E}_{Q,\mu_{L,\xi},\f} \left[ v \right]. 
\end{equation*}
and
\begin{equation*}
\nu_L^*(Q,\mathbf{f},q):=
\sup_{u \in H^1(Q,\mu_{L,\xi})} 
\left( 
\frac{1}{|Q|} \sum_{e\in \mathcal{E}(Q)} 
 \nabla \ell_q(e)
\left\langle  \nabla u(e,\cdot) \right\rangle_{\mu_{L,\xi}}
-
\mathsf{E}_{Q,\mu_{L,\xi},\f} \left[ u \right]
\right).
\end{equation*}
The other finite volume quantities $\mathsf{B}_{Q,\mu_{L,\xi}}, \mathsf{E}_{Q,\mu_{L,\xi},\f}$,$ (\ahom_{\mu_{L,\xi}}(Q), \fhom_{\mu_{L,\xi}}(Q,\f), \bar c_{\mu_{L,\xi}}(Q,\f))$,  $(\ahom_{*,\mu_{L,\xi}}(Q), \fhom_{*,\mu_{L,\xi}}(Q,\f), \bar c_{*,\mu_{L,\xi}}(Q,\f))$ are defined analogously. 
  
\smallskip
  
A consequence to the comparison lemmas for the Helffer-Sj\"ostrand energies (Lemma \ref{l.coupling.D} and \ref{l.coupling.N}) is the comparison of the quadratic forms stated below. For each finite cube~$Q\subseteq\Zd$ we denote by~$\size(Q)$ the side length of~$Q$. 

\begin{lemma}
\label{l.acouple}
There exist~$\beta(\data)\in \left( 0,\frac12\right]$ and~$C(R,\mathsf M, \data)<\infty$ such that, for every~$\xi \in B_R$ and $L, M \in \N$ with $L\le M$ and every cube $Q\subset Q_L$ satisfying $\sqrt{d} \size(Q) 
\leq 
\dist(Q,\partial Q_L)$, we have
\begin{equation*}
\left|\ahom_{\mu_{M,\xi}}(Q) - \ahom(Q) \right| 
+
\left|\ahom_{*,\mu_{M,\xi}}(Q) - \ahom_{*}(Q)\right|
\leq C\size(Q)^{-\beta}.
\end{equation*}
\end{lemma}
\begin{proof}
By \eqref{e.coeffs}, we have the identities
\begin{align*}
\left\{
\begin{aligned}
& e_i \cdot \ahom(Q) e_j
=
 \mathsf{B}_{Q}\left[ \ell_{e_i}, v(\cdot,Q,\mathbf{0},e_j)\right], 
\\ &
e_i \cdot \ahom_{\mu_{M,\xi}}(Q) e_j, 
=
\mathsf{B}_{Q,\mu_{M,\xi}}\left[ \ell_{e_i}, v_M(\cdot,Q,\mathbf{0},e_j)\right].
\end{aligned} 
\right. 
\end{align*}
Let $\Theta$ be the dynamical coupling of $\mu_{M,\xi}$ and $\mu=\mu_{\infty,\xi}$ introduced in Lemma~\ref{l.coupling.D}, and $L_0(R,\data)$ be the constant such that the estimates apply to all cubes $Q$ with $\size(Q) > L_0(R, \data)$,
By the linearity and coercivity of $\mathsf{B}$, we obtain
\begin{align*}
|e_i \cdot \ahom(Q) e_j - e_i \cdot \ahom_{\mu_{M,\xi}}(Q) e_j|
&
=
\left| \mathsf{B}_{Q} \left[ \ell_{e_i}, v(\cdot,Q,\mathbf{0},e_j) \right] - \B_{Q,\mu_M} \left[\ell_{e_i} , v_M(\cdot,Q,\mathbf{0},e_j) \right] \right| \\ &
\leq 
\E_\Theta \left[ \left\| \nabla  v(\cdot,Q,\mathbf{0},e_j)- \nabla v_M(\cdot,Q,\mathbf{0},e_j) \right\|_{\underline{L}^2(\mathcal{E}(Q))}^2 \right]^{\frac12}.
\end{align*}
Applying Lemma \ref{l.coupling.D} (with $V= \tilde V$ and $\f(e,\phi) = \a(e,\phi)e_j$), we obtain the existence of~$\beta(\data)>0$ and $C(\mathsf M, \data)<\infty$, such that for all $\size(Q) > L_0(R, \data)$, the right side of the previous display is bounded by 
\begin{equation*}
C \log^{\frac 12} (\size(Q))\E_\Theta  \left[ \left\| \a(\cdot,\phi) - \a_{\mu_{M,\xi}}(\cdot,\tilde\phi) \right\|_{\underline{L}^2(Q)}^2 \right]^{\frac12}\\
+ C (\size(Q))^{-\beta} \E_\Theta  \left[ \left\| \a(\cdot,\tilde{\phi}) \right\|^4 \right]^{\frac14}.
\end{equation*}
Combining the regularity assumptions for $\mathsf{V}''$ and the conclusion of Proposition \ref{p.coupling}, we obtain
\begin{equation*}
 |e_i \cdot \ahom(Q) e_j - e_i \cdot \ahom_{\mu_{M,\xi}}(Q) e_j| \leq C\size(Q)^{-\beta}.
\end{equation*}
Finally, we notice that the boundedness of $\ahom$ and $\ahom_{\mu_{M,\xi}}$ implies that the estimate applies to all cube if we allow~$C$ to depend on $R$. 
The estimate for~$\ahom_{*,\mu_M}(Q)$ is obtained by a nearly verbatim argument, which is omitted.
\end{proof}
  
  We conclude this section by giving a version of Proposition~\ref{p.convergence} for the finite-volume measures~$\mu_{L,\xi}$. 

\begin{proposition}
\label{p.convergence.muL}
There exist~$\beta(\data)\in \left( 0,\frac12\right]$ and~$C(R,\mathsf M, \data)<\infty$ such that, for every~$L,M\in\N$ with $L\leq M$, we have
\begin{equation}
\left| \nu_M (Q_L,\f,p) - \overline{\nu}(\f,p) \right| 
+
\left| \nu_M^*(Q_L,\f,q) - \overline{\nu}^*\!(\f,q) \right|
\leq
C\left( |p|+|q|+\mathsf{K}_0 \right)^2 L^{-\beta}. 
\end{equation}
\end{proposition}
\begin{proof}
Let $\mathcal{P}$ be a partition of~$Q_L$ consisting of triadic cubes which satisfies
\begin{equation}
\label{e.distpart}
 \size(Q) 
\leq 
\dist(Q,\partial Q_L) 
< 3 \size(Q). 
\end{equation}
This can be constructed, for instance, by taking the cubes to be as large as possible such that their predecessor triadic cube is contained in $Q_L$. A consequence of~\eqref{e.distpart} which we will need is that most of the volume of $Q_L$ is taken up by cubes in $\mathcal{P}$ which are relatively large: precisely, for every $n\in\N$,
\begin{equation}
\label{e.fillerup}
\sum_{Q \in \mathcal{P}} 
\frac{|Q|}{|Q_L|} \indc_{\{ \size(Q) \leq 3^n\}} 
\leq 
C3^nL^{-1}. 
\end{equation}
We then observe that, using the inequality $\ahom_{*,\mu_{M,\xi}}(Q_L) \leq \ahom_{\mu_{M,\xi}} (Q_L)$, Proposition~\ref{p.convergence}, Lemma \ref{l.acouple}  and~\eqref{e.fillerup},  we have for some $C = C(\mathsf M, \data)<\infty$,
\begin{align*}
\ahom_{*,\mu_{M,\xi}}(Q_L) 
\leq 
\ahom_{\mu_{M,\xi}}(Q_L)
\leq 
\sum_{Q\in \mathcal{P}} 
\frac{|Q|}{|Q_L|}
\ahom_{\mu_{M,\xi}}(Q)
&
\leq 
\sum_{Q\in \mathcal{P}} 
\frac{|Q|}{|Q_L|}
\left( 
\ahom(Q) +C \size(Q)^{-\beta}\right)
\\ &
\leq
\sum_{Q\in \mathcal{P}} 
\frac{|Q|}{|Q_L|}
\left( 
\ahom +C \size(Q)^{-\beta} \right)
\\ & 
\leq \ahom + CL^{-\beta}.
\end{align*}
Similarly, we have 
\begin{align*}
\ahom_{\mu_{M,\xi}}^{\,-1}(Q_L)
\leq
\ahom_{*,\mu_{M,\xi}}^{\,-1} (Q_L) \leq \ahom^{\,-1} + CL^{-\beta}. 
\end{align*}
The previous two displays imply that 
\begin{equation}
\left| \ahom_{\mu_{M,\xi}}(Q_L) - \ahom \right| 
+
\left| \ahom_{*,\mu_{M,\xi}}(Q_L) - \ahom \right|
\leq CL^{-\beta}.
\end{equation}
Arguing in the same way yields similar estimates on $\fhom_{\mu_{M,\xi}}$, $\chom_{\mu_{M,\xi}}$, $\fhom_{*,\mu_{M,\xi}}$ and $\chom_{*,\mu_{M,\xi}}$. This completes the proof. 
\end{proof}

\section{Identification and regularity of the surface tension}
\label{s.tension}

In this section, we show that Theorem~\ref{t.surfacetension} is a direct consequence of Proposition~\ref{p.convergence.muL} and the identity
\begin{equation} 
\label{e.D2sigmaL.ahomL}
D^2\sigma_L(\xi) = \ahom_{\mu_{L,\xi}}(Q_L),
\end{equation}
which (will be proved below and) states that the finite-volume surface tension is nothing other than the matrix corresponding to the quadratic form $p\mapsto \nu_L(Q_L,0,p)$. Another version of the identity~\eqref{e.D2sigmaL.ahomL}---with periodic rather than zero boundary conditions---was previously proved in~\cite[(3.78)]{DGI} and \cite[Appendix~A]{GOS}.

\smallskip

The identity~\eqref{e.D2sigmaL.ahomL} and Proposition~\ref{p.convergence.muL} immediately yield that 
\begin{equation} 
\label{e.D2sigma.converge}
\left| D^2\sigma_L(\xi) - \ahom(\xi) \right| \leq CL^{-\alpha}.
\end{equation}
We would like to send $L \to \infty$ in~\eqref{e.D2sigma.converge} to obtain that~$D^2\sigma$ exists and satisfies
\begin{equation} 
\label{e.identity}
D^2\sigma(\xi) = \ahom(\xi).
\end{equation}
To justify this, we just need to show that $D^2\sigma_L$ is continuous (with a modulus of continuity which may depend on~$L$, it does not matter). Indeed, if this can be shown, then~$\sigma$ is the pointwise limit of~$\sigma_L$ of $C^2$ functions with~$D^2\sigma_L$ converging locally uniformly to a continuous function $\ahom(\xi)$. This would imply that~$\sigma\in C^2$ and that~\eqref{e.identity} holds.

\smallskip

To summarize, we have two main assertions left to prove: (i)~the identity~\eqref{e.D2sigmaL.ahomL}, and (ii)~the continuity of $D^2\sigma_L$. 

\smallskip

We first present the proof of~\eqref{e.D2sigmaL.ahomL}, which  begins with the observation that, by a direct computation starting from the definition~\eqref{e.fst}, we have the following formula for $D^2\sigma_L$: for each $i,j \in \left\{ 1,...,d\right\} $,  
\begin{align}
\label{e.finitehessian}
\frac{\partial^2\sigma_L}{\partial \xi_{i }\partial \xi_{j }}(\xi) 
&
=
-\frac{1}{\left| Q_{L}\right|}
\cov_{\mu_{L,\xi}}
\left[ \sum_{x\in Q_{L}}\mathsf{V}^{\prime}\left( \nabla_{i }\phi(x)-\xi_i \right), \sum_{x\in Q_{L}}\mathsf{V}^{\prime }\left( \nabla_{j }\phi (x)-\xi_j\right) \right]  
\\ & \qquad
+\indc_{\{ i = j \}}\frac{1}{\left\vert
Q_L\right\vert }\left\langle \sum_{x \in Q_L} \mathsf{V}^{\prime \prime }\left( \nabla _i\phi(x)-\xi_i \right) \right\rangle _{\mu_{L,\xi}}. \notag
\end{align}
It is immediate from~\eqref{e.varcharvar} and polarization that the first term on the right side of~\eqref{e.finitehessian} can be written as
\begin{align}
\label{e.ST1}
\frac{1}{\left| Q_{L}\right|}
\cov_{\mu_{L,\xi}}
\left[ \sum_{x\in Q_{L}}\mathsf{V}^{\prime}\left( \nabla_{i }\phi(x)-\xi_i \right), \sum_{x\in Q_{L}}\mathsf{V}^{\prime }\left( \nabla_{j }\phi (x)-\xi_j\right) \right]  
= \mathsf{B}_{Q_L,\mu_{L,\xi}} \left[ u_i, u_j \right],  
\end{align}
where $u_i$ is the solution of the Dirichlet problem 
\begin{equation*}
\left\{ 
\begin{aligned}
& -\L_{\mu_{L,\xi}} u_i  + \nabla^* \a \nabla  u_i = \nabla^* \mathbf{f}_i
& \mbox{in} 
& \ Q_L^\circ \times \Omega_U, \\
& u_i = 0 
& \mbox{on} 
& \ \partial Q_L \times \Omega_U,
\end{aligned} 
\right. 
\end{equation*}
and $\mathbf{f}_i$ is the vector field 
\begin{equation}
\label{e.fsurf}
\mathbf{f}_i(e,\phi) = V''(\nabla \phi(e)- \xi_i) e_i = \a(e,\phi)e_i. 
\end{equation}
We see immediately from the definition of~$\nu_{L,\xi}(Q_L,e_i)$ that 
\begin{equation*} \label{}
u_i  = v_{L,\xi}(\cdot,Q_L,e_i) - \ell_{e_i}
\end{equation*}
where, as in the previous two sections, $v_{L,\xi}(\cdot,Q_L,\xi)$ denotes the minimizer in the definition of $\nu_{L,\xi}(Q_L,\xi)$. 
We have that 
\begin{equation*}
 \nu_L(Q_L,\xi) = \frac12 \mathsf{B}_{Q_L,\mu_{L,\xi}} \left[ v_{L,\xi}(\cdot,Q_L,\xi), v_{L,\xi}(\cdot,Q_L,\xi) \right],
\end{equation*}
that is,
\begin{equation*}
\ahom_{\mu_{L,\xi},ij} (Q_L) = \mathsf{B}_{Q_L,\mu_{L,\xi}} \left[ v_{L,\xi}(\cdot,Q_L,e_i), v_{L,\xi}(\cdot,Q_L,e_j) \right].
\end{equation*}
By the equation for $v_{L,\xi}(\cdot,Q_L,\xi)$, we have that 
\begin{equation*}
\mathsf{B}_{Q_L,\mu_{L,\xi}} \left[ v_{L,\xi}(\cdot,Q_L,\xi), u_i \right] 
=
0,
\end{equation*}
and thus
\begin{align*} \label{}
\mathsf{B}_{Q_L,\mu_{L,\xi}} \left[ \ell_{e_i}, \ell_{e_j} \right] 
& 
= \mathsf{B}_{Q_L,\mu_{L,\xi}} \left[ v_{L,\xi}(\cdot,Q_L,e_i) - u_i, v_{L,\xi}(\cdot,Q_L,e_j) - u_j \right] 
\\ & 
= 
\mathsf{B}_{Q_L,\mu_{L,\xi}} \left[ v_{L,\xi}(\cdot,Q_L,e_i) , v_{L,\xi}(\cdot,Q_L,e_j)  \right] 
+
\mathsf{B}_{Q_L,\mu_{L,\xi}} \left[ u_i, u_j \right].
\end{align*}
Finally, we observe that 
\begin{equation*} \label{}
\mathsf{B}_{Q_L,\mu_{L,\xi}}  \left[ \ell_{e_i}, \ell_{e_j} \right] 
=\indc_{\{ i = j \}}\frac{1}{\left\vert
Q_L\right\vert }\left\langle \sum_{x \in Q_L} \mathsf{V}^{\prime \prime }\left( \nabla _i\phi(x)- \xi_i \right) \right\rangle _{\mu_{L,\xi}}.
\end{equation*}
Combining the above displays, we obtain~\eqref{e.D2sigmaL.ahomL}.

\smallskip

We next give the promised estimate for the regularity of~$D^2\sigma_L$.

\begin{lemma}
\label{l.continuity.D2sigmaL}
Let $R\in [1,\infty)$. There exist $\beta (\data)>0$, $C(\mathsf M, R,\data)<\infty$ and, for every $\theta>0$, a constant~$L_0(\theta,R,\data)$ such that, for every $\xi,\xi'\in B_R$ and $L \geq L_0$,
\begin{equation} 
\label{e.sigmaL.almostholder}
\left| D^2\sigma_L(\xi) - D^2\sigma_L(\xi') \right| 
\leq 
C\left( |\xi -\xi'| + \theta \right)^\beta. 
\end{equation}
\end{lemma}
\begin{proof}[{Proof of Lemma~\ref{l.continuity.D2sigmaL}}]
We first combine Proposition \ref{p.convergence.muL} and the identification \eqref{e.D2sigmaL.ahomL}, which yields, for every~$\xi \in \R^d$, 
\begin{equation}
\left|D^2\sigma_L(\xi) - \ahom_{\mu_{2L,\xi}}(Q_L) \right| 
\leq 
CL^{-\beta}.
\end{equation}
By definition,
\begin{align*}
\lefteqn{
\left| 
\ahom_{\mu_{2L,\xi,ij}}(Q_L) - \ahom_{\mu_{2L,\xi',ij}}(Q_L) 
\right| 
} 
\\ & \qquad 
=
\left|   \mathsf{B}_{Q_L,\mu_{2L,\xi}} 
\left[ u_i + \ell_{e_i}, u_j + \ell_{e_j} \right]
- \mathsf{B}_{Q_L,\mu_{2L,\xi'}} 
\left[ u_i + \ell_{e_i}, u_j + \ell_{e_j} \right]
\right|.
\end{align*}
We also have 
\begin{align*}
& \left|   \mathsf{B}_{Q_L,\mu_{2L,\xi}} \left[  \ell_{e_i},  \ell_{e_j} \right]
- \mathsf{B}_{Q_L,\mu_{2L,\xi'}} \left[ \ell_{e_i}, \ell_{e_j} \right]
\right| 
\\ & \qquad 
= \frac{1}{|Q_L|} \indc_{\{ i=j \}} \left|  \sum_{e\in Q_L} \left\langle \mathsf V''\left(\nabla \phi (e)- \xi_i \right)\right\rangle_{\mu_{2L,\xi}} -  \left\langle \mathsf V''\left(\nabla \phi(e) - \xi'_i \right)\right\rangle_{\mu_{2L,\xi'}}
\right|.
\end{align*}
We apply the regularity assumption of $\mathsf V$ and Proposition \ref{p.coupling}, using a coupling~$\Theta$ of~$\mu_{2L,\xi}$ and~$\mu_{2L,\xi'}$, to obtain the existence of~$\beta(\data)>0$ such that the right side of the previous display is at most
\begin{equation}
\label{e.linear}
 \mathsf M |\xi -\xi'|^\gamma + \mathsf M \left[ \E_\Theta \left\| (\nabla \phi_{2L,\xi} - \nabla \phi_{2L,\xi'}) \right\|_{L^\infty(Q_L)} \right]^{\gamma/2}
\leq 
C(|\xi -\xi'|^\gamma + L^{-\beta} + L^{1-\beta} |\xi- \xi'|).
\end{equation}
We next use the representation 
\begin{equation*}
 \mathsf{B}_{Q_L,\mu_{2L,\xi}} \left[ u_i , u_j  \right]
 = \frac{1}{|Q_L|}\sum_{e\in Q_L}   \left\langle \f(e, \nabla \phi_{2L,\xi}) \nabla u (e,  \phi_{2L,\xi}) \right\rangle_{\mu_{2L,\xi}}
 \end{equation*}
and apply Proposition~\ref{p.coupling} again as well as Lemma \ref{l.coupling.D} (with $\tilde{V} (\cdot)= V(\tilde\xi -\xi +\cdot)$ and~$\f$ defined by~\eqref{e.fsurf}) and the regularity assumption~\eqref{e.V.C2gamma} to obtain
 \begin{align*}
& \left| \mathsf{B}_{Q_L,\mu_{2L,\xi}} \left[ u_i(\phi_{2L,\xi}) , u_j (\phi_{2L,\xi}) \right]
 -  \mathsf{B}_{Q_L,\mu_{2L,\xi'}} \left[ u_i(\phi_{2L,\xi'}) , u_j (\phi_{2L,\xi'}) \right] \right|  
\\ & \quad 
\leq 
 \left\| \mathsf{V}''\left(\xi + \nabla \phi_{2L,\xi} \right) - \mathsf{V}''\left(\xi'+ \nabla \phi_{2L,\xi'} \right) \right\|_{\underline{L}^2(Q_L,\Theta)}
 \left\| \nabla u(\phi_{2L,\xi}) - \nabla u(\phi_{2L,\xi'})\right\|_{\underline{L}^2(Q_L,\Theta)} 
 \\ & \quad 
 \leq C (\log L)\left(|\xi -\xi'|^\gamma +  L^{-\beta}+ L^{1-\beta}|\xi - \xi'|\right).
 \end{align*}
Combining this with~\eqref{e.linear} and the previous displays above, we obtain
\begin{equation} 
\label{e.continuity.D2sigmaL}
\left| D^2\sigma_L(\xi) - D^2\sigma_L(\xi') \right| 
\leq 
C(\log L) |\xi- \xi'|^\beta + CL^{1-\beta} |\xi- \xi'| + CL^{-\beta}.
\end{equation}
By~\eqref{e.D2sigma.converge} and~\eqref{e.continuity.D2sigmaL}, for every $L,M\in\N$ with~$L \leq M$, 
\begin{align*} \label{}
\lefteqn{
\left| D^2\sigma_M(\xi) - D^2\sigma_M(\xi') \right| 
} \quad & 
\\ & 
\leq 
\left| D^2\sigma_L(\xi) - D^2\sigma_L(\xi') \right| 
+ \left| D^2\sigma_L(\xi) - D^2\sigma_M(\xi) \right| 
+ \left| D^2\sigma_L(\xi') - D^2\sigma_M(\xi') \right| 
\\ & 
\leq 
C(\log L) \left|\xi- \xi'\right|^\beta + CL^{1-\beta} \left|\xi- \xi'\right|+ C L^{-\beta}.
\end{align*}
To conclude the proof, we simply observe that, for every~$\theta>0$, there exists $L_0$ sufficiently large that, for every $M\geq L_0$, 
\begin{equation}
\min_{L \in \{1,\ldots,M\}} 
\left(C(\log L) \left|\xi- \xi'\right|^\beta + CL^{1-\beta} \left|\xi- \xi'\right|+ C L^{-\beta}  \right) 
\leq 
C \left(  \left| \xi-\xi' \right| + \theta \right)^\beta,
\end{equation}
where, as usual, the exponent~$\beta$ on the right side is smaller than on the left side.
\end{proof}

The previous lemma does not quite imply that $D^2\sigma_L \in C^{0,\beta}$, since the $L$ in the estimate~\eqref{e.sigmaL.almostholder} must be large compared to the parameter~$\theta$. This is in fact an artifact of the way we have written the coupling lemmas to be flexible to finite volume measures of different sizes. When the cubes are the same size, the error term $L^{-\beta}$ in the coupling lemma can be removed, so the lemma can actually be improved to obtain that $D^2\sigma_L\in C^{0,\beta}$. However, this point does not matter for our purposes, because combining the estimate~\eqref{e.sigmaL.almostholder} with~\eqref{e.D2sigma.converge} yields, after sending $L\to \infty$, that~$\ahom(\xi)$ is continuous and, for each $R\in [1,\infty)$, the existence of $C(R,\data)<\infty$ such that, for every $\xi,\xi'\in B_R$,  
\begin{equation}
\left| \ahom(\xi) - \ahom(\xi') \right|
\leq 
C\left| \xi-\xi'\right|^\beta. 
\end{equation}
That is, $\ahom$ is locally H\"older continuous. 
Therefore~\eqref{e.identity} holds and the proof of Theorem~\ref{t.surfacetension} is complete.

\appendix

\section{Auxiliary estimates}
\label{s.aux}

In this appendix we collect some functional inequalities and parabolic decay estimates which are used in the paper. 

\smallskip

We begin with estimates on the~$L^2$ decay of  solution of a parabolic initial-value problems posed in bounded domains. These estimates are essentially classical, but we could not find a precise reference which applies to our particular framework.

\begin{lemma}[Decay estimate, Cauchy-Dirichlet problem]
\label{l.CD.decay}
Fix $T \in (0,\infty)$, $L\in\N$ with $L\geq 2$ and $0<\lambda\leq \Lambda <\infty$. Suppose that $\a:(0,T) \times \mathcal{E}(Q_L) \to \R$ satisfies $\lambda \leq \a \leq \Lambda$. Suppose that $G\in L^2(Q_L)$, $\f \in L^2(\mathcal{E}(Q_L))$, $\h\in L^2((0,T);L^2(\mathcal{E}(Q_L)))$ and let $u\in C^1((0,T);L^2(Q_L))$ satisfy the initial-boundary value problem 
\begin{equation}
\label{e.CD.decay.bvp}
\left\{
\begin{aligned}
& \partial_t u + \nabla^*\a \nabla u = \nabla^*\h & \mbox{in} & \ (0,\infty) \times Q_L^\circ,\\
& u = 0 & \mbox{on} & \ (0,\infty) \times \partial Q_L,\\
& u = \nabla^*\mathbf{f} & \mbox{on} & \ \{ 0\} \times Q_L^\circ.
\end{aligned}
\right.
\end{equation}
Then there exists $C(\data)<\infty$ such that, for all $t\in (0,\infty)$,
\begin{align}
\label{e.CD.decay}
\left\| u(t,\cdot) \right\|_{L^2(Q_L)}^2
&
\leq 
C \left\| \f \right\|_{L^2(Q_L)}^2 \left(1+ t\right)^{-1} \exp\left( -\frac{t}{CL^2} \right)
\\ & \quad \notag
+ C\int_0^t \left\| \h(t-s,\cdot) \right\|_{L^2(Q_L)}^2 \exp\left( -\frac{s}{CL^2} \right) \,ds.
\end{align}
\end{lemma}
\begin{proof}
We first prove~\eqref{e.CD.decay} in the case~$\h=0$. Fix~$s\in (0,\infty)$ and observe that
\begin{equation*}
\partial_t \sum_{x\in Q_L} u(2s-t,x)u(t,x) =0 \quad \mbox{for every} \ t\in (0,2s). 
\end{equation*}
In particular, for every $s\in (0,\infty)$, 
\begin{equation}
\sum_{x\in Q_L} u^2(s,x) = \sum_{e\in \mathcal{E}(Q_L)} \nabla u(2s,e) \mathbf{f}(e)
\end{equation}
Fix $t\in (1,L^2]$. By the Cauchy-Schwarz inequality and the Caccioppoli inequality, we get 
\begin{align*}
\fint_{t}^{2t} \sum_{x\in Q_L} u^2(s,x)\,ds 
&
=
\fint_{t}^{2t}  \sum_{e\in \mathcal{E}(Q_L)} \nabla u(2s,e) \mathbf{f}(e)\,ds
\\ & 
\leq 
C\left(\fint_{2t}^{4t} \sum_{e\in \mathcal{E}(Q_L)} \left(\nabla u(s,e) \right)^2 \,ds \right)^{\frac12} 
\left( \sum_{e\in \mathcal{E}(Q_L)} \left(\mathbf{f} (e) \right)^2 \right)^{\frac12}
\\ & 
\leq 
Ct^{-\frac12}\left\| u(2t,\cdot) \right\|_{L^2(Q_L)} \left\| \f \right\|_{L^2(Q_L)},
\end{align*}
where in the last line we used the fact, which follows from integrating~\eqref{e.diss}, that for every $t_1<t_2$, 
\begin{equation}
\label{e.stabgrad}
\int_{t_1}^{t_2}  \sum_{e\in \mathcal{E}(Q_L)} \left(\nabla u(s,e) \right)^2 \,ds
\leq 
C \left( 
\sum_{x\in Q_L} u^2(t_1,x)
-
\sum_{x\in Q_L} u^2(t_2,x)
\right).
\end{equation}
As $s\mapsto \sum_{x\in Q_L} u^2(s,e)$ is evidently monotone decreasing, we deduce that 
\begin{equation*}
\left\| u(2t,\cdot) \right\|_{L^2(Q_L)}^2
\leq 
\fint_{t}^{2t} \sum_{x\in Q_L} u^2(s,x)\,ds 
\leq
Ct^{-\frac12}\left\| u(2t,\cdot) \right\|_{L^2(Q_L)} \left\| \f \right\|_{L^2(Q_L)}.
\end{equation*}
We have proven that, for every $t\in (2,2L^2]$, 
\begin{equation*}
\left\| u (t,\cdot) \right\|_{L^2(Q_L)}
\leq
Ct^{-\frac12} \left\| \f \right\|_{L^2(Q_L)}.
\end{equation*}
This gives the estimate~\eqref{e.CD.decay} for~$t\in (2,2L^2]$. For $t\in (0,2]$, we can use the result of Step~1 and the fact that $| \nabla^* \f(x) | \leq C \sum_{x\in e} \left| \f(e) \right|$, which implies that $\| \nabla^*\f\|_{L^2(Q_L)} \leq C \left\| \f \right\|_{L^2(Q_L)}$, to immediately obtain~\eqref{e.CD.decay}. For $t\in (2L^2,\infty)$, we apply the bound in Step~1, starting from time~$L^2$, to obtain
\begin{align*}
\sum_{x\in Q_L} u^2(t,x) 
&
\leq C\exp\left( -c\frac{t-L^2}{L^2} \right) \sum_{x\in Q_L} u^2(L^2,x)
\\ &
\leq CL^{-2} \exp\left( -c\frac{t}{L^2} \right)  
\left\| \f \right\|_{L^2(Q_L)}^2
\leq 
Ct^{-1} \exp\left( -c\frac{t}{L^2} \right)  
\left\| \f \right\|_{L^2(Q_L)}^2,
\end{align*}
as desired. 

\smallskip

We next consider the case~$\f=0$. We find that 

\begin{align}
\label{e.diss}
\partial_t \sum_{x\in Q_L} u^2(t,x) 
&
=
-2 \sum_{e\in \mathcal{E}(Q_L)} 
\nabla u(t,e) \left( \a(t,e) (\nabla u(t,e)) - \h(t,e) \right)
\\ &  \notag
\leq
-2\lambda \sum_{e\in \mathcal{E}(Q_L)} 
(\nabla u(t,e))^2
+
C\sum_{e\in \mathcal{E}(Q_L)} \left|\nabla u(t,e) \right| \left| \h(t,e) \right| 
\\ & \notag
\leq 
-cL^{-2}
\sum_{x\in Q_L} u^2(t,x)
+ 
C \sum_{e\in \mathcal{E}(Q_L)} (\h(t,e))^2.
\end{align}
Integration of this differential inequality yields 
\begin{equation}
\sum_{x\in Q_L} u^2(t,x)
\leq 
C\int_0^t \left\| \h(t-s,\cdot) \right\|_{L^2(Q_L)}^2 \exp\left( -\frac{s}{CL^2} \right) \,ds.
\end{equation}
This completes the proof of~\eqref{e.CD.decay} in the case~$\f=0$. 

Combining the two cases above yields the estimate~\eqref{e.CD.decay} in full generality. 
\end{proof}

\begin{lemma}[Decay estimate, Cauchy-Neumann problem]
\label{l.CN.decay}
Let $L\in\N$ with~$L\geq 2$ and~$0<\lambda\leq \Lambda <\infty$. Suppose that~$\a:(0,\infty) \times \mathcal{E}(Q_L) \to \R$ satisfies~$\lambda \leq \a \leq \Lambda$. Let~$u$ be the solution of the initial-boundary value problem 
\begin{equation}
\label{e.CN.decay.bvp}
\left\{
\begin{aligned}
& \partial_t w + \nabla^*\a \nabla w =\nabla^* \h & \mbox{in} & \ (0,\infty) \times U^\circ,\\
& \a \nabla w  = \h & \mbox{on} & \ (0,\infty) \times \partial U,\\
& w = \nabla^*\f & \mbox{on} & \ \{ 0\} \times U^\circ, \\
& \a \nabla w -\f = \h  & \mbox{on} & \ \{ 0\} \times\partial \mathcal{E}(U).
\end{aligned}
\right.
\end{equation}
Then there exists $C(\data)<\infty$ such that, for all $t\in (0,\infty)$,
\begin{align}
\label{e.CN.decay}
\left\| u(t,\cdot) \right\|_{L^2(Q_L)}^2
&
\leq 
C \left( \left\| \f \right\|_{L^2(Q_L)}^2+\left\| \h \right\|_{L^2(\partial Q_L)}^2 \right) \left(1+ t\right)^{-1} \exp\left( -\frac{t}{CL^2} \right)
\\ & \quad \notag
+ C\int_0^t \left\| \h(t-s,\cdot) \right\|_{L^2(Q_L)}^2 \exp\left( -\frac{s}{CL^2} \right) \,ds.
\end{align}
\end{lemma}
\begin{proof}
The proof is completely analogous to the one of Lemma~\ref{l.CD.decay}: the main differences being that we use Poincar\'e inequality for mean-zero functions rather than for functions which vanish on the boundary, and that the validity of the summations by parts are due to the Neumann (zero flux) boundary condition rather than the zero Dirichlet condition. 
\end{proof}

We state a discrete version of the multiscale Poincare inequality, which provides an estimate of the $H^{-1}\left( \cu_n \right) $ norm of a function in terms of its spatial averages in triadic subcubes.

\begin{proposition}
\label{p.MP}({Multiscale Poincare inequality}) Fix $m\in\N$ and denote $\mathcal{Z}_{n}=3^{n}\mathbb{Z}^{2}\cap \cu _{m}$. 
Then, for any $u\in L^2(\cu_m)$,
\begin{equation*}
\left\| u - (u)_{\cu_m} \right\|_{\underline{L}^2( \cu_{m}) }
\leq
C \left\| \nabla u \right\|_{\underline{L}^{2}(\cu_{m})}
+
C\sum_{n=0}^{m-1}3^{n}\left( 
\frac1{\left|\mathcal{Z}_{n}\right|}
\sum_{y\in \mathcal{Z}_{n}}
 \left| \left( \nabla u \right) _{y+\cu _{n}}\right|^{2}  \right)^{1/2}.
\end{equation*}
\end{proposition}
Proposition~\ref{p.MP} is a discrete analogue of~\cite[Corollary 1.14]{AKMbook} and, since its proof is essentially identical to that of the latter, we omit it.

\subsection*{Acknowledgments}
We thank Tom Spencer for many insightful discussions that inspired the project,
and Paul Dario for very helpful discussions as well as his comments on previous drafts of this manuscript. 
SA was partially supported by the National Science Foundation through grant DMS-1700329. WW is partially supported by the EPSRC grant 	
EP/T00472X/1. Both authors were partially supported by a grant from the NYU-PSL Global Alliance.

\small
\bibliographystyle{abbrv}
\bibliography{gradphi}

\end{document}